\documentclass[11pt]{article}

\usepackage[utf8]{inputenc}

\usepackage[ruled,vlined,linesnumbered,noend]{algorithm2e}
\newcommand {\ignore} [1] {}
\usepackage{amsthm}
\usepackage{thmtools}
\usepackage{amsmath}
\usepackage{amssymb}
\usepackage{enumerate}
\usepackage{graphicx}
\usepackage[margin=1in]{geometry}
\usepackage{dsfont}
\usepackage{booktabs} 
\usepackage{tabularx}
\usepackage{multirow}
\usepackage[shortlabels]{enumitem}
\usepackage{tablefootnote}

\usepackage[dvipsnames]{xcolor}
\definecolor{ForestGreen}{rgb}{0.1333,0.5451,0.1333}
\definecolor{DarkRed}{rgb}{0.65,0,0}

\usepackage{hyperref}
\hypersetup{
    colorlinks=true,
    linkcolor=DarkRed,
    filecolor=magenta,      
    urlcolor=cyan,
    citecolor=ForestGreen,
    linktocpage=true,
    pagebackref=true,
    bookmarks=true,
    bookmarksopen=true,
    bookmarksnumbered=true
}

\usepackage{cleveref}

\usepackage{framed}
\usepackage[framemethod=tikz]{mdframed}
\usepackage{tikz-cd}
\usetikzlibrary{positioning}
\usetikzlibrary{decorations.pathmorphing}

\usepackage[normalem]{ulem}

\newenvironment{wrapper}[1]
{
	\begin{center}
		\begin{minipage}{\linewidth}
			\begin{mdframed}[hidealllines=true, backgroundcolor=gray!20, leftmargin=0cm,innerleftmargin=0.4cm,innerrightmargin=0.4cm,innertopmargin=0.4cm,innerbottommargin=0.4cm,roundcorner=3pt]
				#1}
			{\end{mdframed}
		\end{minipage}
	\end{center}
}

\newcommand{\miss}{\mathsf{miss}}

\newcommand{\Amplify}{\textnormal{\texttt{Sparsify-Types}}}
\newcommand{\Extend}{\textnormal{\texttt{Extend-Coloring}}}
\newcommand{\Small}{\textnormal{\texttt{Color-Small}}}

\newcommand{\Compute}{\textnormal{\texttt{Compute-Paths}}}
\newcommand{\ModifyB}{\textnormal{\texttt{Modify-Types}}}
\newcommand{\Flip}{\textnormal{\texttt{Flip-Path}}}

\newcommand{\Depth}{\textnormal{\texttt{depth}}}
\DeclareMathOperator*{\U}{\mathcal U}
\DeclareMathOperator*{\C}{\mathcal C}

\DeclareMathOperator*{\f}{\boldsymbol{f}}

\DeclareMathAlphabet{\mathmybb}{U}{bbold}{m}{n}

\newtheorem{theorem}{Theorem}[section]
\newtheorem{lemma}[theorem]{Lemma}

\newtheorem{corollary}[theorem]{Corollary}

\newtheorem{definition}[theorem]{Definition}

\newtheorem{invariant}[theorem]{Invariant}

\newtheorem{assumption}[theorem]{Assumption}
\newtheorem{observation}[theorem]{Observation}
\newtheorem{claim}[theorem]{Claim}

\DeclareMathOperator*{\poly}{poly}
\DeclareMathOperator*{\polylog}{polylog}

\title{
Vizing’s Theorem in Deterministic Almost-Linear Time
}

\author{}

\author{
Sepehr Assadi\thanks{University of Waterloo, \texttt{sepehr@assadi.info}}
\and
Soheil Behnezhad\thanks{Northeastern University, \texttt{s.behnezhad@northeastern.edu}}
\and
Sayan Bhattacharya\thanks{University of Warwick, \texttt{s.bhattacharya@warwick.ac.uk}}
\and 
Mart\'in Costa\thanks{University of Warwick, \texttt{martin.costa@warwick.ac.uk}}
\and 
Shay Solomon\thanks{Tel Aviv University, \texttt{solo.shay@gmail.com}} \and 
Tianyi Zhang\thanks{Nanjing University, \texttt{tianyiz25@nju.edu.cn} \smallskip}
}

\date{
}

\begin{document}

\maketitle

\pagenumbering{gobble}

\begin{abstract}
Vizing’s theorem states that any $n$-vertex $m$-edge graph of maximum degree $\Delta$ can be edge colored using at most $\Delta + 1$ different colors.
Vizing's original proof is easily translated into a deterministic $O(mn)$ time algorithm.
This deterministic time bound was subsequently improved to $\tilde O(m\sqrt{n})$ time, independently by [Arjomandi, 1982] and by [Gabow et al., 1985].\footnote{We use the $\tilde O$ notation to hide $\polylog(n)$ factors.}

\smallskip

A series of recent papers improved the time bound of $\tilde O(m\sqrt{n})$ using randomization, culminating in the randomized near-linear time $(\Delta+1)$-coloring algorithm by 
[Assadi, Behnezhad, Bhattacharya, Costa, Solomon, and Zhang, 2025]. At the heart of all of these recent improvements, there is some form of a sublinear time algorithm. Unfortunately, sublinear time algorithms as a whole almost always require randomization. This raises a natural question: can the deterministic time complexity of the problem be reduced below the $\tilde O(m\sqrt{n})$ barrier?

\smallskip

In this paper, we answer this question in the affirmative. We present a deterministic almost-linear time $(\Delta+1)$-coloring algorithm, namely, an algorithm running in $m \cdot 2^{O(\sqrt{\log \Delta})} \cdot \log n = m^{1+o(1)}$ time. Our main technical contribution is to entirely forego sublinear time algorithms. We do so by presenting a new deterministic color-type sparsification approach that runs in almost-linear (instead of sublinear) time, but can be used to color a much larger set of edges.
\end{abstract}

\clearpage

\setcounter{tocdepth}{3}
\tableofcontents
\clearpage
\pagenumbering{arabic}

\clearpage
\setcounter{page}{1}

\section{Introduction}

Let $G = (V, E)$ be a simple and undirected $n$-vertex $m$-edge graph with maximum degree $\Delta$. A {\em proper $\mu$-edge coloring} (shortly, $\mu$-coloring) $\chi : E \rightarrow \{1, 2, \ldots, \mu\}$ of $G$,
for an integer parameter $\mu \in \mathbb{N}^+$, is an assignment of colors $\chi(e)$ to each edge $e \in E$, where no two adjacent edges receive the same color.
Since the colors of all edges incident on each vertex must be distinct, any proper edge coloring requires at least $\Delta$ different colors. While $\Delta$ colors always suffice in bipartite graphs, some graphs (e.g.,  the triangle) require more than $\Delta$ colors. Vizing's Theorem asserts that $\Delta+1$ colors are always sufficient~\cite{Vizing}. Furthermore, even for small values of $\Delta$, the problem of determining whether $\Delta$ colors suffice is NP-hard~\cite{holyer1981np}.
The focus of this work is on the {\bf deterministic} time complexity of $(\Delta+1)$-coloring in {\em general} graphs. 

Vizing's original proof is easily translated to a deterministic $O(mn)$ time algorithm. This deterministic time bound was improved to $\tilde{O}(m\sqrt{n})$ in the 1980s, independently by \cite{arjomandi1982efficient} and \cite{gabow1985algorithms} (more recently~\cite{sinnamon2019fast} proved that the algorithm of~\cite{gabow1985algorithms} in fact achieves a clean $O(m\sqrt{n})$ time).
In {\em bipartite} graphs, deterministic near-linear time algorithms (even for $\Delta$-coloring) are known since the 80s  \cite{cole1982edge,combinatorica/ColeOS01,alon2003simple,goel2010perfect}. 
However, in general graphs, the problem is much more challenging. 
Allowing {\em randomization}, a sequence of recent papers \cite{sinnamon2019fast,
BhattacharyaCCSZ24,Assadi24,BhattacharyaCSZ24,ABBC2025}  improved the longstanding time bound of $\tilde O(m\sqrt{n})$, culminating in the randomized near-linear time $(\Delta+1)$-coloring algorithm of~\cite{ABBC2025} (we discuss these papers in more detail below).
Nonetheless, the following fundamental question remains open: 

\begin{wrapper}
Can the {\bf deterministic} time complexity of $(\Delta+1)$-coloring in {\bf general} graphs be reduced below the bound of $\tilde O(m\sqrt{n})$ barrier?
\end{wrapper}

In this paper, we present a deterministic almost-linear time $(\Delta+1)$-coloring algorithm.

\begin{theorem} \label{thm:main:1}
    There is a deterministic algorithm that, for any  $n$-vertex $m$-edge graph $G = (V, E)$ with maximum degree $\Delta$, computes a $(\Delta + 1)$-coloring of $G$ in $m \cdot 2^{O(\sqrt{\log \Delta})} \cdot \log n = m^{1+o(1)}$  time.
\end{theorem}

All recent randomized edge coloring algorithms that break the $\tilde{O}(m\sqrt{n})$ bound rely crucially on some form of a {\em sublinear time} algorithm or subroutine. For example, the algorithms in~\cite{BhattacharyaCCSZ24,BhattacharyaCSZ24,ABBC2025} use randomization to find one or many color-alternating paths in $\tilde{O}(n)$ time, and the algorithm of \cite{Assadi24} iteratively finds large matchings in $\tilde{O}(n)$ time. 
While sublinear time algorithms are a very powerful tool for designing randomized algorithms, they are almost always impossible to de-randomize. 

Our algorithm in~\Cref{thm:main:1} builds on the framework of \cite{ABBC2025} and de-randomizes their main ``color type reduction'' technique, by crucially replacing their sublinear time subroutine with a new ``gradual sparsification approach'' that we introduce in this paper. At a high level, this algorithm is considerably slower and runs in almost-linear time---compared to the sublinear (namely, $O(m/\Delta)$) time guarantee of its counterpart in~\cite{ABBC2025}---but can instead apply the required ``type reduction'' to an almost constant fraction of the graph---as opposed to only $\Theta(1/\Delta)$ fraction in~\cite{ABBC2025}. We elaborate more on our techniques as well as a technical comparison with prior work in~\Cref{sec:tech overview}.

\subsection{Previous Work}
\subsubsection{Recent Randomized Algorithms}

As mentioned, the state-of-the-art deterministic runtime for $(\Delta+1)$-coloring is $O(m\sqrt{n})$ due to~\cite{sinnamon2019fast}, which provides a more careful 
analysis of the $\tilde{O}(m\sqrt{n})$ time algorithm of~\cite{gabow1985algorithms} (see also~\cite{arjomandi1982efficient}). Moreover,~\cite{sinnamon2019fast} also presented a much simpler randomized algorithm with the same $O(m\sqrt{n})$ runtime that succeeds with
(exponentially) high probability. The first polynomial improvement over this longstanding $\tilde{O}(m\sqrt{n})$ time bound was obtained in two independent works \cite{Assadi24} and
\cite{BhattacharyaCCSZ24}, presenting two randomized algorithms with the incomparable time bounds of $\tilde{O}(n^2)$ 
and $\tilde{O}(mn^{1/3})$, respectively. 
In a follow-up work, the $\tilde{O}(mn^{1/3})$ randomized time bound of~\cite{BhattacharyaCCSZ24} was improved to $\tilde{O}(mn^{1/4})$ time in \cite{BhattacharyaCSZ24}. 
Finally, a randomized near-linear time algorithm was presented in \cite{ABBC2025}; more specifically, the algorithm of \cite{ABBC2025} achieves a time bound of $O(m \log \Delta)$ with high probability.

\subsubsection{Other Related Work}

There is a growing body of work on fast algorithms {for edge coloring} that use more than $\Delta + 1$ colors. A simple \emph{randomized} near-linear time algorithm that uses $\Delta+\tilde{O}(\sqrt{\Delta})$ colors
was given already in the 80s  \cite{karloff1987efficient}. More recently,~\cite{Assadi24} improved these bounds to a $\Delta + O(\log{n})$ coloring in $O(m\log{\Delta})$ expected time. There are also recent algorithms that run in near-linear (and sometimes linear) time for $(1+\epsilon)\Delta$-coloring~\cite{duan2019dynamic,BhattacharyaCPS24,elkin2024deterministic,bernshteyn2024linear,dhawan2024simple}, for a constant $\epsilon \in (0,1)$; while the algorithm of  \cite{elkin2024deterministic} is deterministic, all the others are randomized. 

The edge coloring problem has also been studied extensively in restricted graph classes. As mentioned, in bipartite graphs, a $\Delta$-coloring can be computed deterministically in $\tilde{O}(m)$ time~\cite{cole1982edge,combinatorica/ColeOS01,alon2003simple,goel2010perfect}; in particular, a deterministic time bound of $O(m \log \Delta)$ was achieved in \cite{combinatorica/ColeOS01}. In bounded degree graphs, one can deterministically compute a $(\Delta+1)$-coloring in $\tilde{O}(m \Delta)$ time \cite{gabow1985algorithms}, and a randomized algorithm generalizing this result for
bounded arboricity graphs was given in~\cite{BhattacharyaCPS24b}; refer also to \cite{BhattacharyaCPS24c,ChristiansenRV24,Kowalik24} for additional recent deterministic and randomized algorithms on edge coloring in bounded arboricity graphs.
There are also works on edge coloring in subfamilies of bounded arboricity graphs, and in
particular in bounded tree-width graphs,
planar graphs  and bounded genus graphs~\cite{chrobak1989fast,chrobak1990improved,cole2008new}. 

Furthermore, edge coloring has been studied extensively and intensively in various computational models over the past few years, including the dynamic setting~\cite{BarenboimM17,BhattacharyaCHN18,duan2019dynamic,Christiansen23,BhattacharyaCPS24,Christiansen24}, distributed computing ~\cite{panconesi2001some,elkin20142delta,fischer2017deterministic,ghaffari2018deterministic,GrebikP20,balliu2022distributed,ChangHLPU20,Bernshteyn22,Christiansen23,Davies23},
online algorithms ~\cite{CohenPW19,BhattacharyaGW21,SaberiW21,KulkarniLSST22,BilkstadSVW24,BlikstadOnline2025,dudeja2024randomizedgreedyonlineedge}, streaming algorithms ~\cite{BehnezhadDHKS19,behnezhad2023streaming,chechik2023streaming,ghosh2023low}, and sampling algorithms \cite{abdolazimi2022matrix,wang2024sampling,chen2025decay}.

\subsection{Organization} 
In Section \ref{sec:tech overview} we provide a detailed technical overview of our work, which also includes comparison with previous work. In \Cref{sec:prelim} we give the necessary notation and preliminaries. In \Cref{sec:main} we demonstrate how our main result relies on three subroutines; we state three lemmas that describe the behaviour of these subroutines and use them to prove \Cref{thm:main:1}. The rest of the paper is devoted to presenting these subroutines and formally analyzing them by proving these lemmas.

\section{Technical Overview}\label{sec:tech overview}

To highlight the main ideas behind our approach, in this section we instantiate our algorithm on bipartite graphs. Let $G = (V, E)$ be the input bipartite graph with maximum degree $\Delta$. Consider any {\bf partial $\Delta$-coloring} $\chi : E \longrightarrow [\Delta] \cup \{ \bot \}$ of $G$, where $\{ e \in E : \chi(e) = \bot \}$ is the set of {\bf uncolored} edges. We say that $\chi$ is {\bf valid} iff no two edges incident on the same vertex receive the same color (from $[\Delta]$) under $\chi$. Unless explicitly stated otherwise, we will only consider valid partial colorings. Consider any subset $S \subseteq \{ e \in E : \chi(e) = \bot \}$ of uncolored edges. The phrase {\bf extending the coloring $\chi$ to $S$} refers to the following operation: Change the colors (under $\chi$)  of some edges in $\{ e \in E : \chi(e) \neq \bot\}$ and assign a color  $\chi(e) \in [\Delta]$ to every edge  $e \in S$, while ensuring that $\chi$  remains a valid partial coloring of $G$.

In the rest of this section, we present an overview of the proof of the theorem below; in the runtime bound we will not optimize the logarithmic factors.

\begin{theorem}
\label{thm:bipartite}
There is a deterministic algorithm that,
given as input a {\em bipartite} graph $G = (V, E)$ with $n$ vertices, $m$ edges and maximum degree $\Delta$, and a valid partial coloring $\chi : E \rightarrow [\Delta] \cup \{ \bot \}$ such that the set $U = \{ e \in E : \chi(e) = \bot\}$ of uncolored edges forms a matching, extends $\chi$ to the edges in $U$ in time $\tilde{O}\left(m \cdot 2^{\Theta(\sqrt{\log \Delta})}\right)$.
\end{theorem}

We note that using a standard {\em Euler partitioning} technique, \Cref{thm:bipartite} immediately implies an $\tilde{O}\left(m \cdot 2^{\Theta(\sqrt{\log \Delta})} \right)$ time algorithm for $\Delta$-coloring a bipartite graph; see the last two paragraphs of \Cref{sec:forward:pointer} for a detailed discussion on this technique. Moreover, in this technical overview section, we ignore low-level implementation details of the supporting data structures, and hide polylogarithmic factors in runtime inside the $\tilde{O}(\cdot)$ notation. Finally, for simplicity of exposition, we assume that the input graph $G$ is almost-regular, as described below.

\begin{assumption}
\label{assume:regular}
In the input graph $G = (V, E)$, every vertex $v \in V$ has degree at least $\Omega(\Delta)$. Thus, as a corollary, we have $m = \Theta(n\Delta)$.
\end{assumption}

Before proceeding with the proof-sketch of \Cref{thm:bipartite}, we make two important remarks. First, the algorithmic framework developed in this section almost seamlessly extends to $(\Delta+1)$-coloring in general graphs, via the machinery of {\em u-fans} (see \Cref{sec:u-fans def}). Thus, although \Cref{thm:bipartite} in itself does not give us any new result, for it was known since the 1980s how to compute a $\Delta$-coloring in a bipartite graph in deterministic near-linear time~\cite{cole1982edge}, {\em our approach} for deriving \Cref{thm:bipartite} contains all the key conceptual insights that eventually lead to \Cref{thm:main:1}. Second, we use \Cref{assume:regular} purely for ease of exposition in this technical overview; this assumption does not take anything away from the key insights underpinning our approach.

\medskip
\noindent {\bf Roadmap:} In \Cref{sec:bipartite:prelim}, we introduce some key notations and concepts. \Cref{sec:bipartite:framework} presents \Cref{th:bipartite:sparsification}, which summarizes a {\em type sparsification} framework that underpins our entire algorithm. Immediately after stating \Cref{th:bipartite:sparsification}, we explain how it implies \Cref{thm:bipartite}. In \Cref{sec:randomized:bipartite}, we present a  {\em randomized} algorithm for type sparsification. In \Cref{sec:compare:bipartite}, we compare this new randomized algorithm against the algorithm of \cite{ABBC2025}, and point out the barrier towards efficiently derandomizing the latter.  \Cref{sec:deterministic:bipartite} demonstrates how the new randomized algorithm from \Cref{sec:randomized:bipartite} helps us overcome this barrier. Specifically, we derandomize the algorithm from \Cref{sec:randomized:bipartite}, which leads us to the proof of \Cref{th:bipartite:sparsification}.

\subsection{Notations and Preliminaries} 
\label{sec:bipartite:prelim}
Let $C = \{1, \ldots, \Delta\}$ be the palette of available colors. Moreover, let $\miss_{\chi}(u) := \{ \alpha \in C : \chi(u, v) \neq \alpha \text{ for all } (u, v) \in E\}$ denote the set of {\bf missing colors} at a vertex $u \in V$, w.r.t.~$\chi$. A {\bf type} is an unordered pair of distinct colors from $C$. For any two subsets $A, B \subseteq C$, we slightly abuse the notation and define $A \times B := \{ \{\alpha, \beta \} : \alpha \in A, \beta \in B \}$ to be the set of types with one color in $A$ and the other color in $B$. Given any type $\tau = \{ \alpha, \beta \} \in C \times C$
an {\bf alternating path} $P$ of type $\tau$ is a {\em maximal} path in $G = (V, E)$ whose edges are alternatively colored with $\alpha$ and $\beta$. We use the phrase {\bf flipping the alternating path $P$} to denote an operation where every edge on $P$ with color $\alpha$ (resp.~$\beta$) changes its color to $\beta$ (resp.~$\alpha$). It is easy to verify that the underlying partial coloring continues to remain valid even after we flip an alternating path. We say that {\bf an uncolored edge $(u, v) \in E$ is of type $\{\alpha, \beta\}$ iff $\alpha \in \miss_{\chi}(u)$ and $\beta \in \miss_{\chi}(v)$} (or vice versa). We say that an {\bf  alternating path (resp.~uncolored edge) is of type $A \times B$ iff it is of type $\tau$ for some $\tau \in A \times B$}.

\subsection{The Framework: Type Sparsification}
\label{sec:bipartite:framework}

The theorem below encapsulates our main technical contribution in this paper. 

\begin{theorem}
\label{th:bipartite:sparsification}
Consider any palette $C$ of colors, and  any  valid partial coloring $\chi : E \longrightarrow C \cup \{ \bot \}$ of the input graph $G = (V, E)$, such that the set $U = \{ e \in E : \chi(e) = \bot\}$ of uncolored edges forms a matching. Define $\lambda = |U|$. Fix any parameter $\eta \in [|C|]$ such that $|C|/(2\eta)$ is an integer, and any partition of the palette $C$ into $\eta$ subsets $\C_1, \ldots, \C_{\eta} \subseteq C$, such that $|{\C_i}| = |C|/\eta$ for all $i \in [\eta]$.

Then, in $\tilde{O}(m \cdot \poly(\eta))$ time, we can change the colors assigned to some of the edges in $E \setminus U$, and ensure that after these changes a constant fraction of the edges in $U$ have types in $\bigcup_{k=1}^{\eta} (\C_k \times \C_k)$.
\end{theorem}

\medskip
\noindent {\bf Remark.} Unless explicitly specified otherwise, throughout the rest of \Cref{sec:tech overview} we will use the symbol $C = \{1, \ldots, \Delta\}$ to denote a palette of $\Delta$ colors, where $\Delta$ is the maximum degree of the input graph $G = (V, E)$. However, \Cref{th:bipartite:sparsification} holds even if $|C| < \Delta$, as long as $\chi : E \longrightarrow C \cup \{\bot\}$ remains a valid partial coloring. 

We say that \Cref{th:bipartite:sparsification} achieves {\bf type sparsification}, for the following reason. Initially, the types of the edges in $U$ are chosen from the set $C \times C = [\Delta] \times [\Delta]$, and there are $O(\Delta^2)$ many such types. However, after we run the deterministic algorithm guaranteed by this theorem, the types of a constant fraction of the edges in $U$ are chosen from  $\bigcup_{k=1}^\eta (\C_k \times \C_k)$, and there are $\eta \cdot O((\Delta/\eta)^2) = O(\Delta^2/\eta)$ many such types. So, for a constant fraction of the edges in $U$, the set of possible types shrinks by a factor of $\Theta(\eta)$; we say that such edges in $U$ {\bf survive} this type sparsification.

Intuitively, imagine that we have a $\eta\times \eta$ matrix representing a heatmap, where the $(i, j)$-th entry is the number of edges in $U$ whose type belongs to $\C_i\times \C_j$. Suppose that the number of edges with types in each $\C_i\times \C_j$ initially is more or less balanced, so that, for each $1\leq i, j\leq \eta$, the number of edges in $U$ with types in $\C_i\times \C_j$ is at most $|U| / \eta^2$. Then the heatmap initially looks quite uniform. The goal of \Cref{th:bipartite:sparsification} is to modify the partial coloring $\chi$ so that this matrix looks more diagonal. See \Cref{diag} for an illustration.

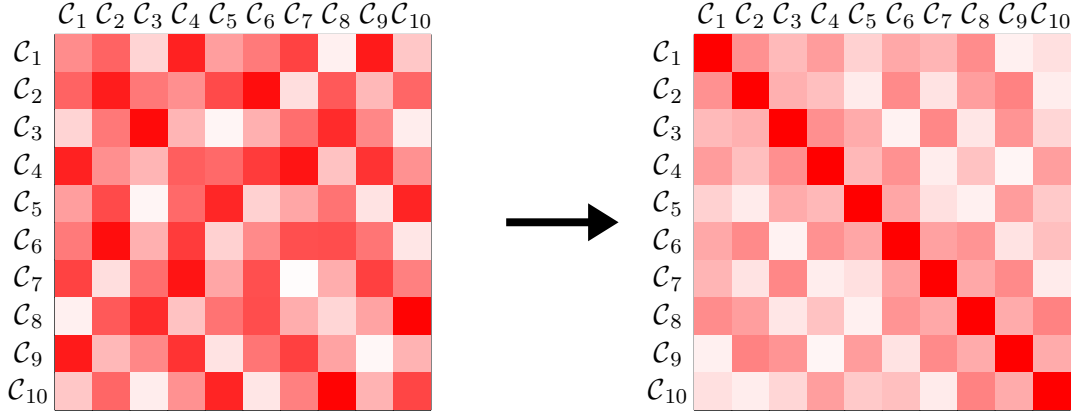
\begin{figure}
	\centering
	\begin{tikzpicture}[scale=0.5]
		\def\rows{10}
		\def\cols{10}
		
		\def\matrixdata{
			{0.45, 0.61, 0.17, 0.87, 0.38, 0.52, 0.74, 0.06, 0.90, 0.22},
			{0.61, 0.89, 0.53, 0.44, 0.71, 0.95, 0.13, 0.65, 0.28, 0.60},
			{0.17, 0.53, 0.96, 0.29, 0.04, 0.31, 0.57, 0.83, 0.47, 0.07},
			{0.87, 0.44, 0.29, 0.63, 0.59, 0.77, 0.92, 0.24, 0.80, 0.43},
			{0.38, 0.71, 0.04, 0.59, 0.85, 0.18, 0.35, 0.55, 0.11, 0.86},
			{0.52, 0.95, 0.31, 0.77, 0.18, 0.46, 0.69, 0.70, 0.54, 0.10},
			{0.74, 0.13, 0.57, 0.92, 0.35, 0.69, 0.01, 0.32, 0.75, 0.50},
			{0.06, 0.65, 0.83, 0.24, 0.55, 0.70, 0.32, 0.16, 0.36, 0.99},
			{0.90, 0.28, 0.47, 0.80, 0.11, 0.54, 0.75, 0.36, 0.03, 0.30},
			{0.22, 0.60, 0.07, 0.43, 0.86, 0.10, 0.50, 0.99, 0.30, 0.73}
		}
		
		\draw (0,0) grid (\cols, \rows);
		
		\foreach \i in {1,...,\rows} {
			\foreach \j in {1,...,\cols} {
				\pgfmathsetmacro{\val}{1-{\matrixdata}[\i-1][\j-1]}
	
				\definecolor{cellcolor}{rgb}{1, \val, \val} 
				\fill[cellcolor] (\j-1, \rows - \i) rectangle ++(1,1);
			}
		}
		
		\foreach \j in {1,...,\cols} {
			\node at (\j - 0.5, \rows + 0.5) {$\C_{\j}$};
		}
		
		\foreach \i in {1,...,\rows} {
			\node at (-0.5, \rows - \i + 0.5) {$\C_{\i}$ \,};
		}
		
		\draw[->,  >={Triangle}, thick, line width=3] (12, 5) -- (15, 5);
		
		\draw (17,0) grid (\cols+17, \rows);
		\def\diagmatrix{
			{1.00, 0.43, 0.27, 0.39, 0.18, 0.34, 0.29, 0.45, 0.06, 0.12},
			{0.43, 1.00, 0.31, 0.25, 0.08, 0.46, 0.11, 0.38, 0.49, 0.07},
			{0.27, 0.31, 1.00, 0.44, 0.33, 0.05, 0.47, 0.10, 0.42, 0.16},
			{0.39, 0.25, 0.44, 1.00, 0.28, 0.43, 0.07, 0.24, 0.04, 0.38},
			{0.18, 0.08, 0.33, 0.28, 1.00, 0.35, 0.12, 0.06, 0.39, 0.21},
			{0.34, 0.46, 0.05, 0.43, 0.35, 1.00, 0.37, 0.42, 0.11, 0.25},
			{0.29, 0.11, 0.47, 0.07, 0.12, 0.37, 1.00, 0.34, 0.46, 0.08},
			{0.45, 0.38, 0.10, 0.24, 0.06, 0.42, 0.34, 1.00, 0.33, 0.49},
			{0.06, 0.49, 0.42, 0.04, 0.39, 0.11, 0.46, 0.33, 1.00, 0.33},
			{0.12, 0.07, 0.16, 0.38, 0.21, 0.25, 0.08, 0.49, 0.33, 1.00}
		}
		
		\foreach \i in {1,...,\rows} {
			\foreach \j in {1,...,\cols} {
				\pgfmathsetmacro{\val}{1-{\diagmatrix}[\i-1][\j-1]}
				\definecolor{cellcolor}{rgb}{1, \val, \val} 
				\fill[cellcolor] (\j-1+17, \rows - \i) rectangle ++(1,1);
			}
		}
		
		Add column labels
		\foreach \j in {1,...,\cols} {
			\node at (\j - 0.5 + 17, \rows + 0.5) {$\C_{\j}$};
		}
		
		\foreach \i in {1,...,\rows} {
			\node at (-0.5 + 17, \rows - \i + 0.5) {$\C_{\i}$ \,};
		}
		
\end{tikzpicture}
	\caption{In this example, we have $\eta = 10$. The matrix heatmap is initially the left one. After we run the algorithm of \Cref{th:bipartite:sparsification}, most weights are concentrated along the diagonal.}
	\label{diag}
\end{figure}

We now explain why \Cref{th:bipartite:sparsification} implies \Cref{thm:bipartite}.
Fix a small $\epsilon \in (0, 1)$. For now, think of $\epsilon$ as being a small constant. Set $\eta := \Delta^\epsilon$.
For simplicity, suppose that $\Delta$ (the maximum degree of the input graph) is $2$ times an integer power $\eta$, i.e. $\Delta = 2 \cdot \eta^{q}$ for some $q \in \mathbb N$.\footnote{We only make this assumption in the technical overview to simplify the technical details, since this leads to a clean recursion tree. We do not need such an assumption in the actual proof.}

Next, we apply the type sparsification algorithm from \Cref{th:bipartite:sparsification} on the input graph $G = (V, E)$. This takes $\tilde{O}\left(m \cdot \Delta^{\Theta(\epsilon)}\right)$ time. Once the algorithm finishes execution, let $U^{(s)} \subseteq U$ be the set of surviving uncolored edges. Consider the following natural partition of $U^{(s)}$ into subsets $U_1, \dots, U_{\eta} \subseteq U^{(s)}$, where $U_i := \{ e \in U^{(s)} : e \text{ is of type } \C_i \times \C_i\}$ for each $i \in [\eta]$. \Cref{th:bipartite:sparsification} guarantees that $|U^{(s)}| \geq \lambda/\kappa$, where $\lambda = |U|$ and $\kappa > 0$ is a sufficiently large absolute constant. Now, for each $i \in [\eta]$, define the subgraph $\mathcal{G}_i = (V, \mathcal{E}_i)$, where $\mathcal{E}_i := U_i \cup \{e \in E : \chi(e) \in \mathcal{C}_i\}$. Let $\chi_i : \mathcal{E}_i \longrightarrow \C_i \cup \{ \bot \}$ denote the restriction of $\chi$ in $\mathcal{G}_i$, so that we have $\chi_i(e) = \chi(e)$ for all $e \in \mathcal{E}_i$. Note that $\chi_i$ is a valid partial coloring with the palette $\C_i$ in $\mathcal{G}_i$, and that $|{\C_i}| = \Delta/\eta$. Furthermore, the edge-sets $\{ \mathcal{E}_i\}_{i \in [\eta]}$ and the palettes $\{ \mathcal{C}_i \}_{i \in [\eta]}$ are both mutually disjoint.

We now recursively invoke the type sparsification algorithm from \Cref{th:bipartite:sparsification} on $(\mathcal{G}_i, \chi_i, \mathcal{C}_i)$, for all $i \in [\eta]$. We refer to the quantity  $|\mathcal{C}_i|$ as the {\bf palette-size} of the recursive call of $(\mathcal{G}_i, \chi_i, \mathcal{C}_i)$. The base-case of this recursive algorithm occurs when the palette-size becomes small, and in particular $2$ (by our assumption that $\Delta$ is $2$ times an integer power of $\eta$);\footnote{Without this assumption, the base case occurs when the palette size becomes $O(\eta)$. 
It is easy to extend the argument of this section to a base case with a palette size of $O(\eta)$, as we demonstrate in the actual proof.} at that point we do not make any further recursive calls.
Finally, the parameter $\eta$ remains  equal to $\Delta^\epsilon$, where $\Delta$ is the maximum degree of the initial input graph, {\em across all the recursive calls}.

It is easy to verify the following facts: (i) Since the palette-size decays by a factor of $\eta = \Delta^\epsilon$ at each recursion layer,
the recursion tree of this procedure has depth $1/\epsilon$. (ii) Across any two consecutive layers of this recursion tree, the number of surviving edges in $U$ drops by at most a factor of $\kappa$. Thus, across the very last layer of this recursion tree, the total number of surviving uncolored edges is $\Omega\left(\lambda/\kappa^{1/\epsilon} \right)$. (iii) Since the edge-sets and color palettes of different subgraphs at each layer are mutually disjoint, the time spent on each layer of this recursion tree is $\tilde{O}(m \cdot \text{poly}(\eta)) = \tilde{O}\left(m \cdot \Delta^{\Theta(\epsilon)} \right)$. So, the overall runtime across all the layers is $\tilde{O}\left(\epsilon^{-1} \cdot m \cdot \Delta^{\Theta(\epsilon)} \right)$.

Now, consider a ``node''\footnote{Not to be confused with a vertex in the input graph.} $\left( \mathcal{G}' = (V, \mathcal{E}'), \chi', \mathcal{C}'\right)$ at the last layer of this recursion tree. Let $\mathcal{U}' \subseteq U \cap \mathcal{E}'$ denote the surviving uncolored edges in $\mathcal{G}'$, under $\chi'$. Since $|{\C'}| = 2$, $\chi'$ is a valid partial coloring of $\mathcal{G}'$ with palette $\mathcal{C}'$, every edge in $\mathcal{U}'$ is of type $\mathcal{C}' \times \mathcal{C}'$, and the edges in $\mathcal{U}'$ form a matching, it is easy to verify that the subgraph $\mathcal{G}'$ has maximum degree $\leq 2$. Thus, in a straightforward manner, we extend the partial coloring $\chi'$ to a constant fraction of the edges in $\mathcal{U}'$ in deterministic $\tilde{O}\left(|\mathcal{E}'| \right)$ time by flipping alternating paths of at most one type, without impacting the other  ``nodes'' at the last layer of the recursion tree (as their palettes have no overlap with $\mathcal{C}'$).

To summarize, we infer that in deterministic $\tilde{O}\left(\epsilon^{-1} \cdot m \cdot \Delta^{\Theta(\epsilon)} \right)$ time, the above procedure extends the partial coloring $\chi$ to $(1/\kappa^{1/\epsilon})$-fraction of the initial set $U$ of uncolored edges. Repeating this procedure $\tilde{O}\left(\kappa^{1/\epsilon}\right)$ many times, we obtain a deterministic algorithm for extending the initial partial coloring $\chi$ to all the edges in $U$. This algorithm runs in $\tilde{O}\left(\kappa^{1/\epsilon} \cdot \epsilon^{-1} \cdot m \cdot \Delta^{\Theta(\epsilon)} \right)$ time. Now, \Cref{thm:bipartite} follows if we balance the relevant terms by setting $\epsilon := \Theta(1 / \sqrt{ \log \Delta})$.

\subsection{A Randomized Algorithm for Type Sparsification}
\label{sec:randomized:bipartite}

We now present a {\em randomized} algorithm that performs the same task as specified in the statement of \Cref{th:bipartite:sparsification}.  \Cref{th:round:bipartite:randomized} summarizes this result. For simplicity of exposition, throughout this section we assume that $C = \{1, \ldots, \Delta\}$ is a palette of $\Delta$ colors (see the remark immediately after \Cref{th:bipartite:sparsification}). We start by introducing a few more useful notations and terminologies.

For every $k \in [\eta]$, we further  partition (arbitrarily) the set $\C_k$ into two equally sized subsets $C_{2k-1} \subseteq \C_k$ and $C_{2k} = \C_k \setminus C_{2k-1}$, so that $|C_{2k-1}| = |C_{2k}| = \Delta/(2\eta) = r$ (say). Thus, the palette $C = [\Delta]$ is partitioned into subsets $C_1,  \ldots, C_{2\eta}$, where $\C_k = C_{2k-1} \cup C_{2k}$ for all $k \in [\eta]$. By relabeling the colors, w.l.o.g.~we assume that $C_i = [(i-1)r+1, ir]$ for all $i \in [2\eta]$. We say that a type $\tau \in C \times C$
is {\bf uniform} iff $\tau \in C_i \times C_i$ for some $i \in [2\eta]$, and {\bf aligned} iff $\tau \in C_{2k-1} \times C_{2k}$ for some $k \in [\eta]$. A type $\tau$ is {\bf social} if it is either uniform or aligned. Finally, we say that an edge $e \in U$ is {\bf social} if $e$ has a type that is social. In words, \Cref{th:bipartite:sparsification} asks us to change the colors of some edges in $E \setminus U$, so as to ensure that a constant fraction of the edges in $U$ become social.

At the start of our type sparsification algorithm, let $U_{\texttt{uni}} \subseteq U$ denote the collection of uncolored edges that have uniform types, and let $U^\star := U \setminus U_{\texttt{uni}}$ denote the remaining set of uncolored edges. By scanning through the edges in $U$, we can identify the subsets $U_{\texttt{uni}}$ and $U^\star$ in $\tilde{O}(|U| \cdot \Delta) = \tilde{O}(n\Delta) = \tilde{O}(m)$ time (see \Cref{assume:regular}). If $|U_{\texttt{uni}}| \geq |U|/2$, then we are done, since a constant fraction of the edges are already social. Thus, for the rest of this section, we assume that $|U^\star| \in [|U|/2, |U|] = [\lambda/2, \lambda]$. Our goal will be to ensure that $\Omega(\lambda)$ many  edges in $U^\star$ are of aligned types (by changing the colors of some of the edges in $E \setminus U$).

We will change the colors of some edges in $E \setminus U$, by flipping only some {\bf relevant} alternating paths. To be more specific, for any two distinct indices $i, j \in [2 \eta]$, define the mapping $\phi_{i \to j} : C_i \longrightarrow C_j$, as follows. For all $t \in [r]$, we have $\phi_{i \to j}((i-1)r+t) = (j-1)r+t$. In words,  $\phi_{i \to j}$ is a one-to-one function which maps the $t^{th}$ color in $C_i$ to the $t^{th}$ color in $C_j$. Observe that if $\phi_{i \to j}(\alpha) = \beta$, then $\phi_{j \to i}(\beta) = \alpha$; i.e., $\phi^{-1}_{i \to j} = \phi_{j \to i}$. We say that a {\bf  type $\tau$ is $j$-relevant, for some index $j \in [2\eta]$ iff $\tau = \{\alpha, \phi_{i \to j}(\alpha)\}$ for some $i \in [2\eta] \setminus \{j\}$ and $\alpha \in C_i$}; we say that {\bf an alternating path $P$ is $j$-relevant iff it is of a $j$-relevant type}; for convenience, define $\Gamma_j$ to be the set of all $j$-relevant types. Check \Cref{fig:color-match} for an illustration.

\begin{figure}
	\centering
	\begin{tikzpicture}[
		square/.style={draw, minimum size=1mm, fill=#1}, 
		node distance=0.5cm, 
		scale=0.8
		]
		
		\node[square=red] (1) at (0, 0) {};
		\node[square=OrangeRed] (2) at (0, 1) {};
		\node[square=BrickRed] (3) at (0, 2) {};
		\node[square=VioletRed] (4) at (0, 3) {};
		\node[square=Salmon] (5) at (0, 4) {};
		
		\draw[draw, thick] (0, 2) ellipse (1.5cm and 3cm);
		\node at (0,-1.5) {$C_i$};
		
		\node[square=cyan] (11) at (7, 0) {};
		\node[square=TealBlue] (12) at (7, 1) {};
		\node[square=MidnightBlue] (13) at (7, 2) {};
		\node[square=SkyBlue] (14) at (7, 3) {};
		\node[square=Cerulean] (15) at (7, 4) {};
		
		\draw[draw, thick] (7, 2) ellipse (1.5cm and 3cm);
		\node at (7,-1.5) {$C_j$};
				
		\draw[
		decorate, 
		decoration={snake, amplitude=0.8mm, segment length=3mm}, 
		thick, 
		gray 
		] (1) -- (11);
		
		\draw[
		decorate, 
		decoration={snake, amplitude=0.8mm, segment length=3mm}, 
		thick, 
		gray 
		] (2) -- (12);
		
		\draw[
		decorate, 
		decoration={snake, amplitude=0.8mm, segment length=3mm}, 
		thick, 
		gray 
		] (3) -- (13);
		
		\draw[
		decorate, 
		decoration={snake, amplitude=0.8mm, segment length=3mm}, 
		thick, 
		gray 
		] (4) -- (14);
		
		\draw[
		decorate, 
		decoration={snake, amplitude=0.8mm, segment length=3mm}, 
		thick, 
		gray 
		] (5) -- (15);
		
\end{tikzpicture}
	\caption{The mapping $\phi_{i \to j}$ defines a matching between colors in $C_i, C_j$. In this example we have drawn $5$ colors from color sets $C_i, C_j$, and colors connected by strings are matched together.}
	\label{fig:color-match}
\end{figure}
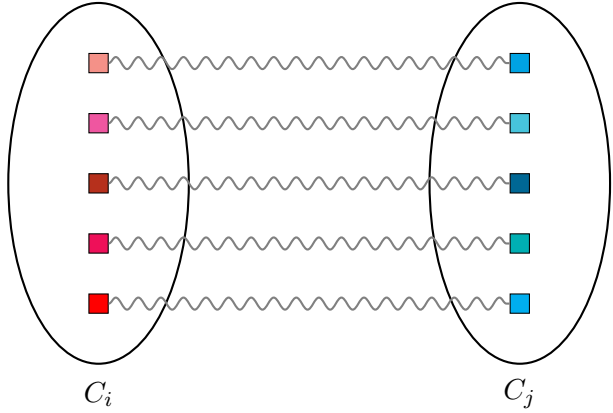

\begin{observation}
\label{obs:num:relevant:types}
For all $j \in [2 \eta]$, there are at most $O(\Delta)$ many types in $C \times C$ that are $j$-relevant.
\end{observation}

\begin{proof}
Suppose that $\tau = (\alpha, \beta)$ is a $j$-relevant type, and we are told that $\beta \in C_{j}$. Then, there are at most $(2\eta -1)$ choices for $\alpha$. Indeed, once we fix the index $i \in [2\eta] \setminus \{j\}$ such that $\alpha \in C_i$, there is only one unique choice for $\alpha$ which makes the type $\{\alpha, \beta\}$ $j$-relevant. Hence, there are at most $|C_{j}| \cdot (2\eta -1) = O((\Delta/\eta) \cdot \eta) = O(\Delta)$ many $j$-relevant types. 
\end{proof}

\subsubsection{Algorithm Description}
\label{sec:algo:describe:bipartite}

Our algorithm will consist of $\eta$ {\bf rounds}. In round $k \in [\eta]$, every alternating path we flip will be either $(2k-1)$-relevant or $(2k)$-relevant. In addition, we will satisfy the invariant stated below.

\begin{invariant}
\label{inv:round:bipartite} At the start of round $k \in [\eta]$, we have a collection of mutually-disjoint subsets $U_1, \ldots, U_{k-1} \subseteq U^\star$. For each $k' \in [k-1]$, the following two conditions hold.
\begin{enumerate}
\item $|U_{k'}| = \lambda/(100 \eta)$.
\item Every edge $e \in U_{k'}$ is of type $C_{2k'-1} \times C_{2k'}$. Thus, the set $U_{k'}$ consists only of social edges.
\end{enumerate}
\end{invariant}

The above invariant guarantees that at the end of the last (i.e., $\eta^{th}$ round), there are $\eta \cdot \lambda/(100 \eta) = \lambda/100$ social edges in $U^\star$, and so the algorithm terminates at this point. It remains to show how to implement a given round. Accordingly,  fix any index $k \in [\eta]$. {\bf We will now explain what happens in round $k$}.

At the start of the round, we have the sets $U_1, \ldots, U_{k-1} \subseteq U^\star$ satisfying \Cref{inv:round:bipartite}, and we initialize $U_k \leftarrow \emptyset$. Subsequently, we perform a sequence of {\bf iterations}. In a given iteration, we sample an edge $e = (u, v) \in U^\star \setminus (U_1 \cup \cdots \cup U_k)$ uniformly at random. If the iteration {\bf succeeds}, then we manage to ensure that the edge $e$ becomes of type $C_{2k-1} \times C_{2k}$ by changing the colors of some of the edges in $E \setminus U$, and we add the edge $e$ to the set $U_k$. Otherwise, the iteration {\bf fails}, and the set $U_k$, along with the colors assigned to the edges in $E \setminus U$, remains unchanged. In addition, we ensure that a successful iteration does not {\bf damage} any edge in $U_1 \cup \cdots \cup U_{k}$. More formally, during a successful iteration the sets $U_1, \ldots, U_{k-1}$ remain unchanged, and the only change that occurs in  $U_k$ is that the edge $e$ gets added to it. Furthermore, for each $k' \in [k]$, every edge $e' \in U_{k'}$ continues to be of type $C_{2k'-1} \times C_{2k'}$.  The round terminates after $\lambda/(100 \eta)$ successful iterations. It is easy to verify that \Cref{inv:round:bipartite} continues to hold at the start of the next round $(k+1)$.

{\bf We now focus on explaining how a given iteration works}. Let $e = (u, v) \in U^\star \setminus (U_1 \cup \cdots \cup U_k)$ be the edge we sample u.a.r.~at the start of the iteration, and suppose that $(u, v)$ is of type $\{\alpha, \beta\} \in C_i \times C_j$,  $i, j \in [2\eta]$,  $\alpha \in \miss_\chi(u) \cap C_i$, and $\beta \in \miss_\chi(v) \cap C_j$. To highlight the main idea, we focus on the most non-trivial scenario, which occurs when $\{i, j\} \cap \{2k-1, 2k\} = \emptyset$ and $\alpha \neq \beta$. W.l.o.g., suppose that $\alpha < \beta$.\footnote{Recall that each color is an integer in $[\Delta]$.} We designate the vertex $u$ as being the {\bf left endpoint} of $e$, and the vertex $v$ as being the {\bf right endpoint} of $e$. Let $\alpha_{2k-1} := \phi_{i \to (2k-1)}(\alpha)$ and $\beta_{2k} := \phi_{j \to 2k}(\beta)$. Let $P_u(e)$ (resp.~$P_v(e)$) be the alternating path starting from $u$ (resp.~$v$) that is of type $\{\alpha, \alpha_{2k-1}\}$ (resp.~$\{\beta, \beta_{2k}\}$).   Note that the paths $P_u(e)$ and $P_v(e)$ are respectively $(2k-1)$-relevant and $(2k)$-relevant (see \Cref{obs:relevant:bipartite}). {\bf We will use the notations $P_u(e)$ and $P_v(e)$ throughout the rest of this section.}

Ideally, we would like to flip the alternating paths $P_u(e)$ and $P_v(e)$: After these flips, we have $\alpha_{2k-1} \in \miss_\chi(u)$ and $\beta_{2k} \in \miss_{\chi}(v)$. So,  the edge $(u, v)$ becomes of type $C_{2k-1} \times C_{2k}$, and we can add the edge $(u, v)$ to the set $U_k$. There is, however, one serious problem with this strategy. Specifically, it might be the case that the path $P_u(e)$  ends at a vertex $u'$ that is the endpoint of an edge (say) $(u', v') \in U_{k'}$, for some $k' \in [k]$, and if we flip the path $P_u(e)$ then the edge $(u', v')$ will no longer be of type $C_{2k'-1} \times C_{2k'}$. (A similar situation can occur with the path $P_v(e)$.) In other words, flipping the path $P_u(e)$ would {\em damage} the edge $(u', v')$. In this case, we say that the edge $(u, v)$ is {\bf bad}, and furthermore, the edge $(u', v')$ is {\bf responsible} for the edge $(u, v)$ being bad. Check \Cref{fig:overview-flip} for an illustration.

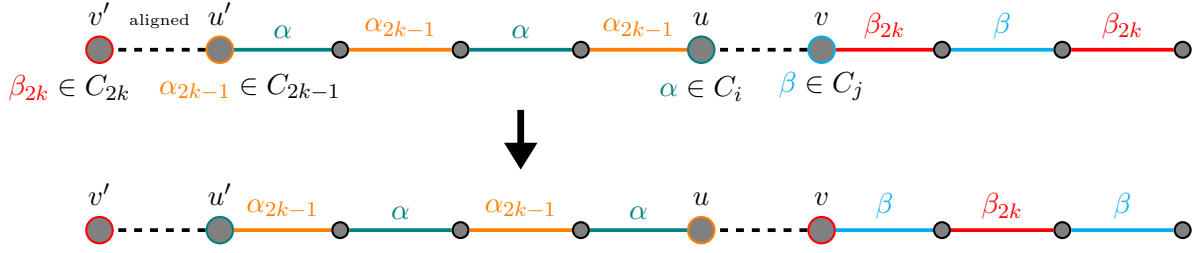
\begin{figure}
	\centering
	\begin{tikzpicture}[thick,scale=0.8]
	\draw (1, 0) node(1)[circle, draw, color=cyan, fill=black!50,
	inner sep=0pt, minimum width=10pt, label = $v$] {};
	
	\draw (-1, 0) node(2)[circle, draw, color=teal, fill=black!50,
	inner sep=0pt, minimum width=10pt, label = $u$] {};
	
	\draw (3, 0) node(3)[circle, draw, fill=black!50,
	inner sep=0pt, minimum width=6pt] {};
	
	\draw (5, 0) node(4)[circle, draw, fill=black!50,
	inner sep=0pt, minimum width=6pt] {};
	
	\draw (7, 0) node(5)[circle, draw, fill=black!50,
	inner sep=0pt, minimum width=6pt] {};
	
	\draw (-3, 0) node(6)[circle, draw, fill=black!50,
	inner sep=0pt, minimum width=6pt] {};
	
	\draw (-5, 0) node(7)[circle, draw, fill=black!50,
	inner sep=0pt, minimum width=6pt] {};
	
	\draw (-7, 0) node(8)[circle, draw, fill=black!50,
	inner sep=0pt, minimum width=6pt] {};
	
	\draw (-9, 0) node(9)[circle, draw, color=orange, fill=black!50,
	inner sep=0pt, minimum width=10pt, label = $u'$] {};
	\draw (-11, 0) node(10)[circle, draw, color=red, fill=black!50,
	inner sep=0pt, minimum width=10pt, label = $v'$] {};

    \draw (-8.5, -1.2) node[label={${\color{orange}\alpha_{2k-1}}\in C_{2k-1}$}] {};
    \draw (-11.5, -1.2) node[label={${\color{red}\beta_{2k}}\in C_{2k}$}] {};
    \draw (-1, -1.2) node[label={${\color{teal}\alpha}\in C_i$}] {};
    \draw (1, -1.2) node[label={${\color{cyan}\beta}\in C_j$}] {};
    \draw (-10, 0) node[label={\tiny aligned}] {};
	
	\draw [line width = 0.5mm, dashed] (1) to (2);
	\draw [line width = 0.5mm, color=red] (1) to node[above] {$\beta_{2k}$} (3);
	\draw [line width = 0.5mm, color=cyan] (3) to node[above] {$\beta$} (4);
	\draw [line width = 0.5mm, color=red] (4) to node[above] {$\beta_{2k}$} (5);
	\draw [line width = 0.5mm, color=orange] (2) to node[above] {$\alpha_{2k-1}$} (6);
	\draw [line width = 0.5mm, color=teal] (6) to node[above] {$\alpha$} (7);
	\draw [line width = 0.5mm, color=orange] (7) to node[above] {$\alpha_{2k-1}$} (8);
	\draw [line width = 0.5mm, color=teal] (8) to node[above] {$\alpha$} (9);
	\draw [line width = 0.5mm, dashed] (9) to (10);
	
    \draw[->, >={Triangle}, thick, line width = 0.9mm] (-4, -1) to (-4, -2);
	
	\draw (1, -3) node(11)[circle, draw, color=red, fill=black!50,
	inner sep=0pt, minimum width=10pt, label = $v$] {};
	
	\draw (-1, -3) node(12)[circle, draw, color=orange, fill=black!50,
	inner sep=0pt, minimum width=10pt, label = $u$] {};
	
	\draw (3, -3) node(13)[circle, draw, fill=black!50,
	inner sep=0pt, minimum width=6pt] {};
	
	\draw (5, -3) node(14)[circle, draw, fill=black!50,
	inner sep=0pt, minimum width=6pt] {};
	
	\draw (7, -3) node(15)[circle, draw, fill=black!50,
	inner sep=0pt, minimum width=6pt] {};
	
	\draw (-3, -3) node(16)[circle, draw, fill=black!50,
	inner sep=0pt, minimum width=6pt] {};
	
	\draw (-5, -3) node(17)[circle, draw, fill=black!50,
	inner sep=0pt, minimum width=6pt] {};
	
	\draw (-7, -3) node(18)[circle, draw, fill=black!50,
	inner sep=0pt, minimum width=6pt] {};
	
	\draw (-9, -3) node(19)[circle, draw, color=teal, fill=black!50,
	inner sep=0pt, minimum width=10pt, label = $u'$] {};
	\draw (-11, -3) node(20)[circle, draw, color=red, fill=black!50,
	inner sep=0pt, minimum width=10pt, label = $v'$] {};

	
	\draw [line width = 0.5mm, dashed] (11) to (12);
	\draw [line width = 0.5mm, color=cyan] (11) to node[above] {$\beta$} (13);
	\draw [line width = 0.5mm, color=red] (13) to node[above] {$\beta_{2k}$} (14);
	\draw [line width = 0.5mm, color=cyan] (14) to node[above] {$\beta$} (15);
	\draw [line width = 0.5mm, color=teal] (12) to node[above] {$\alpha$} (16);
	\draw [line width = 0.5mm, color=orange] (16) to node[above] {$\alpha_{2k-1}$} (17);
	\draw [line width = 0.5mm, color=teal] (17) to node[above] {$\alpha$} (18);
	\draw [line width = 0.5mm, color=orange] (18) to node[above] {$\alpha_{2k-1}$} (19);
	\draw [line width = 0.5mm, dashed] (19) to (20);
	
\end{tikzpicture}
	\caption{In this example, we attempt to align edge $(u, v)$ by flipping the $\{\alpha, \alpha_{2k-1}\}$-alternating path $P_u(e)$ from $u$ and the $\{ \beta, \beta_{2k}\}$-alternating path $P_v(e)$ from $v$. However, flipping the $\{\alpha, \alpha_{2k-1} \}$-alternating path from $u$ damages a previously aligned edge $(u', v')$ as $u'$ will not miss color $\alpha_{2k-1}$ anymore. In this case, $(u', v')$ is responsible for $(u, v)$.}
	\label{fig:overview-flip}
\end{figure}

Otherwise, if flipping the paths $P_u(e)$ and $P_v(e)$ does not damage any edge in $U_1 \cup \cdots \cup U_{k}$, then we say that the edge $(u, v)$ is {\bf good}. To summarize, we implement the concerned iteration as follows. We {\bf traverse} the two alternating paths $P_u(e)$ and $P_v(e)$, and confirm whether or not the edge $(u, v)$ is good. If the edge $(u, v)$ is good, then we {\bf flip} the paths $P_u(e)$ and $P_v(e)$, add the edge $(u, v)$ to the set $U_k$, and then proceed towards the next iteration.

\subsubsection{Analysis}
\label{sec:algo:analyze:bipartite}

 From the description in \Cref{sec:algo:describe:bipartite}, it immediately follows that when the algorithm terminates, a constant fraction of the edges in $U$ have types in $\bigcup_{k=1}^\eta (\C_k \times \C_k)$, as desired by \Cref{th:bipartite:sparsification}.  It remains to bound the expected time complexity of this algorithm. Towards this end, we start with a couple of simple observations. Subsequently, \Cref{cl:iteration:bipartite:0} and \Cref{cl:iteration:bipartite} bound the expected runtime of a given iteration within a round $k \in [\eta]$, whereas \Cref{cl:iteration:bipartite:1} and \Cref{cor:iteration:bipartite:1} bound the probability that a given iteration is successful. Putting everything together, \Cref{lem:round:bipartite} bounds the expected runtime of a given round. This, in turn, leads us to \Cref{th:round:bipartite:randomized}.

 \begin{observation}
\label{obs:relevant:bipartite} During any given  iteration within a round $k \in [\eta]$, we traverse and/or flip at most two alternating paths $P_u(e)$ and $P_v(e)$, where $e = (u,v) \in U^\star \setminus (U_1 \cup \cdots \cup U_k)$ is the edge we sample at the start of the iteration and $u$ (resp.~$v$) is the left (resp.~right) endpoint of $e$. The paths $P_u(e)$ and $P_v(e)$ are $(2k-1)$-relevant and $(2k)$-relevant, respectively.
 \end{observation}

 \begin{proof}
Follows from the description of the algorithm in \Cref{sec:algo:describe:bipartite}.
 \end{proof}

\begin{observation}
\label{obs:size:bipartite} Throughout the duration of round $k \in [\eta]$, we have $|U^\star \setminus (U_1 \cup \cdots \cup U_k)| \in [\lambda/4, \lambda]$.
\end{observation}

\begin{proof}
By \Cref{inv:round:bipartite}, we have $|U_1 \cup \cdots \cup U_k| \leq k \cdot \lambda/(100 \eta) \leq \eta \cdot \lambda/(100 \eta) = \lambda/100$. The observation follows since we have already assumed that  $|U^\star| \in [\lambda/2, \lambda]$,
\end{proof}


\begin{claim}
\label{cl:iteration:bipartite:0} At any given point in time within a round $k \in [\eta]$, define 
\begin{eqnarray*}
\mathcal{P}_{\texttt{left}} & := & \{ P_{u^\star}(e^\star) : e^\star \in U^\star \setminus (U_1 \cup \cdots \cup U_k), u^\star \text{ is the left endpoint of } e^\star \}, \text{ and } \\ \mathcal{P}_{\texttt{right}} & := & \{ P_{v^\star}(e^\star) : e^\star \in U^\star \setminus (U_1 \cup \cdots \cup U_k), v^\star \text{ is the right endpoint of } e^\star \}.
\end{eqnarray*}
Let $|P|$ denote the length of an alternating path $P$. Then, we have 
$\sum_{P \in \mathcal{P}_{\texttt{left}} \cup \mathcal{P}_{\texttt{right}}} |P| = O(m)$.
\end{claim}

\begin{proof}
Recall that the edge-set $U^\star \setminus (U_1 \cup \cdots \cup U_k) \subseteq U$ forms a matching (see \Cref{th:bipartite:sparsification}). By \Cref{obs:num:relevant:types}, we have $|\Gamma_{2k-1}| = O(\Delta)$;
recall that $\Gamma_{2k-1}$ is the set of all $(2k-1)$-relevant types.
Since the total length of all alternating paths of a given type is at most $n$, it follows that $\sum_{P \in \mathcal{P}_{\texttt{left}}} |P| \leq |\Gamma_{2k-1}| \cdot n = O(\Delta n) = O(m)$ (see \Cref{assume:regular}). 
Using an analogous argument, we get $\sum_{P \in \mathcal{P}_{\texttt{right}}} |P| = O(m)$. This concludes the proof.
\end{proof}

\begin{corollary}
\label{cl:iteration:bipartite}
Each iteration within a round $k \in [\eta]$  takes $\tilde{O}(m/\lambda)$ time in expectation.
\end{corollary}

\begin{proof}
Let $e = (u, v) \in U^\star \setminus (U_1 \cup \cdots \cup U_k)$ be the edge we sample u.a.r.~at the start of the iteration. 
 Using appropriate supporting data structures, it is easy to ensure that the time spent during the given iteration is proportional to the total length of the paths $P_u(e)$ and $P_v(e)$. Thus, from now on, we focus on bounding the expected {\bf length} of each of these two paths.

The edge-set $U^\star \setminus (U_1 \cup \cdots \cup U_k) \subseteq U$ forms a matching (see \Cref{th:bipartite:sparsification}). W.l.o.g., we assume that $u$ (resp.~$v$) is the left (resp.~right)  endpoint of the edge $e = (u, v)$. Accordingly, we have
\begin{equation}
\label{eq:length}
\mathbb{E}\left[|P_u(e)|\right] = \frac{\sum_{P \in \mathcal{P}_{\texttt{left}}} |P| }{|U^\star \setminus (U_1 \cup \cdots \cup U_k)|} \text{ and } \mathbb{E}\left[|P_v(e)|\right] = \frac{\sum_{P \in \mathcal{P}_{\texttt{right}}} |P| }{|U^\star \setminus (U_1 \cup \cdots \cup U_k)|}.
\end{equation}



The corollary now follows from \Cref{obs:size:bipartite} and \Cref{cl:iteration:bipartite:0}.
\end{proof}

\begin{claim}
\label{cl:iteration:bipartite:1} At the start of each iteration within round $k \in [\eta]$, at least a constant fraction of the edges in $U^\star \setminus (U_1 \cup \cdots \cup U_k)$ are good.
\end{claim}

\begin{proof}
The claim holds because of the following two key observations.

\begin{observation}
\label{ob:bipartite:100}
Each edge  $e' \in U_1 \cup \cdots \cup U_{k-1}$ is responsible for at most $4$ bad edges.
\end{observation}

\begin{observation}
\label{ob:bipartite:101}
Each edge  $e' \in U_{k}$ is responsible for at most  $4\eta$ bad edges. 
\end{observation}

Taken together, \Cref{ob:bipartite:100} and \Cref{ob:bipartite:101}  (along with \Cref{inv:round:bipartite}) imply that the number of bad edges is at most 
$$\sum_{k'=1}^{k-1} 4 \cdot |U_{k'}| + 4\eta \cdot |U_k| \leq 4k \cdot \lambda/(100 \eta) + 4\eta \cdot \lambda/(100 \eta) \leq 4\eta \cdot \lambda/(100 \eta) + \lambda/25 \leq (2/25) \cdot \lambda.$$
Since $|U^\star \setminus (U_1 \cup \cdots \cup U_k)| \in [\lambda/4, \lambda]$ by \Cref{obs:size:bipartite}, {\em at most} a constant fraction of the edges in $U^\star \setminus (U_1 \cup \ldots \cup U_k)$ are bad; or equivalently, {\em at least} a constant fraction of the edges in $U^\star \setminus (U_1 \cup \ldots \cup U_k)$ are good. It now remains to explain why the two key observations hold.

For \Cref{ob:bipartite:100}, consider any edge $e' = (u', v') \in U_{k'}$, for some $k' \in [k-1]$. W.l.o.g., let $e'$ be of type $\{\alpha', \beta'\}$, where $\alpha' \in \miss_{\chi}(u') \cap C_{2k'-1}$ and $\beta' \in \miss_{\chi}(v') \cap C_{2k'}$ (see \Cref{inv:round:bipartite}). Now, if $e'$ is responsible for some bad edge $e^\star = \left(u^\star, v^\star \right)$, then the following condition must hold.  Either  there exists an $i \in \{2k-1, 2k\}$ and an $x \in \{u^\star, v^\star\}$ such that the  alternating path of type $\{\alpha', \phi_{(2k'-1) \to i}(\alpha')\}$ starting from $u'$ ends at $x$; or  there exists an $i \in \{2k-1, 2k\}$ and an $x \in \{u^\star, v^\star\}$ such that the  alternating path of type $\{\beta', \phi_{(2k') \to i}(\beta')\}$ starting from $v'$ ends at $x$. It is easy to verify that there are at most two such alternating paths that start from $u'$ (resp.~$v'$), one for each  $i \in \{2k-1, 2k\}$. So, the edge $e'$ can be responsible for at most $2 \times 2 = 4$ bad edges. Check \Cref{fig:overview-socialize-damage1} for an illustration.

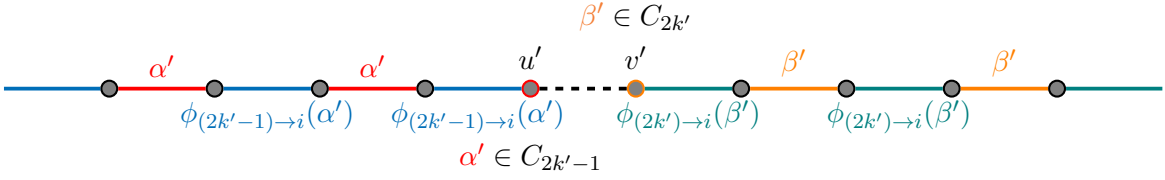
\begin{figure}
	\centering
	\begin{tikzpicture}[thick,scale=0.7]
	\draw (-7, 0) node(0)[circle, draw, color = red, fill=black!50,
	inner sep=0pt, minimum width=6pt, label = $u'$] {};
	\draw (-5, 0) node(1)[circle, draw, color=orange, fill=black!50,
	inner sep=0pt, minimum width=6pt, label = $v'$] {};
	
	\draw (-7, -2) node[label={${\color{red}\alpha'}\in C_{2k'-1}$}] {};
	\draw (-5, 0.7) node[label={${\color{Orange}\beta'}\in C_{2k'}$}] {};
	
	\draw (-3, 0) node(23)[circle, draw, fill=black!50,
	inner sep=0pt, minimum width=6pt] {};
	\draw (-1, 0) node(24)[circle, draw, fill=black!50,
	inner sep=0pt, minimum width=6pt] {};
	\draw (1, 0) node(25)[circle, draw, fill=black!50,
	inner sep=0pt, minimum width=6pt] {};
	\draw (3, 0) node(26)[circle, draw, fill=black!50, 
	inner sep=0pt, minimum width=6pt] {};	
		
	\draw [line width = 0.5mm, color=teal] (1) to node[below] {$\phi_{(2k')\to i}(\beta')$} (23);
	\draw [line width = 0.5mm, color=orange] (23) to node[above] {$\beta'$} (24);
	\draw [line width = 0.5mm, color=teal] (24) to node[below] {$\phi_{(2k')\to i}(\beta')$} (25);
	\draw [line width = 0.5mm, color=orange] (25) to node[above] {$\beta'$} (26);
	\draw [line width = 0.5mm, color=teal] (26) to (5, 0);
		
	\draw (-9, 0) node(73)[circle, draw, fill=black!50,
	inner sep=0pt, minimum width=6pt] {};
	\draw (-11, 0) node(74)[circle, draw, fill=black!50,
	inner sep=0pt, minimum width=6pt] {};
	\draw (-13, 0) node(75)[circle, draw, fill=black!50,
	inner sep=0pt, minimum width=6pt] {};
	\draw (-15, 0) node(76)[circle, draw, fill=black!50, 
	inner sep=0pt, minimum width=6pt] {};
		
	\draw [line width = 0.5mm, color=NavyBlue] (0) to node[below] {$\phi_{(2k'-1)\to i}(\alpha')$} (73);
	\draw [line width = 0.5mm, color=red] (73) to node[above] {$\alpha'$} (74);
	\draw [line width = 0.5mm, color=NavyBlue] (74) to node[below] {$\phi_{(2k'-1)\to i}(\alpha')$} (75);
	\draw [line width = 0.5mm, color=red] (75) to node[above] {$\alpha'$} (76);
    \draw [line width = 0.5mm, color=NavyBlue] (76) to (-17, 0);
	
	\draw [line width = 0.5mm, dashed] (0) to (1);
		
\end{tikzpicture}
	\caption{In this picture, $e' = (u', v') \in U_{k'}$, for some $k'<k$. We have drawn two alternating paths of types $\{\alpha', \phi_{(2k'-1) \to i}(\alpha')\}$ and $\{\beta', \phi_{(2k') \to i}(\beta')\}$ from $u', v'$ which could damage some edges $e^\star$.}
	\label{fig:overview-socialize-damage1}
\end{figure}

For \Cref{ob:bipartite:101}, consider any edge $e' = (u', v') \in U_{k}$. W.l.o.g., let $e'$ be of type $\{\alpha', \beta'\}$, where $\alpha' \in \miss_{\chi}(u') \cap C_{2k-1}$ and $\beta' \in \miss_\chi(v') \cap C_{2k}$ (see \Cref{inv:round:bipartite}). Now, if $e'$ is responsible for a bad edge $e^\star = \left(u^\star, v^\star \right)$, then the following condition must hold.  Either  there exists an $i \in [2\eta] \setminus \{2k-1, 2k\}$ and an $x \in \{u^\star, v^\star\}$  such that the  alternating path of type $\{\alpha', \phi_{(2k-1) \to i}(\alpha')\}$ starting from $u'$ ends at $x$; or  there exists an $i \in [2\eta] \setminus \{2k-1, 2k\}$ and an $x \in \{u^\star, v^\star\}$ such that the  alternating path of type $\{\beta', \phi_{(2k) \to i}(\beta')\}$ starting from $v'$ ends at $x$. It is easy to verify that there are at most $2\eta -2$ such alternating paths that start from $u'$ (resp.~$v'$), one for each  $i \in [2\eta] \setminus \{2k-1, 2k\}$. So, the edge $e'$ can be responsible for at most $2 \cdot (2\eta - 2) \leq 4\eta$ bad edges. Check \Cref{fig:overview-socialize-damage2} for an illustration.
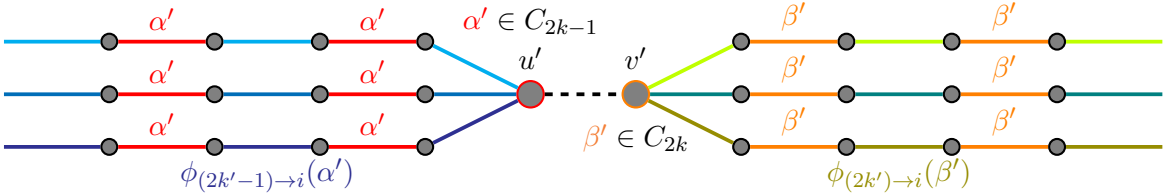
\begin{figure}
	\centering
	\begin{tikzpicture}[thick,scale=0.7]
	\draw (-7, 0) node(0)[circle, draw, color = red, fill=black!50,
	inner sep=0pt, minimum width=10pt, label = $u'$] {};
	\draw (-5, 0) node(1)[circle, draw, color=orange, fill=black!50,
	inner sep=0pt, minimum width=10pt, label = $v'$] {};
	
	 \draw (-7, 0.7) node[label={${\color{red}\alpha'}\in C_{2k-1}$}] {};
	 \draw (-5, -1.5) node[label={${\color{Orange}\beta'}\in C_{2k}$}] {};
	
	\draw (-3, -1) node(3)[circle, draw, fill=black!50,
	inner sep=0pt, minimum width=6pt] {};
	\draw (-1, -1) node(4)[circle, draw, fill=black!50,
	inner sep=0pt, minimum width=6pt] {};
	\draw (1, -1) node(5)[circle, draw, fill=black!50,
	inner sep=0pt, minimum width=6pt] {};
	\draw (3, -1) node(6)[circle, draw, fill=black!50, 
	inner sep=0pt, minimum width=6pt] {};
	
	\draw (-3, 1) node(13)[circle, draw, fill=black!50,
	inner sep=0pt, minimum width=6pt] {};
	\draw (-1, 1) node(14)[circle, draw, fill=black!50,
	inner sep=0pt, minimum width=6pt] {};
	\draw (1, 1) node(15)[circle, draw, fill=black!50,
	inner sep=0pt, minimum width=6pt] {};
	\draw (3, 1) node(16)[circle, draw, fill=black!50, 
	inner sep=0pt, minimum width=6pt] {};
	
	\draw (-3, 0) node(23)[circle, draw, fill=black!50,
	inner sep=0pt, minimum width=6pt] {};
	\draw (-1, 0) node(24)[circle, draw, fill=black!50,
	inner sep=0pt, minimum width=6pt] {};
	\draw (1, 0) node(25)[circle, draw, fill=black!50,
	inner sep=0pt, minimum width=6pt] {};
	\draw (3, 0) node(26)[circle, draw, fill=black!50, 
	inner sep=0pt, minimum width=6pt] {};	
	
	\draw [line width = 0.5mm, color=olive] (1) to (3);
	\draw [line width = 0.5mm, color=orange] (3) to node[above] {$\beta'$} (4);
	\draw [line width = 0.5mm, color=olive] (4) to node[below] {$\phi_{(2k')\to i}(\beta')$} (5);
	\draw [line width = 0.5mm, color=orange] (5) to node[above] {$\beta'$} (6);
	\draw [line width = 0.5mm, color=olive] (6) to (5, -1);
	
	\draw [line width = 0.5mm, color=teal] (1) to (23);
	\draw [line width = 0.5mm, color=orange] (23) to node[above] {$\beta'$} (24);
	\draw [line width = 0.5mm, color=teal] (24) to (25);
	\draw [line width = 0.5mm, color=orange] (25) to node[above] {$\beta'$} (26);
	\draw [line width = 0.5mm, color=teal] (26) to (5, 0);

	\draw [line width = 0.5mm, color=lime] (1) to (13);
	\draw [line width = 0.5mm, color=orange] (13) to node[above] {$\beta'$} (14);
	\draw [line width = 0.5mm, color=lime] (14) to (15);
	\draw [line width = 0.5mm, color=orange] (15) to node[above] {$\beta'$} (16);
		\draw [line width = 0.5mm, color=lime] (16) to (5, 1);

	\draw (-9, 1) node(63)[circle, draw, fill=black!50,
	inner sep=0pt, minimum width=6pt] {};
	\draw (-11, 1) node(64)[circle, draw, fill=black!50,
	inner sep=0pt, minimum width=6pt] {};
	\draw (-13, 1) node(65)[circle, draw, fill=black!50,
	inner sep=0pt, minimum width=6pt] {};
	\draw (-15, 1) node(66)[circle, draw, fill=black!50, 
	inner sep=0pt, minimum width=6pt] {};
	
	\draw (-9, 0) node(73)[circle, draw, fill=black!50,
	inner sep=0pt, minimum width=6pt] {};
	\draw (-11, 0) node(74)[circle, draw, fill=black!50,
	inner sep=0pt, minimum width=6pt] {};
	\draw (-13, 0) node(75)[circle, draw, fill=black!50,
	inner sep=0pt, minimum width=6pt] {};
	\draw (-15, 0) node(76)[circle, draw, fill=black!50, 
	inner sep=0pt, minimum width=6pt] {};
	
	\draw (-9, -1) node(83)[circle, draw, fill=black!50,
	inner sep=0pt, minimum width=6pt] {};
	\draw (-11, -1) node(84)[circle, draw, fill=black!50,
	inner sep=0pt, minimum width=6pt] {};
	\draw (-13, -1) node(85)[circle, draw, fill=black!50,
	inner sep=0pt, minimum width=6pt] {};
	\draw (-15, -1) node(86)[circle, draw, fill=black!50, 
	inner sep=0pt, minimum width=6pt] {};
	
	\draw [line width = 0.5mm, color=cyan] (0) to (63);
	\draw [line width = 0.5mm, color=red] (63) to node[above] {$\alpha'$} (64);
	\draw [line width = 0.5mm, color=cyan] (64) to (65);
	\draw [line width = 0.5mm, color=red] (65) to node[above] {$\alpha'$} (66);
	\draw [line width = 0.5mm, color=cyan] (66) to (-17, 1);	
	
	\draw [line width = 0.5mm, color=NavyBlue] (0) to (73);
	\draw [line width = 0.5mm, color=red] (73) to node[above] {$\alpha'$} (74);
	\draw [line width = 0.5mm, color=NavyBlue] (74) to (75);
	\draw [line width = 0.5mm, color=red] (75) to node[above] {$\alpha'$} (76);
	\draw [line width = 0.5mm, color=NavyBlue] (76) to (-17, 0);
	
	\draw [line width = 0.5mm, color=Blue] (0) to (83);
	\draw [line width = 0.5mm, color=red] (83) to node[above] {$\alpha'$} (84);
	\draw [line width = 0.5mm, color=Blue] (84) to node[below] {$\phi_{(2k'-1)\to i}(\alpha')$} (85);
	\draw [line width = 0.5mm, color=red] (85) to node[above] {$\alpha'$} (86);
	\draw [line width = 0.5mm, color=Blue] (86) to (-17, -1);
	
	\draw [line width = 0.5mm, dashed] (0) to (1);	
	
\end{tikzpicture}
	\caption{In this picture, $e' = (u', v') \in U_{k}$. We have drawn some alternating paths of types $\{\alpha', \phi_{(2k'-1) \to i}(\alpha')\}$ and $\{\beta', \phi_{(2k') \to i}(\beta')\}$ from $u', v'$ for three choices of $i\in [2\eta]\setminus \{2k-1, 2k\}$.}
	\label{fig:overview-socialize-damage2}
\end{figure}
\end{proof}

\noindent{\bf Remark.~} Note that the proofs of \Cref{obs:size:bipartite}, \Cref{cl:iteration:bipartite:0} and \Cref{cl:iteration:bipartite:1} do not rely on randomness in any way; we will use these statements also in the deterministic setting of \Cref{sec:deterministic:bipartite}.

\begin{corollary}
\label{cor:iteration:bipartite:1} Each iteration within a round $k \in [\eta]$ succeeds with constant probability.
\end{corollary}

\begin{proof}
At the start of the iteration, we sample an edge $e \in U^\star \setminus (U_1 \cup \cdots \cup U_k)$. The iteration succeeds if the sampled edge $e$ is good. The corollary now follows from \Cref{cl:iteration:bipartite:1}.
\end{proof}

\begin{lemma}
\label{lem:round:bipartite}
In the algorithm from \Cref{sec:algo:describe:bipartite}, each round takes $\tilde{O}(m/\eta)$ expected time.
\end{lemma}

\begin{proof}
By \Cref{cl:iteration:bipartite}, each iteration within a round takes $\tilde{O}(m/\lambda)$ expected time. By \Cref{cor:iteration:bipartite:1}, each iteration succeeds with $O(1)$ probability, hence w.h.p.\ the  number of iterations per  round is $\tilde{O}(\lambda/\eta)$. So, the expected time complexity of a round is  $\tilde{O}(\lambda/\eta) \cdot \tilde{O}(m/\lambda) = \tilde{O}(m/\eta)$. 
\end{proof}

\begin{theorem}
\label{th:round:bipartite:randomized}
The algorithm from \Cref{sec:algo:describe:bipartite} performs the same task as specified in the statement of \Cref{th:bipartite:sparsification}, and runs in expected $\tilde{O}(m)$ time.
\end{theorem}

\begin{proof}
Since the algorithm consists of $\eta$ rounds, the theorem follows from \Cref{lem:round:bipartite}.
\end{proof}


\subsection{Comparison with the Randomized Algorithm of \cite{ABBC2025}}
\label{sec:compare:bipartite}

We now explain how to derive the result of \cite{ABBC2025} from the randomized algorithm presented in \Cref{sec:randomized:bipartite}. Set $C = \{1, \ldots, \Delta\}$ to be the palette of $\Delta$ colors, and set $\eta := \Delta/2$. Thus, the palette $C = [\Delta]$ is  partitioned into subsets $\C_1, \ldots, \C_\eta \subseteq C$, where $|{\C_k}| = 2$ for all $k \in [\eta]$. W.l.o.g., suppose that $\C_k = \{2k-1, 2k\}$ for all $k \in [\eta]$.

Next, execute only the {\em first} round of the algorithm from \Cref{sec:algo:describe:bipartite}. By \Cref{lem:round:bipartite}, this takes $\tilde{O}(m/\eta) = \tilde{O}(n\Delta/\eta) = \tilde{O}(n)$ {\em expected time}. By \Cref{inv:round:bipartite}, at the end of this round we have a subset $U_1 \subseteq U$ of uncolored edges such that $|U_1| = \Omega(\lambda/\Delta)$, and every edge in $U_1$ is of type $\{1, 2\} \times \{1, 2\}$. At this point, \cite{ABBC2025} extends the partial coloring $\chi$ to a constant fraction of the edges in $U_1$ in {\em deterministic} $\tilde{O}(n)$ time, by flipping some  alternating paths of type $\{1, 2\} \times \{1, 2\}$. 


To summarize, the above procedure extends the initial partial coloring to $\Omega(1/\Delta)$ fraction of the uncolored edges, in $\tilde{O}(n)$ expected time. Repeating this procedure $\tilde{O}(\Delta)$ times, we get an algorithm that runs in $\tilde{O}(n \Delta) = \tilde{O}(m)$ expected time (see \Cref{assume:regular}), and extends the initial partial coloring to {\em all} the uncolored edges in $U$. Thus, the only part where randomization is used in this procedure is for this {\bf key task}: Implement round $k \in [\eta]$ of the algorithm from \Cref{sec:algo:describe:bipartite} in $\tilde{O}(n)$ time, when $\eta = \Delta/2$. We crucially require  that this key task gets implemented deterministically in {\em sublinear} $\hat{O}(n)$ time\footnote{$\hat{O}$ hides sub-polynomial factors.}, if we are to use it to design an almost-linear time deterministic algorithm for $\Delta$-coloring. This requirement is the {\em only} bottleneck towards our ultimate goal of derandomizing the algorithm of \cite{ABBC2025}. Nevertheless, the bottleneck seems insurmountable; just like many other computational problems in the sublinear setting, we do not see how to design a deterministic algorithm for this task that does {\em not}  read the entire input graph.

Alternatively, to implement the above algorithm from \cite{ABBC2025} in a way that is more compatible with our framework discussed in \Cref{sec:randomized:bipartite}, we could also execute all $\eta = \Delta/2$ rounds of socialization steps before doing any color extensions, and then we would end up with a constant fraction of uncolored edges in $U$ whose types are from
$\{1, 2\} \times \{1, 2\}, \{3, 4\} \times \{3, 4\},\ldots,
\{\Delta-1, \Delta\} \times \{\Delta-1, \Delta\}$.
At this point, we can easily 
extend the partial coloring to 
a constant fraction of the  edges in $U$
in $\tilde{O}(\eta\cdot n) = \tilde{O}(m)$ time.
Alas, de-randomizing this alternative approach would basically run into the same problem. 
The core difficulty in a direct derandomization of the extreme case when $\eta = \Delta/2$ is that the total number of possible types of alternating paths the algorithm could potentially flip is $\Omega(\Delta\eta) = \Omega(\Delta^2)$. Hence, since we are aiming at a total time bound of $\tilde{O}(m) = \tilde{O}(n\Delta)$, the average time spent on each type of alternating paths should be $\tilde{O}(n/\Delta)$. 
However, in the deterministic setting, for a single type of alternating paths, we cannot find and process a  $1/\Delta$-fraction of these paths in $\tilde{O}(n/\Delta)$ time,
since we can only exploit the property that their total length is bounded by $O(n)$.
In our deterministic algorithm, we will set $\eta$ to be much smaller than $\Delta/2$, namely $\eta = \Delta^\epsilon$ for a small $\epsilon$, with the advantage that we will only flip $O(\Delta^{1+\epsilon})$ different types of alternating paths.

Our main conceptual contribution in this paper is to formulate the notion of {\em type sparsification}, as captured in the statement of \Cref{th:bipartite:sparsification}, and derive the new randomized algorithm presented in \Cref{sec:randomized:bipartite}. Since \Cref{th:bipartite:sparsification} asks for a time complexity of $\tilde{O}(m \cdot \text{poly}(\eta))$ and the randomized algorithm from \Cref{sec:randomized:bipartite} has $\eta$ rounds, this gives us an added flexibility of  implementing a given round in  $\tilde{O}(m \cdot \text{poly}(\eta))$ time, as opposed to the sublinear $\tilde{O}(m/\eta)$ time guarantee of \Cref{lem:round:bipartite}. In \Cref{sec:deterministic:bipartite}, we show how to exploit this flexibility; and indeed present a {\em deterministic} $\tilde{O}(m \cdot \text{poly}(\eta))$ time algorithm for implementing a given round. This leads us to \Cref{th:bipartite:sparsification}, which, as explained in \Cref{sec:bipartite:framework}, implies \Cref{thm:bipartite}.

\subsection{Derandomization: Proof of \Cref{th:bipartite:sparsification}}
\label{sec:deterministic:bipartite}

We now show how to derandomize the algorithm from \Cref{sec:randomized:bipartite}. Recall the description of the randomized algorithm from \Cref{sec:algo:describe:bipartite}. Specifically, all we need is a deterministic subroutine for the following task: Implement a given round $k \in [\eta]$. We devote the rest of this section towards explaining how to perform this task in deterministic $\tilde{O}(m \cdot \text{poly}(\eta))$ time. Since the algorithm from \Cref{sec:algo:describe:bipartite} consists of $\eta$ rounds, this leads us to \Cref{th:bipartite:sparsification}. 


At any stage during a round $k \in [\eta]$, we say that a set $B \subseteq U^\star \setminus (U_1 \cup \cdots \cup U_k)$ of edges  is a {\bf batch} iff there exist two {\em distinct} indices $i, j \in [2\eta]$ such that every edge $e \in B$ is of type $C_i \times C_j$. The {\bf size} of the batch is given by $|B|$. We say that the batch is {\bf good} iff every edge $e \in B$ is good. Our deterministic implementation of a given round $k \in [\eta]$ is based on three key insights.

The first insight is summarized in \Cref{cl:good:batch:bipartite}, which guarantees that at any point in time during round $k \in [\eta]$, either (i) there exist $\Omega(\lambda)$ uncolored edges that are of uniform types, or (ii) there exists a good batch of $\Omega(\lambda/\eta^2)$ edges.  Under case (i), these $\Omega(\lambda)$ uncolored edges of uniform types are already social, and so we simply terminate the algorithm at this point. Thus, from now on we assume  the existence of a good batch $B$ of $\Omega(\lambda/\eta^2)$ uncolored edges.

The second insight is summarized in \Cref{cl:deterministic:bipartite:1}, which shows that we can compute the set of good edges (say) $U^{(g)} \subseteq U^\star \setminus (U_1 \cup \cdots \cup U_k)$ in deterministic $\tilde{O}(m)$ time. Once we have computed the set $U^{(g)}$, we  explicitly scan through the edges in $U^{(g)}$ to identify the largest good batch $B \subseteq U^{(g)}$. From the preceding discussion, we are guaranteed that $|B| = \Omega(\lambda/\eta^2)$. So, we have managed to identify a good batch $B$ of $\Omega(\lambda/\eta^2)$ edges in deterministic $\tilde{O}(m)$ time. 

To appreciate the third insight, consider any edge $e = (u, v) \in B$. Since $B$ is a good batch, the edge $e$ is good and has a type that is {\em not} aligned. Accordingly, if we flip the two alternating paths $P_u(e)$ and $P_v(e)$, then the edge $e$ would become of type $C_{2k-1} \times C_{2k}$ and get added to the set $U_k$, without damaging any existing edge in $U_1 \cup \cdots \cup U_k$. Let us refer to $P_u(e)$ and $P_v(e)$ as the {\bf characteristic alternating paths} of the edge $e$. We also say that $P_u(e)$ and $P_v(e)$ are {\bf characteristic alternating paths of the batch $B$} (since $e \in B$). Now, \Cref{cl:batch:flip:bipartite} and \Cref{cl:batch:flip:bipartite:101} allow us to conclude that the types of any two characteristic alternating paths of $B$ are either equal or mutually disjoint. Since the edges in $B \subseteq U$ form a matching (see \Cref{th:bipartite:sparsification}), this has the following implications: We can {\em simultaneously}\footnote{When we say that we can flip these paths \emph{simultaneously}, we mean that we can flip them in any arbitrary order and still have the same effect on the coloring.} flip all the characteristic alternating paths of $B$, without damaging any existing edge in $U^\star \setminus (U_1 \cup \cdots \cup U_k)$. Once we flip these paths, {\em all} the edges in $B$ get added to the set $U_k$. Furthermore, the time taken to perform this operation is proportional (up to polylogarithmic factors) to the total length of the characteristic alternating paths (summed over all the edges in $B$);  this, in turn, is at most $O(m)$ as per \Cref{cl:iteration:bipartite:0}.

To summarize, what we have described above is a deterministic procedure that runs in $\tilde{O}(m)$ time, and performs the following task: It adds $\Omega(\lambda/\eta^2)$ edges to the set $U_k$, without damaging any existing edge in $U_1 \cup \cdots \cup U_k$. Thus, we can implement round $k \in [\eta]$ by repeating this procedure $O(\eta)$ times (see \Cref{inv:round:bipartite}). In other words, there is a deterministic procedure for implementing a given round in $\tilde{O}(m \cdot \text{poly}(\eta))$ time. This implies \Cref{th:bipartite:sparsification}, since the  algorithm from \Cref{sec:algo:describe:bipartite} consists of $\eta$ rounds.

\begin{claim}
\label{cl:good:batch:bipartite}
At any stage during a round $k \in [\eta]$, let $\hat{U}_{\texttt{uni}}$ denote the set of  edges in $U^\star \setminus (U_1 \cup \cdots \cup U_k)$ that are of uniform types. Then, either $|\hat{U}_{\texttt{uni}}| = \Omega(\lambda)$, or there exists a good batch $B \subseteq U^\star \setminus (U_1 \cup  \cdots \cup U_k)$ consisting of $\Omega(\lambda/\eta^2)$  edges. 
\end{claim}

\begin{proof}
Fix a sufficiently large absolute constant $\kappa > 1$. If $|\hat{U}_{\texttt{uni}}| \geq \lambda/\kappa$, then the claim clearly holds. For the rest of the proof, we assume that $|\hat{U}_{\texttt{uni}}| < \lambda/\kappa$. By \Cref{obs:size:bipartite} and \Cref{cl:iteration:bipartite:1}, there are $\Omega(\lambda)$ good edges in $U^\star \setminus (U_1 \cup \cdots \cup U_k)$. Since $\kappa$ is a sufficiently large constant and $|\hat{U}_{\texttt{uni}}| < \lambda/\kappa$, there exists a set $S \subseteq U^\star \setminus (U_1 \cup \cdots \cup U_k)$ of $\Omega(\lambda)$ good edges that are {\em not} uniform. This set $S$ is partitioned into a collection of good batches. The total number of such batches is at most $O(\eta^2)$, since each batch is defined by two distinct indices $i, j \in [2\eta]$. Hence, by a simple averaging argument, there must exist a good batch $B$ of size $\Omega(\lambda/\eta^2)$. 
\end{proof}

\begin{claim}
\label{cl:deterministic:bipartite:1}
At any stage during a round $k \in [\eta]$, we can compute the set of good edges in $U^\star \setminus (U_1 \cup \cdots \cup U_k)$ in deterministic $\tilde{O}(m)$ time.
\end{claim}

\begin{proof}
We scan through all the edges in $U^\star \setminus (U_1 \cup \cdots \cup U_k)$, and for each such edge, we check whether it is good or bad. Consider any edge $e = (u, v) \in U^\star \setminus (U_1 \cup \cdots \cup U_k)$, and let $u$ (resp.~$v$) be its left (resp.~right) endpoint.  To check whether $e$ is good or bad, all we need to do is traverse the two alternating paths $P_u(e)$ and $P_v(e)$, and confirm whether flipping any of these paths damages some edge in $U^\star \setminus (U_1 \cup \cdots \cup U_k)$. If we use appropriate supporting data structures, then the time taken to implement this step is dominated by the total lengths of the alternating paths $P_u(e)$ and $P_v(e)$. Thus, upto a polylogarithmic factor, the time complexity of the entire procedure is proportional to $\sum_{P \in \mathcal{P}_{\texttt{left}} \cup \mathcal{P}_{\texttt{right}}} |P| = O(m)$, as per \Cref{cl:iteration:bipartite:0}.
\end{proof}

\begin{claim}
\label{cl:batch:flip:bipartite} At any stage during a round $k \in [\eta]$, let $B \subseteq U^\star \setminus (U_1 \cup \cdots \cup U_k)$ be a batch of edges, such that every edge in $B$ is of type $C_i \times C_j$, for $i, j \in [2\eta]$, $i \neq j$. Consider any two distinct edges $e = (u, v) \in B$ and $e' = (u', v') \in B$, and any two vertices $x\in \{u, v\}$ and $x' \in \{u', v'\}$. Let $\tau$ and $\tau'$ respectively denote the types of the alternating paths $P_x(e)$ and $P_{x'}(e')$. 

Then, it must necessarily be the case that either $\tau = \tau'$ or $\tau \cap \tau' = \emptyset$.
\end{claim}

\begin{proof}
Let $\{\alpha, \beta\}$ be the type of the edge $e = (u, v)$, with $\alpha \in \miss_\chi(u) \cap C_i$ and $\beta \in \miss_\chi(v) \cap C_j$. Similarly, let $\{\alpha', \beta'\}$ be the type of the edge $e' = (u', v')$, with $\alpha' \in \miss_\chi(u') \cap C_i$ and $\beta' \in \miss_\chi(v') \cap C_j$. W.l.o.g., suppose that $i < j$, i.e., $u$ and $u'$ (resp.~$v$ and $v'$) are the left (resp.~right) endpoints of $e$ and $e$. As in \Cref{sec:algo:describe:bipartite}, we focus on the most non-trivial scenario, where $\{i, j\} \cap \{2k-1, 2k\} = \emptyset$. Now, consider four possible cases.

\medskip
\noindent {\em Case 1: $x = u$ and $x' = u'$.} Here, both the types $\tau$ and $\tau'$ are $(2k-1)$-relevant. Specifically, we have $\tau = \{ \alpha, \phi_{i \to (2k-1)}(\alpha)\}$ and $\tau' = \{ \alpha', \phi_{i \to (2k-1)}(\alpha')\}$. Since $\phi_{i \to (2k-1)} : C_i \longrightarrow C_{2k-1}$ is a one-to-one mapping, we conclude that either $\tau = \tau'$ or $\tau \cap \tau' = \emptyset$.

\medskip
\noindent {\em Case 2: $x = v$ and $x' = v'$.} Here, both the types $\tau$ and $\tau'$ are $(2k)$-relevant, and the argument is analogous to the one in Case 1. 

\medskip
\noindent {\em Case 3: $x = u$ and $x' = v'$.} Here, the type $\tau$ is $(2k-1)$-relevant and the type $\tau'$ is $(2k)$-relevant. Specifically, we have $\tau \in C_i \times C_{2k-1}$ and $\tau' \in C_j \times C_{2k}$. Since $i \neq j$, we conclude that $\tau \cap \tau' = \emptyset$.

\medskip
\noindent {\em Case 4: $x = v$ and $x' = u'$.} Here, the type $\tau$ is $(2k)$-relevant and the type $\tau'$ is $(2k-1)$-relevant, and the argument is analogous to the one in Case 3.

\medskip
To summarize, from the above four cases, we infer that either $\tau = \tau'$ or $\tau \cap \tau' = \emptyset$.
\end{proof}

\begin{claim}
\label{cl:batch:flip:bipartite:101} At any stage during a round $k \in [\eta]$, consider an edge $e = (u, v) \in U^\star \setminus (U_1 \cup \cdots \cup U_k)$ whose type is {\em not}  aligned. Let $\tau$ and $\tau'$ respectively denote the types of the alternating paths $P_u(e)$ and $P_v(e)$. Then, it must necessarily be the case that  $\tau \cap \tau' = \emptyset$.
\end{claim}

\begin{proof}
Let $\{ \alpha, \beta\}$ be the type of the edge $e = (u, v)$, where $\alpha \in \miss_\chi(u) \cap C_i$, $\beta \in \miss_\chi(v) \cap C_j$, and $i \neq j$. W.l.o.g., suppose that $i < j$. As in \Cref{sec:algo:describe:bipartite}, we consider the most nontrivial scenario, where $\{i, j\} \cap \{2k-1, 2k\} = \emptyset$. Now, it is easy to verify that $\tau \in C_i \times C_{2k-1}$ and $\tau' \in C_j \times C_{2k}$. Since $i \neq j$, we get $\tau \cap \tau' = \emptyset$.
\end{proof}

\section{Preliminaries}\label{sec:prelim}

In this section, we give the preliminaries and define the notation used throughout the paper. The notation and preliminaries of our paper are very similar to \cite{ABBC2025}.

\subsection{Basic Notation}\label{sec:prelim:notation}

Let $G = (V, E)$ be graph on $n$ vertices with $m$ edges and maximum degree $\Delta$ and let $\chi : E \longrightarrow C \cup \{ \bot\}$ be a partial $|C|$-coloring of $G$ with colors $C$. We refer to edges $e \in E$ with $\chi(e) = \bot$ as \emph{uncolored}. Given a vertex $u \in V$, we denote the set of colors that are not assigned to any edge incident on $u$ by $\miss_\chi(u)$. We sometimes refer to $\miss_\chi(u)$ as the \emph{palette} of $u$. We say that the colors in $\miss_\chi(u)$ are \emph{missing} (or \emph{available}) at $u$. 

Given a path $P = e_1,\dots,e_k$ in $G$, we say that $P$ is an \emph{$\{\alpha, \beta\}$-alternating path} if $\chi(e_i)= \alpha$ whenever $i$ is odd and $\chi(e_i) = \beta$ whenever $i$ is even (or vice versa). We say that the alternating path $P$ is \emph{maximal} if one of the colors $\alpha$ or $\beta$ is missing at each of the endpoints of $P$. We refer to the process of changing the color of each edge $e_i \in P$ with color $\alpha$ (resp.~$\beta$) to $\beta$ (resp.~$\alpha$) as \emph{flipping} the path $P$. We denote by  $|P|$ the length (i.e., the number of edges) of the alternating path $P$. We define the \emph{length $i$ prefix} of the path $P$ to be the path $P_{\leq i} := e_1,\dots,e_i$.

Consider a set $U \subseteq E$ of edges that are uncolored under $\chi$, i.e., $\chi(e) = \bot$ for all $e \in U$. We use the phrase {\em ``extending $\chi$ to $U$''} to mean the following: Modify $\chi$ so as to ensure that $\chi(e)\neq \bot$ for all $e \in U$, without creating any new uncolored edges. When the set $U$ consists of a single edge $e$ (i.e., when $U = \{e\}$), we use the phrase {\em ``extending $\chi$ to the edge $e$''} instead of {\em ``extending $\chi$ to $U$''}.

Our algorithms will always work by modifying a partial coloring $\chi$; unless explicitly specified otherwise, every new concept we define (such as u-fans and separable collections in \Cref{sec:u-fans def}) will be defined with respect to this particular partial coloring $\chi$.

\subsection{U-Fans and Separable Collections}\label{sec:u-fans def}

We begin by defining the notion of \emph{u-fans} that was used by \cite{gabow1985algorithms}.\footnote{The term `u-fan' was originally introduced by \cite{gabow1985algorithms} as an abbreviation for `uncolored fan'.}

\begin{definition}[u-fan, \cite{gabow1985algorithms}]\label{def:u-fan}
    A \emph{u-fan} is a tuple $\f = (u, v, w, c_{\f}(u), c_{\f}(v), c_{\f}(w))$ where $u$, $v$ and $w$ are distinct vertices and $c_{\f}(u)$, $c_{\f}(v)$ and $c_{\f}(w)$ are colors such that:
    \begin{enumerate}
        \item $(u,v)$ and $(u,w)$ are uncolored edges.
        \item $c_{\f}(u) \in \miss_\chi(u)$, $c_{\f}(v) \in \miss_\chi(v)$ and $c_{\f}(w) \in \miss_\chi(w)$.
        \item $c_{\f}(u) \neq c_{\f}(v)$ and $c_{\f}(v) = c_{\f}(w)$.\label{item:diff cols}
    \end{enumerate}
\end{definition}
\noindent
We say that $u$ is the \emph{center} of $\f$ and that $v$ and $w$ are the \emph{leaves} of $\f$.
Given a vertex $x \in \f$, we say that $c_{\f}(x)$ is the available color that $\f$ `assigns' to $x$. 
We define a \textbf{damaged u-fan} in the same way as a u-fan, except that we do not require Condition~\ref{item:diff cols} in \Cref{def:u-fan}. 

\medskip
\noindent \textbf{Color Types:} Given a set of colors $C$, we define a \textbf{type} as an unordered pair of colors in $C$. Given $A, B \subseteq C$, we abuse notation slightly and let $A \times B$ denote the set of types with one color in $A$ and one color in $B$, i.e.~the set $\{ \{\alpha, \beta\} \mid \alpha \in A, \beta \in B \}$.

Given a u-fan $\f$, we define \emph{the type of $\f$} as $\tau(\f) := \{c_{\f}(x)\}_{x \in \f}$. Note that $|\tau(\f)| = 2$, so $\tau(\f)$ satisfies the definition of a type. For a damaged u-fan $\f$, we define $\tau(\f)$ in the same way, but it is not necessarily a type since we might have $|\tau(\f)| = 1$ or $|\tau(\f)| = 3$. For notational convenience, we still refer to $\tau(\f)$ as the type of $\f$ even if $\f$ is a damaged u-fan.

\medskip
\noindent \textbf{Activating U-Fans:}
Let $\f$ be a u-fan with center $u$, leaves $v$ and $w$, and type $\tau(\f) = \{\alpha,\beta\}$ such that $c_{\f}(u) = \alpha$ and $c_{\f}(v) = c_{\f}(w) = \beta$. The \emph{key property} of u-fans is that at most one of the $\{\alpha, \beta\}$-alternating paths starting at $v$ or $w$ ends at $u$. Suppose that the $\{\alpha, \beta\}$-alternating path starting at $v$ does not end at $u$. Then, after flipping this $\{\alpha, \beta\}$-alternating path, both vertices $u$ and $v$ are missing color $\alpha$. Thus, we can extend the coloring $\chi$ by assigning $\chi(u, v)\leftarrow \alpha$.
We refer to this as \emph{activating} the u-fan $\f$.

\medskip
\noindent \textbf{Collections of U-Fans:} Following the approach of \cite{ABBC2025},
throughout this paper, we often consider collections of u-fans $\U$. We only use the term `collection' in this context, so we abbreviate `collection of u-fans' by just `collection'.
We will be particularly interested in collections satisfying the following useful property, which we refer to as \emph{separability}.

\begin{definition}[Separable Collection, \cite{ABBC2025}]
    Let $\chi$ be a partial $\mu$-coloring of $G$ with colors $C$ and $\mathcal U$ be a collection of u-fans.
    We say that the collection $\mathcal U$ is \emph{separable} if the following holds:
    \begin{enumerate}
        \item All u-fans in $\mathcal U$ are edge-disjoint.
        \item For each $x \in V$, the colors in the multi-set
        $C_{\mathcal U}(x) := \{ c_{\f}(x) \mid \f \in \mathcal U, \, x \in \f\}$ are distinct.
    \end{enumerate}
\end{definition}
As discussed in \cite{ABBC2025}, the second property of this definition is rather important since we need to ensure that different u-fans are not interfering with each other when they share common vertices. See \Cref{fig:separable} for an illustration.

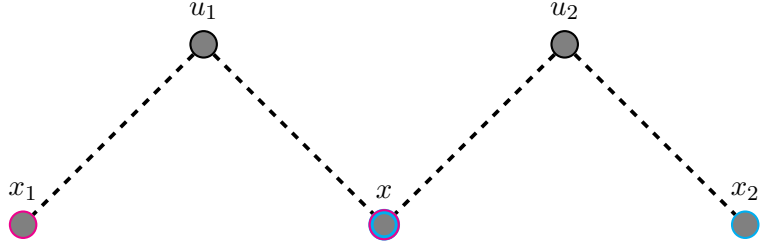
\begin{figure}
    \centering
    \begin{tikzpicture}[thick,scale=1.2]
	\draw (0, 0) node(x)[circle, draw, color=cyan, fill=black!50,
inner sep=0pt, minimum width=10pt, line width=2pt, label = $x$] {};
    \draw (0, 0) node(xx)[circle, draw, color=magenta, inner sep=0pt, minimum width=11pt] {};

	\draw (-2, 2) node(u1)[circle, draw, fill=black!50,
	inner sep=0pt, minimum width=10pt, label = $u_1$] {};
		
	\draw (-4, 0) node(x1)[circle, draw, color=magenta, fill=black!50,
	inner sep=0pt, minimum width=10pt, label = $x_1$] {};
		
	\draw (2, 2) node(u2)[circle, draw, fill=black!50,
	inner sep=0pt, minimum width=10pt, label=$u_2$] {};
	
	\draw (4, 0) node(x2)[circle, draw,  color=cyan, fill=black!50,
	inner sep=0pt, minimum width=10pt, label = $x_2$] {};
	
	\draw [line width = 0.5mm, dashed] (u1) to (x1);
	\draw [line width = 0.5mm, dashed] (u2) to (x2);
        \draw [line width = 0.5mm, dashed] (u1) to (x);
	\draw [line width = 0.5mm, dashed] (u2) to (x);
	
\end{tikzpicture}
    \caption{This picture shows two u-fans $(u_1, x_1, x, \cdot \,, {\color{magenta}\beta_1}, {\color{magenta}\beta_1})$ and $(u_2, x_2, x, \cdot \,, {\color{cyan}\beta_2}, {\color{cyan}\beta_2})$ sharing a common vertex $x$. The separability condition requires that $\beta_1\neq \beta_2$; for instance, $\beta_1$ and $\beta_2$ could be magenta and cyan as shown here.} 
    \label{fig:separable}
\end{figure}

\subsection{Constructing Separable Collections}

To construct separable collections of u-fans, we use a fast deterministic algorithm given by \cite{ABBC2025} which takes a set $U$ of $\lambda$ uncolored edges and either extends the coloring to a constant fraction of the edges in $U$ or shifts these uncolored edges around the graph in order to construct a separable collection of $\Omega(\lambda)$ u-fans. More specifically, we use the following lemma of \cite{ABBC2025} as a black box.

\begin{lemma}[\cite{ABBC2025}, Lemma 6.2]\label{lem:build u-fans}
    Given a graph $G$, a partial $(\Delta + 1)$-coloring $\chi$ of $G$ and a set of $\lambda$ uncolored edges $U$, there is a deterministic algorithm that does one of the following in $O((m + \Delta \lambda)\log \Delta)$ time:
    \begin{enumerate}
        \item  Extends the coloring to $\Omega(\lambda)$ uncolored edges.
        \item  Modifies $\chi$ to obtain a separable collection of $\Omega(\lambda)$ u-fans $\mathcal U$.
    \end{enumerate}
\end{lemma}

We remark that the algorithm described by \Cref{lem:build u-fans} crucially uses Vizing fans and chains. The rest of our algorithm does not use Vizing fans and chains.

\subsection{Alternating Paths in Separable Collections}

The main way that our algorithms modify the coloring $\chi$ of a graph is by flipping alternating paths. However, our algorithms will also often explicitly maintain a separable collection of u-fans $\U$ while modifying the coloring $\chi$, and we want to ensure that $\U$ remains separable at all times. Thus, whenever we flip an alternating path, we need to appropriately modify the missing colors assigned to (and hence also the types of) some u-fans.
We do this by using the following simple subroutine $\Flip$ to flip alternating paths.

\medskip
\noindent \textbf{The Algorithm $\Flip$:} As input, the subroutine $\Flip$ is given a graph $G$, a coloring $\chi$ of $G$ with colors $C$, a separable collection $\U$, and a maximal $\{\alpha,\beta\}$-alternating path $P$. 
The subroutine then flips the colors of the alternating path $P$ and modifies the missing colors assigned to the vertices of some u-fans in $\U$ to ensure that $\U$ remains separable.
The pseudocode in \Cref{alg:flip} gives a formal description of how we implement the procedure.

\begin{algorithm}[H]
    \SetAlgoLined
    \DontPrintSemicolon
    \SetKwRepeat{Do}{do}{while}
    \SetKwBlock{Loop}{repeat}{EndLoop}
    Let $x$ and $y$ be the endpoints of $P$\;
    Flip the colors of the $\{\alpha,\beta\}$-alternating path $P$\;
    \For{$z \in \{x,y\}$}{
        \If{\textnormal{\textbf{there exists}} $\f \in \U$ s.t.~$c_{\f}(z) \in \{\alpha,\beta\}$}{
            Let $c \in \{\alpha, \beta\} \setminus \{ c_{\f}(z) \}$\;
            Set $c_{\f}(z) \leftarrow c$\;
        }
    }
    \caption{$\Flip(P)$}
    \label{alg:flip}
\end{algorithm}

\begin{lemma}[\cite{ABBC2025}, Lemma 5.4]
    For any maximal alternating path $P$, the collection $\U$ remains separable after applying $\Flip(P)$. Furthermore, this damages at most $2$ u-fans.
\end{lemma}

\subsection{Data Structures}\label{sec:data struc overview}

In \Cref{sec:data structs}, we describe the data structures that we use to implement our algorithm. Our data structures are simple and almost identical to the data structures of \cite{ABBC2025}, with the modification that we use balanced binary trees instead of hashmaps to make our data structures deterministic, incurring an extra $O(\log \mu)$ factor overhead in the running time of each operation, where $\mu$ is the number of colors in the edge coloring $\chi$.

On top of the standard data structures used to maintain the $\mu$-coloring $\chi$ with colors $C$, we also use data structures that allow us to efficiently maintain a separable collection $\U$. More specifically, the data structures that we use to implement a separable collection $\U$ support the following queries.
\begin{itemize}
    \item $\textsc{Insert}_{\U}(\f)$:  The input to this query is a u-fan $\f$. In response, the data structure adds $\f$ to $\U$ if $\U \cup \{\f\}$ is separable and outputs $\texttt{fail}$ otherwise.
    \item $\textsc{Delete}_{\U}(\f)$: The input to this query is a u-fan $\f$. In response, the data structure removes $\f$ from $\U$ if $\f \in \U$ and outputs $\texttt{fail}$ otherwise.
    \item $\textsc{Find-U-Fan}_{\U}(x,c)$: The input to this query is a vertex $x \in V$ and a color $c \in C$. In response, the data structure returns the u-fan $\f \in \U$ with $c_{\f}(x) = c$ if such a u-fan exists and outputs $\texttt{fail}$ otherwise.
    \item  $\textsc{Missing-Color}_{\U}(x)$: The input to this query is a vertex $x \in V$. In response, the data structure returns an arbitrary color from the set $\miss_\chi(x) \setminus C_{\U}(x)$.\footnote{Note that, since $|C_{\U}(x)| < |\miss_\chi(x)|$, such a color always exists.}
\end{itemize}
The following claim shows that it is always possible to answer a $\textsc{Missing-Color}$ query.

\begin{claim}\label{claim:missing color}
    For each $x \in V$, the set $\miss_\chi(x) \setminus C_{\U}(x)$ is non-empty.
\end{claim}

\begin{proof}
    Let $d$ denote the number of uncolored edges incident on $x$.
    Since the collection $\U$ is separable, we have that $|C_{\U}(x)| \leq d$. Since $|\miss_\chi(x)| \geq d + 1$, it follows that $|C_{\U}(x)| < |\miss_\chi(x)|$ and so $\miss_\chi(x) \setminus C_{\U}(x) \neq \emptyset$.
\end{proof}

\noindent
Furthermore, the data structure supports the following initialization operation.
\begin{itemize}
    \item $\textsc{Initialize}(G, \chi)$: Given a graph $G$ and an edge coloring $\chi$ of $G$, we can initialize the data structure with an empty separable collection $\U = \emptyset$.
\end{itemize}
In \Cref{sec:data structs}, we show how to implement the initialization operation in $O(m \log \mu)$ time and each of these queries in $O(\log \mu)$ time with the appropriate data structures.
\textbf{These queries provide the `interface' via which our algorithm will interact with the u-components.}

\section{The Main Algorithm (Proof of \Cref{thm:main:1})}\label{sec:main}

Our final algorithm can be easily described using three different deterministic subroutines as black boxes, which will be presented in subsequent sections. In this section, we state lemmas that summarize the behavior of these subroutines and show how to use them to prove \Cref{thm:main:1}.

The first of these subroutines and the main technical component of our algorithm is a subroutine called $\Amplify$, which is described in \Cref{sec:proof key}. The following lemma
(proved in \Cref{sec:proof key})
summarizes the behavior of $\Amplify$.

\begin{restatable}{lemma}{sparsifytypes}\label{lem:key}
    Given a graph $G$, a partial $\mu$-coloring $\chi$ of $G$ with colors $C$, a separable collection $\U$ of size $\lambda$, and an integer $10 \leq \eta \leq \mu/ 10$, the algorithm $\Amplify$ modifies the coloring $\chi$ (without changing which edges are colored) and the separable collection $\U$, and constructs disjoint subsets of colors $\C_{1},\dots,\C_\eta \subseteq C$ with the following properties:
    \begin{enumerate}
        \item For all $k \in [\eta]$, we have that $|{\C_k}| \leq \mu / \eta$.
        \item The u-fans in $\U$ have types in $\bigcup_{k=1}^\eta (\C_k \times \C_k)$ and $|{\U}| \geq \lambda/100$.
    \end{enumerate}
    Furthermore, the algorithm is deterministic and runs in time $O(m \eta^4 \log \mu)$.
\end{restatable}

The second subroutine, described in \Cref{app:small},
is called $\Small$, and we use it to extend the coloring to sufficiently small subgraphs. The following lemma (proved in \Cref{app:small})
 summarizes the behavior of $\Small$.

\begin{restatable}{lemma}{colorsmall}
\label{lem:Sinnamon}
    Given a graph $G$, a partial $\mu$-coloring $\chi$ with colors $C$, and a separable collection $\U$ of size $\lambda$, the algorithm $\Small$ extends the coloring $\chi$ to at least $\lambda/100$ uncolored edges. Furthermore, the algorithm is deterministic and runs in $O(m \mu^2 \log \mu)$ time.
\end{restatable}

Using the algorithms $\Amplify$ and $\Small$, we construct a \emph{recursive} algorithm $\Extend$, described in \Cref{sec:extend}, which takes a partial coloring as input and efficiently extends the coloring to a large proportion of the uncolored edges. The following lemma (proved in \Cref{sec:extend})
summarizes the behavior of $\Extend$.

\begin{restatable}{lemma}{extendcolor}\label{lem:extend}
    Given a graph $G$, a partial $\mu$-coloring $\chi$ of $G$ with colors $C$, a separable collection $\U$ of size $\lambda$, and an integer $\eta \geq 10$, the algorithm $\Extend$ extends the coloring $\chi$ to at least $\lambda / 100^{(\log \mu / \log \eta) + 1}$ uncolored edges.
    Furthermore, the algorithm is deterministic and runs in time $O(m \eta^4 \log^2 (\mu) / \log (\eta))$. 
\end{restatable}

Using \Cref{lem:extend}, we can now prove \Cref{thm:main:1}.

\subsection{Proof of \Cref{thm:main:1}}
\label{sec:forward:pointer}

We begin by proving the following corollary of \Cref{lem:extend}, which we use to efficiently extend a $(\Delta + 1)$-coloring of a graph $G$.

\begin{corollary}\label{cor:extend}
     Given a graph $G$ and a partial $(\Delta + 1)$-coloring $\chi$ of $G$ with $\lambda = O(m / \Delta)$ uncolored edges $U$, we can extend $\chi$ to the remaining uncolored edges in time $O(m \cdot 2^{13 \sqrt{\log_2 \Delta}}  \log n)$.
\end{corollary}

\begin{proof}
    Suppose that we have a partial $(\Delta + 1)$-coloring $\chi$ with $\lambda$ uncolored edges. Applying \Cref{lem:build u-fans}, we either extend the coloring $\chi$ to $\Omega(\lambda)$ uncolored edges or modify $\chi$ to obtain a separable collection of $\Omega(\lambda)$ u-fans $\U$ in $O(m)$ time. In the latter case, we can then apply \Cref{lem:extend} to the coloring $\chi$ and the separable collection $\U$ with $\eta = 10 \cdot \lceil 2^{\sqrt{\log_2 \Delta}} \rceil$ to extend the coloring $\chi$ to an $\Omega (1 / 100^{\log \Delta / \log \eta + 1}) \geq \Omega (1 / 100^{\sqrt{\log_2 \Delta}})$ proportion of the uncolored edges.
    We can repeat this $O(100^{\sqrt{\log_2 \Delta}} \log \lambda)$ times to extend the coloring $\chi$ to all of the edges in $U$.
    Each iteration of this process can be implemented in $O(m \cdot 2^{4 \sqrt{\log_2 \Delta}} \log^2 \Delta) \leq O(m \cdot 2^{6 \sqrt{\log_2 \Delta}})$ time, since we repeat this process for $O(100^{\sqrt{\log_2 \Delta}} \log \lambda) \leq O(2^{7\sqrt{\log_2 \Delta}} \log n)$ iterations, this takes $O(m \cdot 2^{13\sqrt{\log_2 \Delta}} \log n)$ time in total.
\end{proof}

We prove \Cref{thm:main:1} by applying \Cref{cor:extend} to the standard Euler partition framework \cite{gabow1985algorithms, sinnamon2019fast}. Given a graph $G$, we partition it into two edge-disjoint subgraphs $G_1$ and $G_2$ on the same vertex set such that $\Delta(G_i) \leq \lceil \Delta/2 \rceil$ for each $G_i$, where $\Delta(G_i)$ denotes the maximum degree of $G_i$. We then recursively compute a $(\Delta(G_i) + 1)$-coloring $\chi_i$ for each $G_i$. Combining $\chi_1$ and $\chi_2$, we obtain a $(\Delta + 3)$-coloring $\chi$ of $G$. We then uncolor the two smallest color classes in $\chi$, which contain $O(m / \Delta)$ edges, and apply \Cref{cor:extend} to recolor all of the uncolored edges in $\chi$ using only $\Delta + 1$ colors in $O(m \cdot 2^{13 \sqrt{\log_2 \Delta}}  \log n)$ time.

To show that the total running time of the algorithm is $O(m \cdot 2^{14 \sqrt{\log_2 \Delta}}  \log n)$,
first note that the depth of the recursion tree is $O(\log \Delta)$.
Next, consider the $i^{th}$ level of the recursion tree, for an arbitrary $i = O(\log \Delta)$: we have $2^i$ edge-disjoint subgraphs $G_1,\dots,G_{2^i}$ such that $\Delta(G_j) \leq O(\Delta/2^i)$ and $\sum_{j = 1}^{2^i} |E(G_j)| = m$.
Since the total running time at recursion level $i$ is 
$O(m \cdot 2^{13 \sqrt{\log_2 \Delta}} \log n)$
and the depth of the recursion tree is $O(\log \Delta)$, it follows that the total running time is $O(m \cdot 2^{13 \sqrt{\log_2 \Delta}} \log n \log \Delta) \leq O(m \cdot 2^{14 \sqrt{\log_2 \Delta}} \log n)$.

\section{The Algorithm $\Amplify$ (Proof of \Cref{lem:key})}\label{sec:proof key}

In this section, we describe and analyze the subroutine $\Amplify$, proving \Cref{lem:key}, which we restate below.

\sparsifytypes*

\subsection{Preliminaries for $\Amplify$}

We begin by giving some preliminaries and defining notations and subroutines that we use to describe the algorithm $\Amplify$.

\medskip
\noindent \textbf{The Subsets of Colors $C_1,\dots,C_{2\eta}$:}
Let $G$ be a graph, $\chi : E(G) \longrightarrow C \cup \{\bot\}$ be a partial $\mu$-coloring of $G$, $\U$ a separable collection with types in $C \times C$, and $\eta$ an integer such that $10 \leq \eta \leq \mu/10$.
By relabeling the colors, we can assume that $C = [\mu]$. 
Let $r := \lfloor \mu/ 2\eta \rfloor$, $C_i := [(i - 1)r + 1, ir]$ for each $i \in [2\eta]$, and $\C_k := C_{2k-1} \cup C_{2k}$ for each $k \in [\eta]$. Note that $|{\C_k}| = 2r \leq \mu/\eta$ for each $k \in [\eta]$. For each $i \in [2\eta]$ and $j \in [r]$, we denote the $j^{th}$ color in $C_i$ by $C_i(j) := (i-1)r + j$.
We can see that the subsets of colors $C_1, \dots, C_{2\eta} \subseteq [\mu]$ partition the set of colors $[2\eta r] \subseteq [\mu]$.

Our algorithm will modify $\chi$ by only changing the colors of edges $e \in E$ with $\chi(e) \in [q]$, where $q := 2\eta r$. 
Thus, our algorithm will only have the ability to change the types of u-fans $\f \in \U$ with $\tau(\f) \in [q] \times [q]$.\footnote{Recall that $[q] \times [q]$ denotes the set of \emph{unordered} subsets $\{\{c, c'\} \mid c, c' \in [q] \}$. }
To ensure that sufficiently many u-fans have types in $[q] \times [q]$, we perform a simple preprocessing step that involves relabeling the colors; we describe this in the proof of \Cref{lem:ignore last colors}.
After this preprocessing step, our algorithm constructs the subsets of colors $C_1,\dots,C_{2\eta}$ and removes all u-fans $\f$ from $\U$ that do \emph{not} have a type $\tau(\f) \in [q] \times [q]$. Thus, from this point onward, \emph{we assume that all u-fans in $\U$ have types in $\tau(\f) \in [q] \times [q]$}.
We later show that this operation only removes at most $2\lambda/5$ u-fans from $\U$.

\medskip
\noindent \textbf{Uniform and Aligned U-Fans:} Let $\f$ be a u-fan in $\U$. We say that the u-fan $\f$ is \textbf{uniform} if $\tau(\f) \subseteq C_i$ (i.e.~$\tau(\f) \in C_i \times C_i$) for some $i \in [2\eta]$, and we say that $\f$ is \textbf{aligned} if $\tau(\f)$ intersects both $C_{2k-1}$ and $C_{2k}$ (i.e.~$\tau(\f)\in C_{2k-1}\times C_{2k}$) for some $k \in [\eta]$. We refer to u-fans that are uniform or aligned as \textbf{social}. 
We let 
$\hat \U \subseteq \U$ denote the subset of social u-fans.
Note that for any social fan $\f \in \hat \U$, we have that $\tau(\f) \in \bigcup_{i=1}^\eta (\C_i \times \C_i)$. Thus, our objective is to ensure that a constant fraction of the u-fans in $\U$ are social, i.e.~that $|\hat \U| = \Omega(\lambda)$.

\medskip
\noindent \textbf{Modifying the Types of U-Fans:} The algorithm $\Amplify$ works by repeatedly modifying the types of u-fans to increase the number of social u-fans.
Suppose that the algorithm wants to modify the type of a u-fan $\f \notin \hat \U$ with $\tau(\f) \in C_{i} \times C_{i'}$ (by changing the colors of some edges) so that $\tau(\f) \in \C_k \times \C_k$ for some $k \in [\eta]$ after the modification. The algorithm does this by flipping some specific alternating paths, that we refer to as the \textbf{$k$-relevant} alternating paths of $\f$, which changes the type of $\f$ so that it is aligned. 
Suppose that $\f = (u,v,w,c_{\f}(u),c_{\f}(v),c_{\f}(w))$ where $c_{\f}(u) = C_{i}(j)$ and $c_{\f}(v) = c_{\f}(w) = C_{i'}(j')$ for $i,i' \in [2\eta]$ and $j, j' \in [r]$. Note that, since $\f$ is not social (and hence not uniform) we have that $i \neq i'$. In the case that $\{2k-1,2k\} \cap \{i,i'\} = \emptyset$, the $k$-relevant alternating paths of $\f$ are 
\begin{enumerate}
    \item The $\{C_{2k-1}(j), C_{i}(j)\}$-alternating path $P_u$ starting at $u$.
    \item The $\{C_{2k}(j'), C_{i'}(j')\}$-alternating paths $P_v$ and $P_w$ starting at $v$ and $w$ respectively.
\end{enumerate}
We let $\mathcal P_{k}(\f)$ denote the set of alternating-paths $\{P_u, P_v, P_w\}$.\footnote{If $P_v = P_w$, the set $\mathcal P_{k}(\f)$ only contains one copy of this path, i.e.~it is not a multiset.} If $i = 2k$ or $i' = 2k - 1$, we consider the $\{C_{2k}(j), C_{i}(j)\}$-alternating path starting at $u$ and the $\{C_{2k - 1}(j'), C_{i'}(j')\}$-alternating paths starting at $v$ and $w$ instead. Flipping the alternating paths in $\mathcal P_{k}(\f)$ by calling $\Flip(P)$ for each $P \in \mathcal P_{k}(\f)$ turns the type of the u-fan $\f$ into $\{C_{2k - 1}(j),C_{2k}(j')\} \in C_{2k-1} \times C_{2k} \subseteq \C_k \times \C_k$. We note  that flipping these paths might also change the types of at most $3$ other u-fans in the collection $\U$, \emph{potentially damaging these u-fans}.
\Cref{alg:modify} provides the pseudocode for the algorithm $\Compute$ that computes the set of $k$-relevant alternating paths of $\f$.

\begin{algorithm}[H]
    \SetAlgoLined
    \DontPrintSemicolon
    \SetKwRepeat{Do}{do}{while}
    \SetKwBlock{Loop}{repeat}{EndLoop}
    Let $\f = (u,v,w,c_{\f}(u),c_{\f}(v),c_{\f}(w))$\;
    Let $c_{\f}(u) = C_{i}(j)$ and $c_{\f}(v) = c_{\f}(w) = C_{i'}(j')$, where $i,i' \in [2\eta]$ and $j, j' \in [r]$\;
    Let $\ell \leftarrow 2k - 1$ and $\ell' \leftarrow 2k$\;
    \If{$i = 2k$ or $i' = 2k-1$}{
        Let $\ell \leftarrow 2k$ and $\ell' \leftarrow 2k - 1$\;
    }
    Let $P_u$ be the $\{C_{\ell}(j), C_{i}(j)\}$-alternating path starting at $u$\;
    Let $P_v$ and $P_w$ be the $\{C_{\ell'}(j'), C_{i'}(j')\}$-alternating path starting at $v$ and $w$ respectively\;
    \Return $\{P_u, P_v, P_w\}$
    \caption{$\Compute(\f, k)$}
    \label{alg:modify}
\end{algorithm}


For $i,i' \in [2\eta]$, let $\U_{i,i'} := \{\f \in \U \mid \tau(\f) \in C_i \times C_{i'}\} \subseteq \U$ denote the subset of u-fans with types in $C_i \times C_{i'}$. 
A key observation is that the $k$-relevant alternating paths corresponding to the u-fans in $\U_{i,i'} \setminus \hat \U$ are either edge-disjoint or the same. We summarize this property in the following claim.

\begin{claim}\label{cl:disjoint paths}
    For any u-fans $\f, \f' \in \U_{i,i'} \setminus \hat \U$, $P \in \mathcal P_k(\f)$ and $P' \in \mathcal P_k(\f')$, the alternating paths $P$ and $P'$ are either edge-disjoint or the same.
\end{claim}

For any subset $\mathcal V \subseteq \U_{i,i'}\setminus \hat \U$, let $\mathcal P_k(\mathcal V)$ denote the set $\bigcup_{\f \in \mathcal V} \mathcal P_k(\f)$ of all $k$-relevant alternating paths of the u-fans $\f \in \mathcal V$. It follows from \Cref{cl:disjoint paths} that all of the alternating paths in $\mathcal P_k(\mathcal V)$ are edge-disjoint.
Thus, we can flip them all \textbf{\emph{simultaneously}}
to modify the types of the u-fans in $\mathcal V$ so that they are all contained in $C_{2k - 1} \times C_{2k} \subseteq \C_k \times \C_k$, making them social.\footnote{When we say that we can flip these paths \emph{simultaneously}, we mean that we can flip them in any arbitrary order and still have the same effect on the coloring.} See \Cref{disjoint} for an illustration.

\begin{figure}
	\centering
	\begin{tikzpicture}[thick,scale=0.9]
	\draw (-1, 1) node(u1)[circle, draw, fill=black!50,
	inner sep=0pt, minimum width=6pt, label = above:$u_1$] {};
	\draw (0, 0) node(v1)[circle, draw, color=red, fill=black!50,
	inner sep=0pt, minimum width=6pt, label = below:$v_1$] {};
	\draw (-2, 0) node(w1)[circle, draw, color=red, fill=black!50,
	inner sep=0pt, minimum width=6pt, label = below:$w_1$] {};
	
	\draw [gray!50] plot [smooth cycle] coordinates {(-1, 1.8) (0.5, -0.3) (-1, -0.8) (-2.5, -0.3)};
	\node at (-1,2.2) {u-fan $\f_1$};
	
	\draw (1, -1) node(1)[circle, draw, fill=black!50,
	inner sep=0pt, minimum width=6pt] {};
	\draw (2, 0) node(2)[circle, draw, fill=black!50,
	inner sep=0pt, minimum width=6pt] {};
	\draw (4, 0) node(3)[circle, draw, fill=black!50,
	inner sep=0pt, minimum width=6pt] {};
	\draw (5, -1) node(4)[circle, draw, fill=black!50,
	inner sep=0pt, minimum width=6pt] {};
	\draw (6, 0) node(5)[circle, draw, fill=black!50,
	inner sep=0pt, minimum width=6pt] {};
	\draw (8, 0) node(6)[circle, draw, fill=black!50,
	inner sep=0pt, minimum width=6pt] {};
	\draw (9, -1) node(7)[circle, draw, fill=black!50,
	inner sep=0pt, minimum width=6pt] {};
	\draw (10, 0) node(8)[circle, draw, fill=black!50,
	inner sep=0pt, minimum width=6pt] {};
	\draw (12, 0) node(9)[circle, draw, fill=black!50,
	inner sep=0pt, minimum width=6pt] {};
	
	\draw [line width = 0.5mm, dashed] (u1) to (v1);
	\draw [line width = 0.5mm, dashed] (u1) to (w1);
	\draw [line width = 0.5mm, color=orange] (v1) to (1);
	\draw [line width = 0.5mm, color=red] (1) to (2);
	\draw [line width = 0.5mm, color=orange] (2) to (3);
	\draw [line width = 0.5mm, color=red] (3) to (4);
	\draw [line width = 0.5mm, color=orange] (4) to (5);
	\draw [line width = 0.5mm, color=red] (5) to (6);
	\draw [line width = 0.5mm, color=orange] (6) to (7);
	\draw [line width = 0.5mm, color=red] (7) to (8);
	\draw [line width = 0.5mm, color=orange] (8) to (9);
	
	\draw (-1, -3) node(u2)[circle, draw, fill=black!50,
	inner sep=0pt, minimum width=6pt, label = below:$u_2$] {};
	\draw (0, -2) node(v2)[circle, draw, color=cyan, fill=black!50,
	inner sep=0pt, minimum width=6pt, label = above:$v_2$] {};
	\draw (-2, -2) node(w2)[circle, draw, color=cyan, fill=black!50,
	inner sep=0pt, minimum width=6pt, label = above:$w_2$] {};
	
	\draw [gray!50] plot [smooth cycle] coordinates {(-1, -3.8) (0.5, -1.7) (-1, -1.2) (-2.5, -1.7)};
	\node at (-1,-4.2) {u-fan $\f_2$};
	
	\draw (2, -2) node(12)[circle, draw, fill=black!50,
	inner sep=0pt, minimum width=6pt] {};
	\draw (4, -2) node(13)[circle, draw, fill=black!50,
	inner sep=0pt, minimum width=6pt] {};
	\draw (6, -2) node(15)[circle, draw, fill=black!50,
	inner sep=0pt, minimum width=6pt] {};
	\draw (8, -2) node(16)[circle, draw, fill=black!50,
	inner sep=0pt, minimum width=6pt] {};
	\draw (10, -2) node(18)[circle, draw, fill=black!50,
	inner sep=0pt, minimum width=6pt] {};
	\draw (12, -2) node(19)[circle, draw, fill=black!50,
	inner sep=0pt, minimum width=6pt] {};
	
	\draw [line width = 0.5mm, dashed] (u2) to (v2);
	\draw [line width = 0.5mm, dashed] (u2) to (w2);
	\draw [line width = 0.5mm, color=Emerald] (v2) to (1);
	\draw [line width = 0.5mm, color=cyan] (1) to (12);
	\draw [line width = 0.5mm, color=Emerald] (12) to (13);
	\draw [line width = 0.5mm, color=cyan] (13) to (4);
	\draw [line width = 0.5mm, color=Emerald] (4) to (15);
	\draw [line width = 0.5mm, color=cyan] (15) to (16);
	\draw [line width = 0.5mm, color=Emerald] (16) to (7);
	\draw [line width = 0.5mm, color=cyan] (7) to (18);
	\draw [line width = 0.5mm, color=Emerald] (18) to (19);
	
\end{tikzpicture}
	\caption{In this picture we have two different u-fans $\f_1 = (u_1, v_1, w_1, \cdot\,, {\color{red}C_i(j_1)}, {\color{red}C_i(j_1)})$ and $\f_2 = (u_2, v_2, w_2, \cdot \,, {\color{cyan}C_{i}(j_2)}, {\color{cyan}C_{i}(j_2)})$. When $j_1\neq j_2$, we have $\{{\color{red}C_i(j_1)}, {\color{orange}C_{2k}(j_1)}\}\cap \{{\color{cyan}C_{i}(j_2)}, {\color{Emerald}C_{2k}(j_2)}\} = \emptyset$. Then two $k$-relevant alternating paths from $v_1, v_2$ of type-$\{{\color{red}C_i(j_1)}, {\color{orange}C_{2k}(j_1)}\}$ and type-$\{{\color{cyan}C_{i}(j_2)}, {\color{Emerald}C_{2k}(j_2)}\}$ are edge-disjoint.}
	\label{disjoint}
\end{figure}
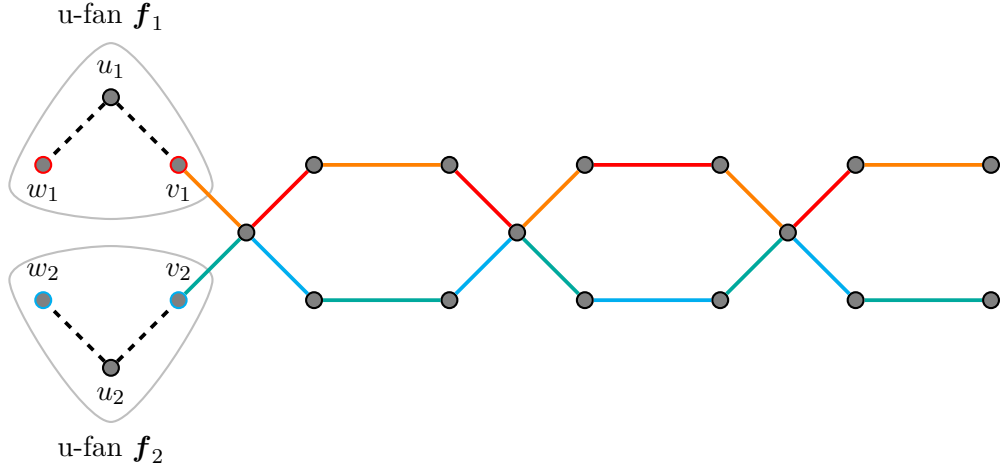

We refer to a subset of u-fans $\mathcal V \subseteq \U$ such that $\mathcal V \subseteq \U_{i,i'}$ for some $i,i' \in [2\eta]$ as a \textbf{batch}.
The subroutine $\ModifyB$ modifies the types of a batch of u-fans $\mathcal V \subseteq \U \setminus \hat \U$, which we describe below in \Cref{alg:modifyB}.


\begin{algorithm}[H]
    \SetAlgoLined
    \DontPrintSemicolon
    \SetKwRepeat{Do}{do}{while}
    \SetKwBlock{Loop}{repeat}{EndLoop}
    \SetKwInOut{Input}{input}
    \Input{A batch $\mathcal V \subseteq \U_{i,i'}\setminus \hat \U$, for some $i,i' \in [2\eta]$}
    Let $\mathcal P \leftarrow \mathcal P_k(\mathcal V)$ be the set of $k$-relevant alternating paths of the u-fans in $\mathcal V$\;  
    \For{$P \in \mathcal P$}{
        Call $\Flip(P)$\;
    }
    \caption{$\ModifyB(\mathcal V, k)$}
    \label{alg:modifyB}
\end{algorithm}

\medskip
\noindent \textbf{Good U-Fans:} 
Whenever we modify the types of some subset of u-fans to make these u-fans social, we want to make sure that we are increasing the total number of social u-fans. Since applying $\ModifyB(\mathcal V, k)$ to some batch $\mathcal V \subseteq \U_{i,i'} \setminus \hat \U$ can change the types of up to $3 |\mathcal V|$ u-fans that are not contained in $\mathcal V$---potentially damaging them---calling $\ModifyB(\mathcal V, k)$ for an arbitrary batch $\mathcal V \subseteq \U_{i,i'}\setminus \hat \U$ might be counterproductive.

To this end, we refer to a u-fan $\f \in \U$ as being \textbf{$k$-good} if it is not social and making a call to $\ModifyB(\{\f\}, k)$ does not change the type of any u-fan already in $\hat \U$.
We say that a u-fan is \textbf{$k$-bad} if it is not $k$-good, and denote the set of $k$-bad u-fans by $\mathcal B_k \subseteq \U$.
We can observe that calling $\ModifyB(\{\f\}, k)$ for a $k$-good u-fan increases the size of the set $\hat \U$ by at least $1$.

\subsection{The Algorithm $\Amplify$}

The algorithm $\Amplify$ works in \textbf{iterations}. During each iteration, the algorithm finds a large batch $\mathcal V$ of good u-fans, modifies their types to make them social by making a call to $\ModifyB$, and then removes the damaged u-fans from $\U$.
Once the number of social u-fans is at least $\lambda / 100$, the algorithm terminates. 

The pseudocode in \Cref{alg:amplify} gives a formal description of the algorithm.
We note that, throughout the run of \Cref{alg:amplify}, the separable collection $\U$ is updated during the calls to $\ModifyB$ (where we change the missing colors assigned to u-fans) in \Cref{line:up sep 1} and when we remove damaged u-fans from $\U$ in \Cref{line:bye damaged}.
Whenever we refer to $\U$ or any subset of $\U$ (such as $\hat \U$, $\U_{i,i'}$ and $\mathcal B_k$) in \Cref{alg:amplify}, these are defined with respect to the \emph{current state} of the separable collection $\U$. In \Cref{sec:amp imp}, we show how to compute these subsets efficiently.

\begin{algorithm}[H]
    \SetAlgoLined
    \DontPrintSemicolon
    \SetKwRepeat{Do}{do}{while}
    \SetKwBlock{Loop}{repeat}{EndLoop}
    Set $\U \leftarrow \{\f \in \U \mid \tau(\f) \in [q] \times [q]\}$\label{line:remove types}\;
    \While{$|\hat \U| < \lambda / 100$}{
        Let $k^\star \leftarrow \arg \min_{k \in [\eta]} |{\U_{2k-1,2k}} \cup {\U_{2k-1,2k-1}}\cup {\U_{2k,2k}}|$\;
        Let $i^\star ,i'^\star \leftarrow \arg \max_{1 \leq i < i' \leq 2\eta} |{\U_{i,i'}} \setminus \mathcal B_{k^\star}|$\;
        Let $\mathcal V \leftarrow {\U_{i^\star,i'^\star}} \setminus\mathcal B_{k^\star}$\;
        Call $\ModifyB(\mathcal V,k^\star)$\label{line:up sep 1}\; 
        Remove any damaged u-fans from $\U$\label{line:bye damaged}\;
    }
    Set $\U \leftarrow \hat \U$\label{line:end}\;
    \caption{$\Amplify(\U)$}
    \label{alg:amplify}
\end{algorithm}


In \Cref{sec:amp anal}, we analyze the algorithm $\Amplify$ and show that, after making a call to $\Amplify(\U)$, the coloring $\chi$ and separable collection $\U$ satisfy the properties described in \Cref{lem:key}. In \Cref{sec:amp imp}, we show how to implement the algorithm $\Amplify$ to run in time $O(m \eta^4 \log \mu)$.

\subsection{Analysis of $\Amplify$}\label{sec:amp anal}


We begin by showing that at least half of the u-fans in $\U$ initially have types in $[q] \times [q]$.

\begin{claim}\label{lem:ignore last colors}
    By relabeling the colors (i.e.~setting $\chi = \chi \circ \pi$ for a permutation $\pi : [\mu] \longrightarrow [\mu]$), we have that at least $3\lambda/5$ of the u-fans in $\U$ have types in $[q] \times [q]$.
\end{claim}

\begin{proof}
    Given some u-fan $\f \in \U$ and $c \in [\mu]$, let $X_c^{\f}$ denote the indicator for the event that $c \in \tau(\f)$, i.e.~that the type of $\f$ contains the color $c$. Furthermore, let $X_c = \sum_{\f \in \U} X_c^{\f}$. Since the type of each (non-damaged) u-fan contains exactly $2$ colors, we can see that $\sum_{c = 1}^\mu X_c \leq 2 |{\U}| = 2\lambda$.
    By \emph{relabeling the colors}, we can assume w.l.o.g. that $X_1 \geq \dots \geq X_\mu$.
    It follows from this assumption that 
    $$\sum_{c = 1}^q X_c \geq \frac{q}{\mu} \cdot \sum_{c = 1}^\mu X_c.$$ Thus, we get that
    $$ \sum_{c = q + 1}^\mu X_c \leq \left(1 - \frac{q}{\mu} \right) \cdot \sum_{c = 1}^\mu X_c \leq 2 \left(1 - \frac{q}{\mu} \right) \cdot \lambda \leq \frac{2}{5} \cdot \lambda. $$
    Recalling that $\mu \ge 10\eta$, the last inequality follows from $q = 2\eta \lfloor \mu/2\eta \rfloor \geq 2\eta (\mu/2\eta - 1) \geq \mu - 2\eta \geq (4/5) \mu$. In other words,
     at most $2\lambda/5$ of the u-fans in $\U$ have types not in $[q] \times [q]$.
\end{proof}

The following lemma describes how the sizes of the sets $\U$ and $\hat \U$ change during each iteration of the algorithm.

\begin{lemma}\label{lem:augment good:2}
    Consider some iteration of \Cref{alg:amplify} where we call $\ModifyB(\mathcal V, k^\star)$ for a batch of u-fans $\mathcal V$. During this iteration, the size of $\hat \U$ increases by $|\mathcal V|$ and the size of $\U$ decreases by at most $3 |\mathcal V|$.
\end{lemma}

\begin{proof}
    Since none of the u-fans in $\mathcal V$ are $k$-bad, we know that calling $\ModifyB(\{\f\}, k)$ for any $\f \in \mathcal V$ will not remove any u-fans from $\hat \U$ and will add the u-fan $\f$ to $\hat \U$ after changing its type. Thus, calling $\ModifyB(\mathcal V, k)$, which flips all of the $k$-relevant alternating paths corresponding to the u-fans in $\mathcal V$, adds each u-fan in $\mathcal V$ to $\hat \U$ without removing any u-fans from $\hat \U$. However, flipping the $k$-relevant alternating paths corresponding to a u-fan can damage up to $3$ u-fans in $\U$. Thus, calling $\ModifyB(\mathcal V, k)$ can damage up to $3|\mathcal V|$ u-fans in $\U$, which are removed from $\U$ during \Cref{line:bye damaged}. Thus, $|\hat \U|$ increases by $|\mathcal V|$ and $|{\U}|$ decreases by at most $3|\mathcal V|$.
\end{proof}

We can now prove the following claim, which lower bounds the size of $\U$ throughout the run of the algorithm.

\begin{claim}\label{claim:U>=1/2}
    At the start of each iteration of \Cref{alg:amplify}, we have that $|{\U}| \geq \lambda/2$.
\end{claim}

\begin{proof}
    It follows from \Cref{lem:ignore last colors} that \Cref{line:remove types} of \Cref{alg:amplify} only removes at most $2\lambda/5$ u-fans from $\U$. Thus, at the start of the first iteration, we have that that $|{\U}| \geq 3\lambda/5$. It follows from \Cref{lem:augment good:2} that, if $|\hat \U|$ increases by $\Phi$ during some iteration, then $|{\U}|$ decreases by at most $3 \Phi$ during the same iteration. Thus, at the start of each iteration, we have that $|{\U}| \geq 3\lambda/5 - 3 |\hat \U|$. Since the algorithm terminates once $|\hat \U| \geq \lambda/100$, it follows that $|{\U}| \geq 3\lambda/5 - 3\lambda/100 \geq \lambda/2$.
\end{proof}

The following lemma lower bounds the size of the batch of good edges $\mathcal V$ that we construct during each iteration and is the main technical component in the analysis of this algorithm.

\begin{lemma}\label{lem:find good}
    During each iteration of \Cref{alg:amplify}, we construct a batch $\mathcal V \subseteq \U_{i,i'}$ of $k$-good u-fans with size $|\mathcal V| \geq \lambda/(4\eta^2)$, for some $i,i' \in [2\eta]$ s.t.~$i < i'$ and $k \in [\eta]$.
\end{lemma}

\begin{proof}
    Consider the state of the algorithm at the start of some iteration. We now show that the set $\mathcal U_{i^\star, i'^\star} \setminus \mathcal B_{k^\star}$ constructed during the iteration has size at least $\lambda/(4\eta^2)$. 
    
    We can observe that the subsets $\{\U_{i,i'}\}_{1 \leq i < i' \leq \eta} \cup \{\U_{i,i}\}_{i \in [\eta]}$ partition the set $\U$. Since, for any $i \in [2\eta]$, $\U_{i,i} \subseteq \hat \U \subseteq \mathcal B_{k^\star}$, it follows that
    \begin{equation}\label{eq:union}
        \U \setminus \mathcal B_{k^\star} = \bigcup_{1 \leq i < i' \leq \eta} (\mathcal U_{i,i'} \setminus \mathcal B_{k^\star}).
    \end{equation}
    By a simple averaging argument and \Cref{eq:union}, we get that
    \begin{equation}\label{eq:union2}
        |\mathcal U_{i^\star, i'^\star} \setminus \mathcal B_{k^\star}| \geq \frac{1}{\eta^2} \cdot \sum_{1 \leq i < i' \leq \eta} |\mathcal U_{i,i'} \setminus \mathcal B_{k^\star}| = \frac{1}{\eta^2} \cdot |\mathcal U \setminus \mathcal B_{k^\star}|.
    \end{equation}
    Now, for each $k \in [\eta]$, let $\hat \U_k := {\U_{2k-1,2k}} \cup {\U_{2k-1,2k-1}}\cup {\U_{2k,2k}}$. Then we have that $\hat \U = \bigcup_{k = 1}^\eta \hat \U_k$. Since the sets $\{\hat \U_k\}_{k \in [\eta]}$ are all disjoint, it follows that $\sum_{k=1}^\eta |\hat \U_k| = |\hat \U| < \lambda/100$ and thus $|\hat \U_{k^\star}| \leq \lambda/(100 \eta)$ by a simple averaging argument.
    We now prove the following claim, which upper bounds the number of $k$-bad u-fans in terms of $|\hat \U|$ and $|\hat \U_k|$.

    \begin{claim}\label{claim:bad count}
        For each $k \in [\eta]$, we have that $|\mathcal B_k| \leq 7 \cdot |\hat \U| + 6\eta \cdot |\hat \U_{k}|$.
    \end{claim}

    \begin{proof}[Proof of claim]
        We first note that, for each non-social $k$-bad u-fan $\f \in \mathcal B_{k}$, there is some social u-fan $\f' \in \hat \U$ that is \emph{responsible} for the u-fan $\f$ being $k$-bad, meaning that calling $\ModifyB(\{\f\} ,k)$ changes the type of $\f'$. Thus, to upper bound the size of $\mathcal B_{k}$, it suffices to upper bound the number of non-social u-fans for which the social u-fans are responsible.

        To this end, let $\f' \in \hat \U$. If $\f' \in \hat \U \setminus \hat \U_k$, then $\f'$ is responsible for at most $6$ $k$-bad non-social u-fans. To see why, note that there are at most $2$ $k$-relevant alternating paths
        ending at each vertex of $\f'$ (at most one with a color in $C_{2k-1}$ and one with a color in $C_{2k}$) such that flipping these paths would remove $\f'$ from $\hat \U$, and for each of these alternating paths there is at most one non-social u-fan $\f$ such that calling $\ModifyB(\{\f\}, k)$ flips this path. 
        More specifically, given a vertex $x \in \f'$ such that $c_{\f'}(x) = C_{\ell}(j)$, for some $\ell \notin \{2k-1,2k\}$, the type of any such $k$-relevant alternating path ending at $x$ is $\{C_{\ell}(j), C_i(j)\}$, for some $i \in \{2k-1, 2k\} \}$.
        See \Cref{socialize-damage1} for an illustration.
        
        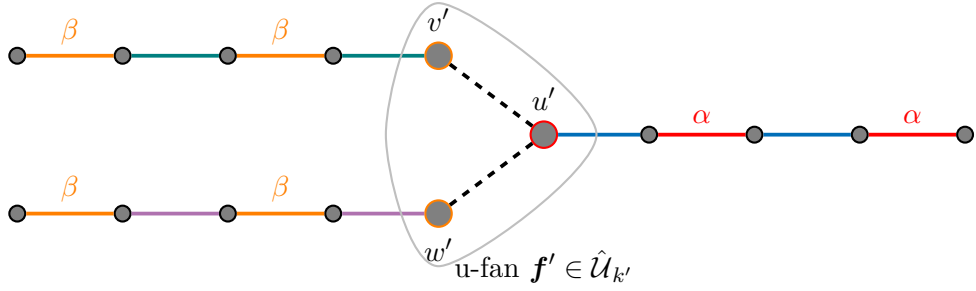
\begin{figure}
        	\centering
        	\begin{tikzpicture}[thick,scale=0.7]
	\draw (7, -1.5) node(0)[circle, draw, color = red, fill=black!50,
	inner sep=0pt, minimum width=10pt, label = $u'$] {};
	\draw (5, 0) node(1)[circle, draw, color=orange, fill=black!50,
	inner sep=0pt, minimum width=10pt, label = $v'$] {};
	\draw (5, -3) node(2)[circle, draw, color=orange, fill=black!50, inner sep=0pt, minimum width=10pt, label = -90:{$w'$}] {};

	\draw (3, 0) node(23)[circle, draw, fill=black!50,
	inner sep=0pt, minimum width=6pt] {};
	\draw (1, 0) node(24)[circle, draw, fill=black!50,
	inner sep=0pt, minimum width=6pt] {};
	\draw (-1, 0) node(25)[circle, draw, fill=black!50,
	inner sep=0pt, minimum width=6pt] {};
	\draw (-3, 0) node(26)[circle, draw, fill=black!50, 
	inner sep=0pt, minimum width=6pt] {};	
		
	\draw [line width = 0.5mm, color=teal] (1) to (23);
	\draw [line width = 0.5mm, color=orange] (23) to node[above] {$\beta$} (24);
	\draw [line width = 0.5mm, color=teal] (24) to (25);
	\draw [line width = 0.5mm, color=orange] (25) to node[above] {$\beta$} (26);

	\draw (3, -3) node(43)[circle, draw, fill=black!50,
	inner sep=0pt, minimum width=6pt] {};
	\draw (1, -3) node(44)[circle, draw, fill=black!50,
	inner sep=0pt, minimum width=6pt] {};
	\draw (-1, -3) node(45)[circle, draw, fill=black!50,
	inner sep=0pt, minimum width=6pt] {};
	\draw (-3, -3) node(46)[circle, draw, fill=black!50, 
	inner sep=0pt, minimum width=6pt] {};

	\draw [line width = 0.5mm, color=Orchid] (2) to (43);
	\draw [line width = 0.5mm, color=orange] (43) to node[above] {$\beta$} (44);
	\draw [line width = 0.5mm, color=Orchid] (44) to (45);
	\draw [line width = 0.5mm, color=orange] (45) to node[above] {$\beta$} (46);

	\draw (9, -1.5) node(73)[circle, draw, fill=black!50,
	inner sep=0pt, minimum width=6pt] {};
	\draw (11, -1.5) node(74)[circle, draw, fill=black!50,
	inner sep=0pt, minimum width=6pt] {};
	\draw (13, -1.5) node(75)[circle, draw, fill=black!50,
	inner sep=0pt, minimum width=6pt] {};
	\draw (15, -1.5) node(76)[circle, draw, fill=black!50, 
	inner sep=0pt, minimum width=6pt] {};
		
	\draw [line width = 0.5mm, color=NavyBlue] (0) to (73);
	\draw [line width = 0.5mm, color=red] (73) to node[above] {$\alpha$} (74);
	\draw [line width = 0.5mm, color=NavyBlue] (74) to (75);
	\draw [line width = 0.5mm, color=red] (75) to node[above] {$\alpha$} (76);
	
	\draw [line width = 0.5mm, dashed] (0) to (1);
	\draw [line width = 0.5mm, dashed] (0) to (2);
	
	\draw [gray!50] plot [smooth cycle] coordinates {(8, -1.5) (5, 1) (4, -1.5) (5, -4)};
	\node at (7, -4) {u-fan $\f'\in \hat{\U}_{k'}$};
	
\end{tikzpicture}
        	\caption{In this example, $\f' = (u', v', w', {\color{red}\alpha}, {\color{orange}\beta}, {\color{orange}\beta})\in \hat{\U}_{k'}$, where $k'\neq k$ and ${\color{red}\alpha} = C_{2k'-1}(j_1), {\color{orange}\beta} = C_{2k'}(j_2)$. Then, for any $i\in \{2k-1, 2k\}$,  $u'$ could be the endpoint of a $k$-relevant path of type $\{{\color{red}\alpha} = C_{2k'-1}(j_1), C_i(j_1)\}$, and $v', w'$ could be endpoints of $k$-relevant paths of types $\{{\color{orange}\beta} = C_{2k'}(j_2), C_{i}(j_2)\}$. Thus, the total number of $k$-bad u-fans with respect to $\f'$ is at most $6$.}\label{socialize-damage1}
        \end{figure}

        On the other hand, if $\f' \in {\hat \U}_k$, then $\f'$ is responsible for at most $6\eta$ 
        $k$-bad u-fans. To see why, note that there are at most $2\eta - 2$ $k$-relevant alternating paths 
        ending at each vertex of $\f'$ (for each $i \in [2\eta] \setminus \{2k-1, 2k\}$, there is at most one such alternating path with a color in $C_i$) such that flipping these paths would remove $\f'$ from $\hat \U$, and for each of these alternating paths there is at most one non-social u-fan $\f$ such that calling $\ModifyB(\{\f\}, k)$ flips this path.
        More specifically, given a vertex $x \in \f'$ such that $c_{\f'}(x) = C_{\ell}(j)$, for some $\ell \in \{2k-1,2k\}$, the type of any such $k$-relevant alternating path ending at $x$ is $\{C_{\ell}(j), C_i(j)\}$, for some $i \in [2\eta] \setminus \{2k-1, 2k\} \}$.
        See \Cref{socialize-damage2} for an illustration.
        
        \begin{figure}
        	\centering
        	\begin{tikzpicture}[thick,scale=0.7]
	\draw (7, -1.5) node(0)[circle, draw, color = red, fill=black!50,
	inner sep=0pt, minimum width=10pt, label = $u'$] {};
	\draw (5, 0) node(1)[circle, draw, color=orange, fill=black!50,
	inner sep=0pt, minimum width=10pt, label = $v'$] {};
	\draw (5, -3) node(2)[circle, draw, color=orange, fill=black!50, inner sep=0pt, minimum width=10pt, label = -90:{$w'$}] {};	
	
	\draw (3, -1) node(3)[circle, draw, fill=black!50,
	inner sep=0pt, minimum width=6pt] {};
	\draw (1, -1) node(4)[circle, draw, fill=black!50,
	inner sep=0pt, minimum width=6pt] {};
	\draw (-1, -1) node(5)[circle, draw, fill=black!50,
	inner sep=0pt, minimum width=6pt] {};
	\draw (-3, -1) node(6)[circle, draw, fill=black!50, 
	inner sep=0pt, minimum width=6pt] {};
	
	\draw (3, 1) node(13)[circle, draw, fill=black!50,
	inner sep=0pt, minimum width=6pt] {};
	\draw (1, 1) node(14)[circle, draw, fill=black!50,
	inner sep=0pt, minimum width=6pt] {};
	\draw (-1, 1) node(15)[circle, draw, fill=black!50,
	inner sep=0pt, minimum width=6pt] {};
	\draw (-3, 1) node(16)[circle, draw, fill=black!50, 
	inner sep=0pt, minimum width=6pt] {};
	
	\draw (3, 0) node(23)[circle, draw, fill=black!50,
	inner sep=0pt, minimum width=6pt] {};
	\draw (1, 0) node(24)[circle, draw, fill=black!50,
	inner sep=0pt, minimum width=6pt] {};
	\draw (-1, 0) node(25)[circle, draw, fill=black!50,
	inner sep=0pt, minimum width=6pt] {};
	\draw (-3, 0) node(26)[circle, draw, fill=black!50, 
	inner sep=0pt, minimum width=6pt] {};	
	
	\draw [line width = 0.5mm, color=olive] (1) to (3);
	\draw [line width = 0.5mm, color=orange] (3) to node[above] {$\beta$} (4);
	\draw [line width = 0.5mm, color=olive] (4) to (5);
	\draw [line width = 0.5mm, color=orange] (5) to node[above] {$\beta$} (6);
	\draw [line width = 0.5mm, color=olive] (6) to (-5, -1);
	
	\draw [line width = 0.5mm, color=teal] (1) to (23);
	\draw [line width = 0.5mm, color=orange] (23) to node[above] {$\beta$} (24);
	\draw [line width = 0.5mm, color=teal] (24) to (25);
	\draw [line width = 0.5mm, color=orange] (25) to node[above] {$\beta$} (26);
	\draw [line width = 0.5mm, color=teal] (26) to (-5, 0);
	
	\draw [line width = 0.5mm, color=lime] (1) to (13);
	\draw [line width = 0.5mm, color=orange] (13) to node[above] {$\beta$} (14);
	\draw [line width = 0.5mm, color=lime] (14) to (15);
	\draw [line width = 0.5mm, color=orange] (15) to node[above] {$\beta$} (16);
	\draw [line width = 0.5mm, color=lime] (16) to (-5, 1);
	
	\draw (3, -2) node(33)[circle, draw, fill=black!50,
	inner sep=0pt, minimum width=6pt] {};
	\draw (1, -2) node(34)[circle, draw, fill=black!50,
	inner sep=0pt, minimum width=6pt] {};
	\draw (-1, -2) node(35)[circle, draw, fill=black!50,
	inner sep=0pt, minimum width=6pt] {};
	\draw (-3, -2) node(36)[circle, draw, fill=black!50, 
	inner sep=0pt, minimum width=6pt] {};
	
	\draw (3, -3) node(43)[circle, draw, fill=black!50,
	inner sep=0pt, minimum width=6pt] {};
	\draw (1, -3) node(44)[circle, draw, fill=black!50,
	inner sep=0pt, minimum width=6pt] {};
	\draw (-1, -3) node(45)[circle, draw, fill=black!50,
	inner sep=0pt, minimum width=6pt] {};
	\draw (-3, -3) node(46)[circle, draw, fill=black!50, 
	inner sep=0pt, minimum width=6pt] {};
	
	\draw (3, -4) node(53)[circle, draw, fill=black!50,
	inner sep=0pt, minimum width=6pt] {};
	\draw (1, -4) node(54)[circle, draw, fill=black!50,
	inner sep=0pt, minimum width=6pt] {};
	\draw (-1, -4) node(55)[circle, draw, fill=black!50,
	inner sep=0pt, minimum width=6pt] {};
	\draw (-3, -4) node(56)[circle, draw, fill=black!50, 
	inner sep=0pt, minimum width=6pt] {};
	
	\draw [line width = 0.5mm, color=Lavender] (2) to (33);
	\draw [line width = 0.5mm, color=orange] (33) to node[above] {$\beta$} (34);
	\draw [line width = 0.5mm, color=Lavender] (34) to (35);
	\draw [line width = 0.5mm, color=orange] (35) to node[above] {$\beta$} (36);
	\draw [line width = 0.5mm, color=Lavender] (36) to (-5, -2);
	
	\draw [line width = 0.5mm, color=Orchid] (2) to (43);
	\draw [line width = 0.5mm, color=orange] (43) to node[above] {$\beta$} (44);
	\draw [line width = 0.5mm, color=Orchid] (44) to (45);
	\draw [line width = 0.5mm, color=orange] (45) to node[above] {$\beta$} (46);
	\draw [line width = 0.5mm, color=Orchid] (46) to (-5, -3);
	
	\draw [line width = 0.5mm, color=Plum] (2) to (53);
	\draw [line width = 0.5mm, color=orange] (53) to node[above] {$\beta$} (54);
	\draw [line width = 0.5mm, color=Plum] (54) to (55);
	\draw [line width = 0.5mm, color=orange] (55) to node[above] {$\beta$} (56);
	\draw [line width = 0.5mm, color=Plum] (56) to (-5, -4);
	
	\draw (9, -0.5) node(63)[circle, draw, fill=black!50,
	inner sep=0pt, minimum width=6pt] {};
	\draw (11, -0.5) node(64)[circle, draw, fill=black!50,
	inner sep=0pt, minimum width=6pt] {};
	\draw (13, -0.5) node(65)[circle, draw, fill=black!50,
	inner sep=0pt, minimum width=6pt] {};
	\draw (15, -0.5) node(66)[circle, draw, fill=black!50, 
	inner sep=0pt, minimum width=6pt] {};
	
	\draw (9, -1.5) node(73)[circle, draw, fill=black!50,
	inner sep=0pt, minimum width=6pt] {};
	\draw (11, -1.5) node(74)[circle, draw, fill=black!50,
	inner sep=0pt, minimum width=6pt] {};
	\draw (13, -1.5) node(75)[circle, draw, fill=black!50,
	inner sep=0pt, minimum width=6pt] {};
	\draw (15, -1.5) node(76)[circle, draw, fill=black!50, 
	inner sep=0pt, minimum width=6pt] {};
	
	\draw (9, -2.5) node(83)[circle, draw, fill=black!50,
	inner sep=0pt, minimum width=6pt] {};
	\draw (11, -2.5) node(84)[circle, draw, fill=black!50,
	inner sep=0pt, minimum width=6pt] {};
	\draw (13, -2.5) node(85)[circle, draw, fill=black!50,
	inner sep=0pt, minimum width=6pt] {};
	\draw (15, -2.5) node(86)[circle, draw, fill=black!50, 
	inner sep=0pt, minimum width=6pt] {};
	
	\draw [line width = 0.5mm, color=cyan] (0) to (63);
	\draw [line width = 0.5mm, color=red] (63) to node[above] {$\alpha$} (64);
	\draw [line width = 0.5mm, color=cyan] (64) to (65);
	\draw [line width = 0.5mm, color=red] (65) to node[above] {$\alpha$} (66);
	\draw [line width = 0.5mm, color=cyan] (66) to (17, -0.5);

	\draw [line width = 0.5mm, color=NavyBlue] (0) to (73);
	\draw [line width = 0.5mm, color=red] (73) to node[above] {$\alpha$} (74);
	\draw [line width = 0.5mm, color=NavyBlue] (74) to (75);
	\draw [line width = 0.5mm, color=red] (75) to node[above] {$\alpha$} (76);
	\draw [line width = 0.5mm, color=NavyBlue] (76) to (17, -1.5);
	
	\draw [line width = 0.5mm, color=Blue] (0) to (83);
	\draw [line width = 0.5mm, color=red] (83) to node[above] {$\alpha$} (84);
	\draw [line width = 0.5mm, color=Blue] (84) to (85);
	\draw [line width = 0.5mm, color=red] (85) to node[above] {$\alpha$} (86);
	\draw [line width = 0.5mm, color=Blue] (86) to (17, -2.5);
	
	\draw [line width = 0.5mm, dashed] (0) to (1);
	\draw [line width = 0.5mm, dashed] (0) to (2);
	
	\draw [gray!50] plot [smooth cycle] coordinates {(8, -1.5) (5, 1) (4, -1.5) (5, -4)};
	\node at (7, -4) {u-fan $\f'\in \hat{\U}_k$};
	
\end{tikzpicture}
        	\caption{In this example, $\f' = (u', v', w', {\color{red}\alpha}, {\color{orange}\beta}, {\color{orange}\beta})\in \hat{\U}_k$, where ${\color{red}\alpha} = C_{2k-1}(j_1), {\color{orange}\beta} = C_{2k}(j_2)$. Then, for any $i\in [2\eta]\setminus \{2k-1, 2k\}$,  $u'$ could be the endpoint of a $k$-relevant path of type $\{{\color{red}\alpha}= C_{2k-1}(j_1), C_i(j_1)\}$, and $v', w'$ could be endpoints of $k$-relevant paths of types $\{{\color{orange}\beta} = C_{2k}(j_2), C_{i}(j_2)\}$. Thus, the total number of $k$-bad u-fans with respect to $\f'$ is at most $6\eta-6$.}\label{socialize-damage2}
        \end{figure}
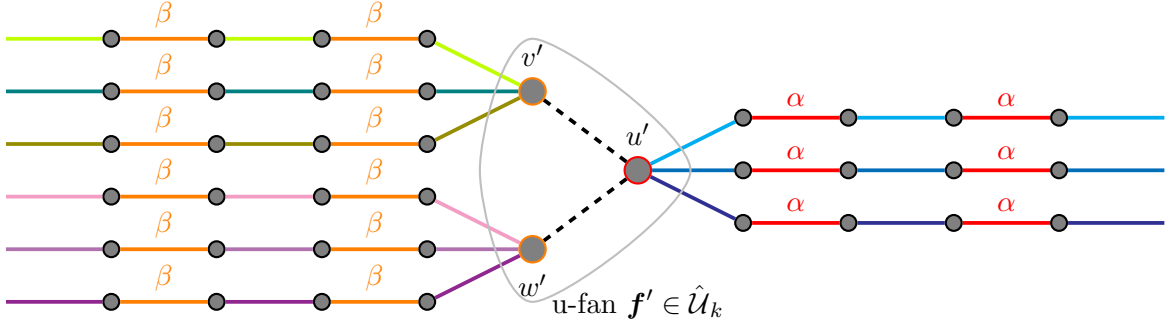
        
        Thus, accounting for the social u-fans, the total number of $k$-bad u-fans is at most 
        $$6 \cdot |\hat \U \setminus \hat {\mathcal U}_k| + 6\eta \cdot |\hat {\mathcal U}_k| + |\hat \U | \leq 7 \cdot |{\hat \U}| + 6\eta \cdot |\hat{ \mathcal U}_k|.\qedhere $$ 
    \end{proof}

    It follows from \Cref{claim:bad count} that $|\mathcal B_{k^\star}| \leq 7 \cdot \lambda/100 + 6\eta \cdot \lambda/(100\eta) \leq \lambda/4$. Since, by \Cref{claim:U>=1/2}, the size of $\mathcal U$ is at least $\lambda/2$, it follows from \Cref{eq:union2} and the upper bound on $|\mathcal B_{k^\star}|$ that
    $$ |\mathcal U_{i^\star, i'^\star} \setminus \mathcal B_{k^\star}| \geq \frac{1}{\eta^2} \cdot |\mathcal U \setminus \mathcal B_{k^\star}| \geq \frac{1}{\eta^2} \cdot( |\mathcal U| - |\mathcal B_{k^\star}|) \geq \frac{1}{\eta^2} \cdot \left( \frac{\lambda}{2} - \frac{\lambda}{4} \right) \geq \frac{\lambda}{4\eta^2}.\qedhere  $$
\end{proof}

Combining \Cref{lem:find good,lem:augment good:2}, we get the following corollary, which upper bounds the number of iterations performed by the algorithm.

\begin{corollary}
    During each iteration of \Cref{alg:amplify}, the size of $\hat \U$ increases by at least $\lambda/(4\eta^2)$.
\end{corollary}

Thus, after $O(\eta^2)$ many iterations, we have that $|\hat \U| \geq \lambda/100$, so the algorithm sets $\U \leftarrow \hat \U$ and terminates. Since $\hat \U$ is a subset of a separable collection, it must be a separable collection itself. Furthermore, it follows from the definition of $\hat \U$ that the type of every u-fan in $\hat \U$ is contained in the subset of types $\bigcup_{k=1}^\eta \C_k \times \C_k$.

\subsection{Implementation of $\Amplify$}\label{sec:amp imp}

In this section, we show how to implement each iteration of \Cref{alg:amplify} in $O(m \eta^2 \log \mu)$ time. By the analysis in \Cref{sec:amp anal}, the algorithm runs for $O(\eta^2)$ iterations. Hence, that would imply that the running time of  $\Amplify$ is $O(m \eta^4 \log \mu)$. 

We first note that our algorithm maintains the separable collection $\U$ explicitly.
Thus, we can construct the separable collections $\{\hat \U_k\}_{k \in [\eta]}$ by initializing an empty separable collection (using our data structures) for each $\hat \U_k$ and then scanning through each of the u-fans in $\U$, adding them to the appropriate collection in $\{\hat \U_k\}_{k \in [\eta]}$ whenever they are contained in one. This can be done in $\eta \cdot O(m \log \mu)$ time. Given these separable collections, we can easily find $k^\star = \arg \min_{k \in [\eta]} |{\hat \U_{k}}|$. Similarly, given the set of $k^\star$-bad u-fans $\mathcal B_{k^\star}$, we can construct the separable collections $\{{\U_{i,i'}} \setminus \mathcal B_{k^\star}\}_{1 \leq i < i' \leq \eta}$ in $O(m \eta^2 \log \mu)$ time and then find the values $i^\star ,i'^\star = \arg \max_{1 \leq i < i' \leq \eta} |{\U_{i,i'}} \setminus \mathcal B_{k^\star}|$.

We now show to compute the set of $k$-bad u-fans $\mathcal B_{k}$ for any $k \in [\eta]$ and implement the subroutine $\ModifyB(\mathcal V, k)$ in $O(m \eta \log \mu)$ time.

\medskip
\noindent \textbf{Finding $k$-Bad U-Fans:}
Following the proof of \Cref{claim:bad count}, which upper bounds the number of $k$-bad u-fans, we note that finding all of the non-social $k$-bad u-fans can be done by simply traversing the $k$-relevant alternating paths starting at the vertices of the social u-fans. Specifically, for each social u-fan $\f' \in \hat \U$ and each vertex $x \in \f'$, one can traverse the $k$-relevant alternating paths with the color $c_{\f'}(x)$ starting at $x$. If $c_{\f'}(x) \in \mathcal C_k$, there are $2\eta - 2$ such paths. Otherwise, there are only $2$. 
For each of these paths $P$, we can traverse the path in time $O(|P| \log \mu)$ using our data structures and find its other endpoint $y$. We can then check in $O(\log \mu)$ time using our data structures if there is a non-social u-fan $\f$ with a vertex at $y$ such that calling $\ModifyB(\{\f\},k)$ flips this path.
If so, we add this u-fan $\f$ to the set of $k$-bad u-fans. After doing this for each social u-fan, we have found all of the non-social $k$-bad u-fans in $\mathcal B_{k}$.

For each $i \in [2 \eta] \setminus \{2k, 2k - 1\}$ and $\ell \in \{2k, 2k - 1\}$, the total length of all $\{C_i(j), C_{\ell}(j)\}$-alternating paths, for all $j\in [r]$, is upper bounded by $m$, since all of these paths are edge disjoint. Thus, summing over all $i \in [2 \eta] \setminus \{2k, 2k - 1\}$ and $\ell \in \{2k, 2k - 1\}$, it follows that the total length of all $k$-relevant alternating paths is $O(m\eta)$.  
Since we traverse each $k$-relevant alternating path at most twice during this process, this entire process can be implemented in $O(m \eta \log \mu)$ time using our data structures.

\medskip
\noindent \textbf{Implementing $\ModifyB(\mathcal V, k)$:} We construct the set of $k$-relevant alternating paths $\mathcal P$ by scanning through each u-fan $\f \in \mathcal V$ and computing the $k$-relevant alternating paths corresponding to $\f$. By a similar argument as above, we conclude that the total time taken to compute and flip each of these paths is $O(\log \mu)$ times the total lengths of these paths, which is $O(m \eta)$. Thus, this can be done in $O(m \eta \log \mu)$ time.

\section{The Algorithm $\Small$ (Proof of \Cref{lem:Sinnamon})}\label{app:small}

In this section, we describe and analyze the subroutine $\Small$, proving \Cref{lem:Sinnamon}, which we restate below.

\colorsmall*

\begin{proof}
    We begin by identifying the most common type $\tau^\star$ among the u-fans in $\U$. By a simple averaging argument, we know that $\ell \geq |{\U}| / \mu^2$ of the u-fans in $\U$ have type $\tau^\star$. We can find this type in $O(m \log \mu)$ time by scanning over each u-fan in $\U$ and counting the number of occurrences of each type.
    We then activate these u-fans by flipping only alternating paths with type $\tau^\star$, extending the coloring to $\ell$ uncolored edges in $O(n \log \mu)$ time, and remove the u-fans whose type changed during this process from $\U$. We repeat this process for $O(\mu^2)$ iterations, each time extending the coloring to at least a $1/\mu^2$ proportion of the u-fans in $\U$. Thus, after $\mu^2$ iterations, we have that $|{\U}| \leq \lambda \cdot (1 - 1/\mu^2)^{\mu^2} \leq \lambda/2$.
    Since activating a u-fan can change the type of at most one other u-fan in the separable collection, it follows that we have extended the coloring to at least $\lambda/4 \geq \lambda/100$ edges. The total running time of this algorithm is $O(m \mu^2 \log \mu)$.
\end{proof}

\section{The Algorithm $\Extend$ (Proof of \Cref{lem:extend})}\label{sec:extend}

In this section, we describe and analyze the subroutine $\Amplify$, proving \Cref{lem:extend}, which we restate below.

\extendcolor*

\subsection{The Algorithm $\Extend$}
Let $G$ be a graph, $\chi : E(G) \longrightarrow C \cup \{\bot\}$ be a partial $\mu$-coloring of $G$, $\U$ a separable collection of size $\lambda$, and $\eta \geq 10$ an integer.
The algorithm $\Extend$ is recursive;
given input $(G, \chi, \U, \eta)$, the algorithm does the following:

\medskip
\noindent \textbf{Base Case:} If $\mu \leq 10\eta$, the algorithm calls $\Small(G, \chi, \U)$ to extend the coloring $\chi$ to at least $\lambda/100$ uncolored edges.

\medskip
\noindent \textbf{Recursive Case:} Otherwise, if $\mu > 10\eta$, the algorithm calls $\Amplify(G, \chi, \U, \eta)$, which (by \Cref{lem:key}) modifies $\chi$ and $\U$ (without changing which edges are colored), and constructs disjoint subsets of colors $\mathcal C_1,\dots, \mathcal C_\eta \subseteq C$ with the following properties:
\begin{enumerate}
    \item For all $k \in [\eta]$, we have that $|{\C_k}| \leq \mu / \eta$.
    \item The u-fans in $\U$ have types in $\bigcup_{k=1}^\eta (\C_k \times \C_k)$ and $|{\U}| \geq \lambda/100$.
\end{enumerate}
The algorithm then uses the coloring $\chi$ and the subsets $\{\mathcal C_k\}_{k \in [\eta]}$ to compute the following subgraphs of $G$: For each $k \in [\eta]$, let $\U_k \subseteq \U$ be the subset of u-fans in $\U$ with a type in $\mathcal C_k \times \mathcal C_k$.
We then define a set of edges $E_k := \chi^{-1}(\mathcal C_k) \cup E(\U_k)$, i.e.~$E_k$ is the set of all colored edges with colors in $\C_k$ and uncolored edges contained in the u-fans in $\U$ that have a type in $\C_k \times \C_k$,
and let $G_k := (V, E_k)$ be the subgraph of $G$ containing these edges.
We emphasize that the subgraphs $\{G_k\}_{k \in [\eta]}$ are all edge-disjoint, which enables to proceed recursively on all of them `in parallel' (without any possible conflicts).
We define $\chi_k : E_k \longrightarrow \mathcal C_k \cup \{\bot\}$ to be the coloring of the graph $G_k$ obtained by restricting the coloring $\chi$ to the edges in $E_k$. We also define $E_{\eta + 1} := E(G) \setminus (E_1 \cup \dots \cup E_{\eta})$ and $\chi_{\eta + 1} : E_{\eta + 1} \longrightarrow ( \mathcal C \setminus \bigcup_{i=1}^\eta \mathcal C_i ) \cup \{\bot\}$ to be the coloring obtained by restricting the coloring $\chi$ to the edges in $E_{\eta + 1}$.

For each $k \in [\eta]$, the algorithm calls $\Extend(G_k, \chi_k, \U_k, \eta)$, which modifies the coloring $\chi_k$.
Finally, we set $\chi$ to be the union of the colorings $\chi_1,\dots, \chi_{\eta + 1}$.

\subsection{Analysis of $\Extend$}

We prove \Cref{lem:extend} by induction on the value of $$\Depth(\mu, \eta) := \max (\lceil \log_\eta(\mu / (10 \eta)) \rceil,0) \leq (\log \mu / \log \eta) + 1.$$
More specifically, we show that the algorithm extends the coloring $\chi$ to at least $\lambda/100^{\Depth(\mu, \eta) + 1}$ edges in $O(m \eta^4 \log(\mu) \cdot (\Depth(\mu, \eta) + 1))$ time.
Note that the value $\Depth(\mu, \eta)$ is an upper bound on the depth of the recursion tree of an instance: since the recursion bottoms when $\mu \le 10\eta$, we have $\mu/\eta^d \le 10\eta$, where $d$ is the recursion depth.

We first note that, if $\mu \leq 10 \eta$, the algorithm uses $\Small$ to extend the coloring $\chi$ of $G$ to at least $\lambda/100$ uncolored edges. It follows from \Cref{lem:Sinnamon} that the time it takes to compute this coloring is $O(m \mu^2 \log \eta) = O(m \eta^2 \log \mu)$. Furthermore, this is the base case of the induction and $\Depth(\mu, \eta) = 0$. 
Thus, \Cref{lem:extend} holds whenever $\mu \leq 10 \eta$.

For the induction step, suppose that \Cref{lem:extend} holds whenever $\Depth \leq L$ for some integer $L \geq 0$ and that $\Depth(\mu, \eta) = L + 1$.
In this case, the algorithm uses the subroutine $\Amplify$ to modify the coloring $\chi$ and the separable collection $\U$. Let $\tilde \chi$ and $\tilde \U$ denote the state of the coloring $\chi$ and the separable collection $\U$ after calling $\Amplify(G, \chi, \U, \eta)$. 
The algorithm then splits the instance $(G, \tilde \chi, \tilde \U, \eta)$ into subproblems $\{(G_k, \tilde \chi_k, {\tilde \U}_k, \eta)\}_{k \in [\eta]}$ 
and recurses on each of these subproblems to modify the colorings $\{\tilde \chi_k\}_{k \in [\eta]}$. Let $\chi^\star_k$ denote the state of the coloring $\tilde \chi_k$ after calling $\Extend(G_k, \tilde \chi_k, {\tilde \U}_k, \eta)$, and let $\chi^\star_{\eta + 1}$ denote the restriction of the coloring $\tilde \chi$ to the edges that are not contained in any $G_k$.
The solution returned by the algorithm is the union of the colorings $\chi^\star_1, \dots, \chi^\star_{\eta + 1}$, which we denote by $\chi^\star$.

The following claim shows that these subproblems are all feasible.

\begin{claim}\label{claim:feasible sp}
    Each of the subproblems in $\{(G_k, \tilde \chi_k, {\tilde \U}_k, \eta)\}_{k \in [\eta]}$ is a valid instance. Moreover, the subgraphs corresponding to these $\eta$ subproblems are edge-disjoint, and thus any color (re)assignments done for one subproblem have no effect on the other subproblems.
\end{claim}

\begin{proof}
    Since $\tilde \chi_k$ is the restriction of $\tilde \chi$ to $G_k$, it follows from the definition of $G_k$ that $\tilde \chi_k$ is a coloring of $G_k$ with colors $\mathcal C_k$. Since $\tilde \U_k \subseteq \tilde \U$ is a subset of a separable collection, it is a separable collection itself. Furthermore, the edges of each u-fan in $\tilde \U_k$ are contained in $G_k$, and for each u-fan $\f \in \tilde \U_k$, we have that $\tau(\f) \subseteq \C_k$. It follows immediately from the definition of the subgraphs $G_k$ that they are edge-disjoint.
\end{proof}

Since $|\mathcal C_k| \leq \mu / \eta$ for each $k \in [\eta]$ and $\Depth(\mu / \eta, \eta) \leq \Depth(\mu, \eta) - 1 = L$, it follows by \Cref{claim:feasible sp} and the induction hypothesis that, for each $k \in [\eta]$, the coloring $\chi^\star_k$ has at least $|\tilde \U_k| \cdot 100^{-L}$ more colored edges than $\tilde \chi_k$.
By \Cref{lem:key}, we know that $\tilde \U = \bigcup_{k=1}^\eta \tilde \U_k$.
Since the domains of the colorings $\{\chi^\star_k\}_{k \in [\eta + 1]}$ partition the edge set of $G$ and the colors used by these colorings partition the set of colors $C$ used by $\chi$, it follows that the union of these colorings $\chi^\star$ is a partial coloring that assigns colors to 
at least
$$\sum_{k=1}^\eta |{\tilde \U}_k| \cdot  100^{-L} = |\tilde \U| \cdot  100^{-L} \geq |{\U}| \cdot 100^{-(L+1)}$$
more edges than $\chi$, where the second inequality follows from the fact that we have $|{\tilde\U}| \geq |{\U}|/100$ by \Cref{lem:key}. Thus, the coloring $\chi^\star$ has the required properties.

To bound the running time of the algorithm, we can observe that, for each $\ell$, the total running time at the $\ell^{th}$
level of the recursion tree is $O(m \eta^4 \log \mu)$.
To see why this is the case, note that the subgraphs in each branch of the $\ell^{th}$ level of the recursion tree are edge-disjoint by \Cref{claim:feasible sp}, and further note that the running time at any level other than the bottom level of the recursion is dominated by the time taken to handle the calls to $\Amplify$. Given any such subgraph with $m'$ edges, the time taken to handle the corresponding call to $\Amplify$ is $O(m'\eta^4 \log \mu)$ by \Cref{lem:key}. Due to the edge-disjointness of these subgraphs, summing over all these subgraphs yields a total running time of $O(m \eta^4 \log \mu)$ at this level. Similarly, the running time at the bottom level of the recursion is dominated by the time taken to handle calls to $\Small$, which by \Cref{lem:Sinnamon} gives rise to a runtime of $O (m \eta^2 \log \mu)$.
Since the depth of the recursion tree is $\Depth(\mu, \eta) = O(\log \mu / \log \eta)$, the desired running time follows.

\section{Implementation and Data Structures}\label{sec:data structs}

In this section, we describe the key data structures that we use to implement an edge coloring $\chi$ and a separable collection $\U$, allowing us to efficiently implement the operations performed by our algorithms.
Our data structures are almost identical to the data structures used by \cite{ABBC2025}, with the modification that we use balanced binary trees instead of hashmaps to make our data structures deterministic.

We begin by describing the data structures and then show how they can be used to efficiently implement the queries described in \Cref{sec:data struc overview}.

\medskip
\noindent \textbf{Preprocessing the Graph:} We can assume that, whenever we deal with a graph $G = (V,E)$ with $n$ vertices and $m$ edges, we have that $V = \{1,\dots, n\}$, i.e.~the vertices are the integers from $1$ to $n$. We can achieve this by having an initial preprocessing phase where we sort the vertices in our graph $G$ into a list $u_1 ,\dots,u_n$ and replace each occurrence of $u_i$ with $i$. This can easily be done in time $O(m \log n)$. Since the running time of our final algorithm (\Cref{thm:main:1}) is $\Omega(m \log n)$, the overhead of this preprocessing is negligible.

\medskip
\noindent \textbf{Implementing an Edge Coloring:}
Let $G = (V, E)$ be a graph of maximum degree $\Delta$ and let $C$ be a set of $\mu$ colors. For each $u \in V$, we let $N(u)$ denote the edges in $E$ incident on $u$.
Let $C = \{c_1,\dots c_\mu\}$ such that $c_1 \leq \dots \leq c_\mu$ (recall that the colors used by our algorithm are always integers). Then we define $C[j] := \{c_{j'} \mid c_{j'} \leq j\}$.
We implement an edge coloring $\chi : E \longrightarrow C$ of $G$ using the following, for each $u \in V$:
\begin{itemize}
    \item The map $\phi_u : N(u) \longrightarrow C$ where $\phi_u(e) := \chi(e)$ for all $e \in N(u)$.
    \item The map $\phi_u' : C \longrightarrow N(u)$ where $\phi_u'(c) := \{e \in N(u) \mid \chi(e) = c\}$.
    \item The set $\miss_\chi(u) \cap C[\deg_G(u) + 1]$.\footnote{We take this intersection with $C[\deg_G(u)+1]$ instead of maintaining $\miss_\chi(u)$ directly to ensure that the space complexity and initialization time of the data structures are $\tilde O(m)$ and not $\Omega(\Delta n)$.}
\end{itemize}
We implement all of the maps and sets using balanced binary trees, allowing us to perform membership queries, insertions and deletions in $O(\log \mu)$ time, since each of these maps and sets have size at most $\mu$. Furthermore, we store an array of pointers to these binary trees for each $u \in V$, allowing us to retrieve any of these trees in $O(1)$ given the corresponding vertex $u$.

The map $\phi'$ allows us to check if a color $c \in C$ is available at a vertex $u \in V$ in $O(\log \mu)$ time, and if it is not, to find the edge $e \in N(u)$ with $\chi(e) = c$.
Each time an edge $e$ changes color under $\chi$, we can easily update all of these data structures in $O(\log \mu)$ time. Furthermore, given $O(\log \mu)$ time query access to an edge coloring $\chi$, we can initialize these data structures in $O(m \log \mu)$ time.
We note that $\chi$ is a proper edge coloring if and only if $|\phi'_u(c)| \leq 1$ for all $u \in V$ and $c \in C$.

Since the binary trees used to implement the maps $\phi_u$ stores $m$ elements, it follows that it can be implemented with $O(m)$ space. Similarly, the maps $\phi_u'$ store $2m$ elements in total (if $\{e \in N(u) \mid \chi(e) = c\} = \emptyset$, then we do not store anything for $\phi'_u(c)$) and thus can be implemented with $O(m)$ space since each element has size $O(1)$ (recall that $|\phi'_u(c)| \leq 1$ since the coloring is proper). 
Since $\sum_u |\miss_\chi(u) \cap C[\deg_G(u) + 1]| = O(m)$, the sets $\miss_\chi(u) \cap C[\deg_G(u) + 1]$ can be implemented in $O(m)$ space.

\medskip
\noindent \textbf{Implementing a Separable Collection:}
We implement a separable collection $\U$ in a similar manner using the following, for each $u \in V$:
\begin{itemize}
    \item The map $\psi_u : C \longrightarrow \U$ where $\psi_u(c) := \{\f \in \U \mid u \in \f, c_{\f}(u) = c\}$.
    \item The set $C_{\U}(u) := \{c_{\f}(u) \mid \f \in \U, u \in \f\}$.
    \item The set $\overline{C}_{\U}(u) := \left(\miss_\chi(u) \cap C[\deg_G(u) + 1] \right) \setminus C_{\U}(u)$.
\end{itemize}
We again implement all of the maps and sets using balanced binary trees, allowing us to access and change entries in $O(\log \mu)$ time.
We note that, since $\U$ is separable, $|\psi_u(c)| \leq 1$ for all $u \in V$ and $c \in C$. Each time we remove a color $c \in C$ from the palette $\miss_\chi(u)$ of a vertex $u \in V$, we can update $\overline{C}_{\U}(u)$ in $O(\log \mu)$ time and check $\psi_u(c)$ in $O(\log \mu)$ time to find any u-component that has been damaged. Each time we add or remove a u-fan from $\U$, we can update all of these data structures in $O(\log \mu)$ time. 
Furthermore, we can initialize these data structures for an empty collection in $O(m \log \mu)$ time by creating a empty maps $\psi_u$ and empty sets $C_{\U}(u)$ for each $u \in V$ and copying the sets $\overline{C}_{\U}(u) = \miss_\chi(u) \cap C[\deg_G(u) + 1]$ for each $u \in V$ which are maintained by the data structures for the edge coloring $\chi$.
Since $\U$ is separable, we can see that $\overline{C}_{\U}(u) \neq \emptyset$. Thus, whenever we want a color from the set $\miss_\chi(u) \setminus C_{\U}(u)$, it suffices to take an arbitrary color from the set $\overline{C}_{\U}(u)$.

For each $\f \in \U$, we can see that $\f$ is contained at most $3$ different $\psi_u$. Thus, the total space required to store the binary trees that implement the $\psi_u$ is $O(|\U|)$. Since $\U$ is separable, the u-fans in $\U$ are edge-disjoint, and thus $|\mathcal U| \leq m$. It follows that the maps $\psi_u$ can be stored with $O(m)$ space. For each $u \in V$, we can observe that $|C_{\U}(u)| \leq \deg_G(u)$ since at most $\deg_G(u)$ many u-fans in $\U$ contain the vertex $u$, and $|\overline{C}_{\U}(u)| \leq \deg_G(u) + 1$ since $\overline{C}_{\U}(u) \subseteq C[\deg_G(u) + 1]$. Thus, the total space required to store the sets $\{C_{\U}(u)\}_{u \in V}$ and $\{\overline{C}_{\U}(u)\}_{u \in V}$ is $O(m)$.

\subsection{Implementing the Operations from \Cref{sec:data struc overview}}\label{sec:implementing queries}

We now describe how to implement each of the operations from \Cref{sec:data struc overview}.

\medskip
\noindent \textbf{Implementing} $\textsc{Initialize}(G, \chi)$: Suppose that we are given the graph $G$ and $O(\log \mu)$ time query access to an edge coloring $\chi$ of $G$. We can initialize the data structures used to maintain the maps $\phi_u$ and $\phi_u'$ in $O(m \log \mu)$ time.
We can then scan through the vertices $u \in V$ and initialize the sets $\miss_\chi(u) \cap C[\deg_G(u) + 1]$ in $O(m \log\mu)$ time. We can then initialize the data structures for an empty separable collection in $O(m\log \mu)$ time by creating empty maps $\psi_u$ and initializing the sets $C_{\U}(u) \leftarrow \emptyset$ and $\overline{C}_{\U}(u) \leftarrow \miss_\chi(u) \cap C[\deg_G(u) + 1]$ for each $u \in V$.

\medskip
\noindent \textbf{Implementing} $\textsc{Insert}_{\U}(\f)$: By performing at most $3$ queries to the maps $\psi_u$, we can check if $\U \cup \{\f\}$ is separable. If so we can update the $\psi_u$ and the sets $C_{\U}(x)$ and $\overline{C}_{\U}(x)$ for $x \in \f$ in $O(\log \mu)$ time in order to insert $\f$ into $\U$. Otherwise, we return $\texttt{fail}$.

\medskip
\noindent \textbf{Implementing} $\textsc{Delete}_{\U}(\f)$: We can first make a query to the $\psi_u$ to ensure that $\f \in \U$. If so, we can update the $\psi_u$ and the sets $C_{\U}(x)$ and $\overline{C}_{\U}(x)$ for $x \in \f$ in $O(\log \mu)$ time to remove $\f$ from $\U$.

\medskip
\noindent \textbf{Implementing} $\textsc{Find-U-Fan}_{\U}(x,c)$: We make a query to $\psi_x$ and check if there is an element $\psi_x(c)$. If no such element is contained in $\psi_x$, then return $\texttt{fail}$. Otherwise, return the unique u-fan in the set $\psi_x(c)$. This takes $O(\log \mu)$ time.

\medskip
\noindent \textbf{Implementing} $\textsc{Missing-Color}_{\U}(x)$: Return an arbitrary color from the set $\overline{C}_{\U}(x)$.

\section*{Acknowledgements} 
Sepehr Assadi is supported in part by a Sloan Research Fellowship, an NSERC Discovery Grant (RGPIN-2024-04290), and a Faculty of Math Research Chair grant from University of Waterloo.
Soheil Behnezhad is funded by an NSF CAREER award CCF-2442812 and a Google Faculty Research Award. Mart\'in Costa is supported by a Google PhD Fellowship. Shay Solomon is funded by the European Union (ERC, DynOpt, 101043159). Views and opinions expressed are however those of the author(s) only and do not necessarily reflect those of the European Union or the European Research Council. Neither the European Union nor the granting authority can be held responsible for them. Shay Solomon is also funded by a grant from the United States-Israel Binational Science Foundation (BSF), Jerusalem, Israel, and the United States National Science Foundation (NSF). Work of Tianyi Zhang was done while at ETH Z\"urich when supported by funding from the starting grant ``A New Paradigm for Flow and Cut Algorithms'' (no. TMSGI2\_218022) of the Swiss National Science Foundation.

\bibliography{bibliography.bib}

\newcommand{\etalchar}[1]{$^{#1}$}
\begin{thebibliography}{GKMU18}

\bibitem[ABB{\etalchar{+}}25]{ABBC2025}
Sepehr Assadi, Soheil Behnezhad, Sayan Bhattacharya, Mart\'in Costa, Shay Solomon, and Tianyi Zhang.
\newblock Vizing's theorem in near-linear time.
\newblock In {\em 57th Annual {ACM} {SIGACT} Symposium on Theory of Computing (STOC)}, 2025.

\bibitem[ALG22]{abdolazimi2022matrix}
Dorna Abdolazimi, Kuikui Liu, and Shayan~Oveis Gharan.
\newblock A matrix trickle-down theorem on simplicial complexes and applications to sampling colorings.
\newblock In {\em 2021 IEEE 62nd Annual Symposium on Foundations of Computer Science (FOCS)}, pages 161--172. IEEE, 2022.

\bibitem[Alo03]{alon2003simple}
Noga Alon.
\newblock {A simple algorithm for edge-coloring bipartite multigraphs}.
\newblock {\em Information Processing Letters}, 85(6):301--302, 2003.

\bibitem[Arj82]{arjomandi1982efficient}
Eshrat Arjomandi.
\newblock {An efficient algorithm for colouring the edges of a graph with $\Delta+1$ colours}.
\newblock {\em INFOR: Information Systems and Operational Research}, 20(2):82--101, 1982.

\bibitem[Ass25]{Assadi24}
Sepehr Assadi.
\newblock {Faster Vizing and Near-Vizing Edge Coloring Algorithms}.
\newblock In {\em Annual ACM-SIAM Symposium on Discrete Algorithms (SODA)}, 2025.

\bibitem[BBKO22]{balliu2022distributed}
Alkida Balliu, Sebastian Brandt, Fabian Kuhn, and Dennis Olivetti.
\newblock {Distributed edge coloring in time polylogarithmic in $\Delta$}.
\newblock In {\em Proceedings of the 2022 ACM Symposium on Principles of Distributed Computing}, pages 15--25, 2022.

\bibitem[BCC{\etalchar{+}}24]{BhattacharyaCCSZ24}
Sayan Bhattacharya, Din Carmon, Mart{\'\i}n Costa, Shay Solomon, and Tianyi Zhang.
\newblock {Faster $(\Delta+1)$-Edge Coloring: Breaking the $m\sqrt{n}$ Time Barrier}.
\newblock In {\em 65th IEEE Symposium on Foundations of Computer Science (FOCS)}, 2024.

\bibitem[BCHN18]{BhattacharyaCHN18}
Sayan Bhattacharya, Deeparnab Chakrabarty, Monika Henzinger, and Danupon Nanongkai.
\newblock {Dynamic Algorithms for Graph Coloring}.
\newblock In {\em Proceedings of the Twenty-Ninth Annual {ACM-SIAM} Symposium on Discrete Algorithms (SODA)}, pages 1--20. {SIAM}, 2018.

\bibitem[BCPS24a]{BhattacharyaCPS24c}
Sayan Bhattacharya, Mart{\'{\i}}n Costa, Nadav Panski, and Shay Solomon.
\newblock {Arboricity-Dependent Algorithms for Edge Coloring}.
\newblock In {\em 19th Scandinavian Symposium and Workshops on Algorithm Theory (SWAT)}, volume 294 of {\em LIPIcs}, pages 12:1--12:15, 2024.

\bibitem[BCPS24b]{BhattacharyaCPS24b}
Sayan Bhattacharya, Mart{\'{\i}}n Costa, Nadav Panski, and Shay Solomon.
\newblock {Density-Sensitive Algorithms for ({\(\Delta\)} + 1)-Edge Coloring}.
\newblock In {\em 32nd Annual European Symposium on Algorithms, {ESA} 2024}, volume 308 of {\em LIPIcs}, pages 23:1--23:18, 2024.

\bibitem[BCPS24c]{BhattacharyaCPS24}
Sayan Bhattacharya, Mart\'in Costa, Nadav Panski, and Shay Solomon.
\newblock {Nibbling at Long Cycles: Dynamic (and Static) Edge Coloring in Optimal Time}.
\newblock In {\em Proceedings of the {ACM-SIAM} Symposium on Discrete Algorithms (SODA)}. {SIAM}, 2024.

\bibitem[BCSZ25]{BhattacharyaCSZ24}
Sayan Bhattacharya, Mart{\'{\i}}n Costa, Shay Solomon, and Tianyi Zhang.
\newblock {Even Faster ({\(\Delta\)} + 1)-Edge Coloring via Shorter Multi-Step Vizing Chains}.
\newblock In {\em Proceedings of the 2025 Annual {ACM-SIAM} Symposium on Discrete Algorithms, {SODA} 2025}, pages 4914--4947. {SIAM}, 2025.

\bibitem[BD24]{bernshteyn2024linear}
Anton Bernshteyn and Abhishek Dhawan.
\newblock {A linear-time algorithm for $(1+\epsilon)\Delta$-edge-coloring}.
\newblock {\em arXiv preprint arXiv:2407.04887}, 2024.

\bibitem[BDH{\etalchar{+}}19]{BehnezhadDHKS19}
Soheil Behnezhad, Mahsa Derakhshan, MohammadTaghi Hajiaghayi, Marina Knittel, and Hamed Saleh.
\newblock Streaming and massively parallel algorithms for edge coloring.
\newblock In {\em 27th Annual European Symposium on Algorithms (ESA)}, volume 144 of {\em LIPIcs}, pages 15:1--15:14, 2019.

\bibitem[Ber22]{Bernshteyn22}
Anton Bernshteyn.
\newblock {A fast distributed algorithm for $(\Delta+1)$-edge-coloring}.
\newblock {\em J. Comb. Theory, Ser. {B}}, 152:319--352, 2022.

\bibitem[BGW21]{BhattacharyaGW21}
Sayan Bhattacharya, Fabrizio Grandoni, and David Wajc.
\newblock Online edge coloring algorithms via the nibble method.
\newblock In {\em Proceedings of the{ACM-SIAM} Symposium on Discrete Algorithms (SODA)}, pages 2830--2842. {SIAM}, 2021.

\bibitem[BM17]{BarenboimM17}
Leonid Barenboim and Tzalik Maimon.
\newblock Fully-dynamic graph algorithms with sublinear time inspired by distributed computing.
\newblock In {\em International Conference on Computational Science (ICCS)}, volume 108 of {\em Procedia Computer Science}, pages 89--98. Elsevier, 2017.

\bibitem[BSVW24]{BilkstadSVW24}
Joakim Blikstad, Ola Svensson, Radu Vintan, and David Wajc.
\newblock Online edge coloring is (nearly) as easy as offline.
\newblock In {\em Proceedings of the Annual {ACM} Symposium on Theory of Computing (STOC)}. {ACM}, 2024.

\bibitem[BSVW25]{BlikstadOnline2025}
Joakim Blikstad, Ola Svensson, Radu Vintan, and David Wajc.
\newblock {Deterministic Online Bipartite Edge Coloring}.
\newblock In {\em Annual ACM-SIAM Symposium on Discrete Algorithms (SODA)}, 2025.

\bibitem[CH82]{cole1982edge}
Richard Cole and John Hopcroft.
\newblock On edge coloring bipartite graphs.
\newblock {\em SIAM Journal on Computing}, 11(3):540--546, 1982.

\bibitem[CHL{\etalchar{+}}20]{ChangHLPU20}
Yi{-}Jun Chang, Qizheng He, Wenzheng Li, Seth Pettie, and Jara Uitto.
\newblock {Distributed Edge Coloring and a Special Case of the Constructive Lov{\'{a}}sz Local Lemma}.
\newblock {\em {ACM} Trans. Algorithms}, 16(1):8:1--8:51, 2020.

\bibitem[Chr23]{Christiansen23}
Aleksander Bj{\o}rn~Grodt Christiansen.
\newblock {The Power of Multi-step Vizing Chains}.
\newblock In {\em Proceedings of the 55th Annual {ACM} Symposium on Theory of Computing (STOC)}, pages 1013--1026. {ACM}, 2023.

\bibitem[Chr24]{Christiansen24}
Aleksander B.~G. Christiansen.
\newblock Deterministic dynamic edge-colouring.
\newblock {\em CoRR}, abs/2402.13139, 2024.

\bibitem[CK08]{cole2008new}
Richard Cole and {\L}ukasz Kowalik.
\newblock New linear-time algorithms for edge-coloring planar graphs.
\newblock {\em Algorithmica}, 50(3):351--368, 2008.

\bibitem[CMZ24]{chechik2023streaming}
Shiri Chechik, Doron Mukhtar, and Tianyi Zhang.
\newblock Streaming edge coloring with subquadratic palette size.
\newblock In {\em 51st International Colloquium on Automata, Languages, and Programming, {ICALP} 2024, July 8-12, 2024, Tallinn, Estonia}, volume 297 of {\em LIPIcs}, pages 40:1--40:12. Schloss Dagstuhl - Leibniz-Zentrum f{\"{u}}r Informatik, 2024.

\bibitem[CN90]{chrobak1990improved}
Marek Chrobak and Takao Nishizeki.
\newblock Improved edge-coloring algorithms for planar graphs.
\newblock {\em Journal of Algorithms}, 11(1):102--116, 1990.

\bibitem[COS01]{combinatorica/ColeOS01}
Richard Cole, Kirstin Ost, and Stefan Schirra.
\newblock {Edge-Coloring Bipartite Multigraphs in $O(E \log D)$ Time}.
\newblock {\em Comb.}, 21(1):5--12, 2001.

\bibitem[CPW19]{CohenPW19}
Ilan~Reuven Cohen, Binghui Peng, and David Wajc.
\newblock Tight bounds for online edge coloring.
\newblock In {\em 60th {IEEE} Annual Symposium on Foundations of Computer Science (FOCS)}, pages 1--25. {IEEE} Computer Society, 2019.

\bibitem[CRV24]{ChristiansenRV24}
Aleksander B.~G. Christiansen, Eva Rotenberg, and Juliette Vlieghe.
\newblock Sparsity-parameterised dynamic edge colouring.
\newblock In {\em 19th Scandinavian Symposium and Workshops on Algorithm Theory (SWAT)}, volume 294 of {\em LIPIcs}, pages 20:1--20:18, 2024.

\bibitem[CWZZ25]{chen2025decay}
Zejia Chen, Yulin Wang, Chihao Zhang, and Zihan Zhang.
\newblock {Decay of correlation for edge colorings when $q> 3\Delta$}.
\newblock {\em arXiv preprint arXiv:2502.06586}, 2025.

\bibitem[CY89]{chrobak1989fast}
Marek Chrobak and Moti Yung.
\newblock Fast algorithms for edge-coloring planar graphs.
\newblock {\em Journal of Algorithms}, 10(1):35--51, 1989.

\bibitem[Dav23]{Davies23}
Peter Davies.
\newblock Improved distributed algorithms for the lov{\'{a}}sz local lemma and edge coloring.
\newblock In {\em Proceedings of the {ACM-SIAM} Symposium on Discrete Algorithms (SODA)}, pages 4273--4295. {SIAM}, 2023.

\bibitem[DGS25]{dudeja2024randomizedgreedyonlineedge}
Aditi Dudeja, Rashmika Goswami, and Michael Saks.
\newblock {Randomized Greedy Online Edge Coloring Succeeds for Dense and Randomly-Ordered Graphs}.
\newblock In {\em Annual ACM-SIAM Symposium on Discrete Algorithms (SODA)}, 2025.

\bibitem[Dha24]{dhawan2024simple}
Abhishek Dhawan.
\newblock A simple algorithm for near-vizing edge-coloring in near-linear time.
\newblock {\em arXiv preprint arXiv:2407.16585}, 2024.

\bibitem[DHZ19]{duan2019dynamic}
Ran Duan, Haoqing He, and Tianyi Zhang.
\newblock Dynamic edge coloring with improved approximation.
\newblock In {\em 30th Annual ACM-SIAM Symposium on Discrete Algorithms (SODA)}, 2019.

\bibitem[EK24]{elkin2024deterministic}
Michael Elkin and Ariel Khuzman.
\newblock {Deterministic Simple $(1+\epsilon)$-Edge-Coloring in Near-Linear Time}.
\newblock {\em arXiv preprint arXiv:2401.10538}, 2024.

\bibitem[EPS14]{elkin20142delta}
Michael Elkin, Seth Pettie, and Hsin-Hao Su.
\newblock {$(2\Delta - 1)$-Edge-Coloring is Much Easier than Maximal Matching in the Distributed Setting}.
\newblock In {\em Proceedings of the Twenty-Sixth Annual ACM-SIAM Symposium on Discrete Algorithms}, pages 355--370. SIAM, 2014.

\bibitem[FGK17]{fischer2017deterministic}
Manuela Fischer, Mohsen Ghaffari, and Fabian Kuhn.
\newblock Deterministic distributed edge-coloring via hypergraph maximal matching.
\newblock In {\em 2017 IEEE 58th Annual Symposium on Foundations of Computer Science (FOCS)}, pages 180--191. IEEE, 2017.

\bibitem[GKK10]{goel2010perfect}
Ashish Goel, Michael Kapralov, and Sanjeev Khanna.
\newblock Perfect matchings in $o(n\log n)$ time in regular bipartite graphs.
\newblock In {\em Proceedings of the Forty-second ACM Symposium on Theory of Computing}, pages 39--46, 2010.

\bibitem[GKMU18]{ghaffari2018deterministic}
Mohsen Ghaffari, Fabian Kuhn, Yannic Maus, and Jara Uitto.
\newblock Deterministic distributed edge-coloring with fewer colors.
\newblock In {\em Proceedings of the 50th Annual ACM SIGACT Symposium on Theory of Computing}, pages 418--430, 2018.

\bibitem[GNK{\etalchar{+}}85]{gabow1985algorithms}
Harold~N Gabow, Takao Nishizeki, Oded Kariv, Daneil Leven, and Osamu Terada.
\newblock Algorithms for edge coloring.
\newblock {\em Technical Report}, 1985.

\bibitem[GP20]{GrebikP20}
Jan Grebik and Oleg Pikhurko.
\newblock Measurable versions of vizing’s theorem.
\newblock {\em Advances in Mathematics}, 374(107378), 2020.

\bibitem[GS24]{ghosh2023low}
Prantar Ghosh and Manuel Stoeckl.
\newblock Low-memory algorithms for online edge coloring.
\newblock In {\em 51st International Colloquium on Automata, Languages, and Programming, {ICALP} 2024, July 8-12, 2024, Tallinn, Estonia}, volume 297 of {\em LIPIcs}, pages 71:1--71:19. Schloss Dagstuhl - Leibniz-Zentrum f{\"{u}}r Informatik, 2024.

\bibitem[Hol81]{holyer1981np}
Ian Holyer.
\newblock The np-completeness of edge-coloring.
\newblock {\em SIAM Journal on computing}, 10(4):718--720, 1981.

\bibitem[KLS{\etalchar{+}}22]{KulkarniLSST22}
Janardhan Kulkarni, Yang~P. Liu, Ashwin Sah, Mehtaab Sawhney, and Jakub Tarnawski.
\newblock Online edge coloring via tree recurrences and correlation decay.
\newblock In {\em 54th Annual {ACM} {SIGACT} Symposium on Theory of Computing (STOC)}, pages 104--116. {ACM}, 2022.

\bibitem[Kow24]{Kowalik24}
Lukasz Kowalik.
\newblock {Edge-Coloring Sparse Graphs with {\(\Delta\)} Colors in Quasilinear Time}.
\newblock In {\em 32nd Annual European Symposium on Algorithms, {ESA} 2024}, volume 308 of {\em LIPIcs}, pages 81:1--81:17, 2024.

\bibitem[KS87]{karloff1987efficient}
Howard~J Karloff and David~B Shmoys.
\newblock Efficient parallel algorithms for edge coloring problems.
\newblock {\em Journal of Algorithms}, 8(1):39--52, 1987.

\bibitem[PR01]{panconesi2001some}
Alessandro Panconesi and Romeo Rizzi.
\newblock Some simple distributed algorithms for sparse networks.
\newblock {\em Distributed computing}, 14(2):97--100, 2001.

\bibitem[SB24]{behnezhad2023streaming}
Mohammad Saneian and Soheil Behnezhad.
\newblock Streaming edge coloring with asymptotically optimal colors.
\newblock In {\em 51st International Colloquium on Automata, Languages, and Programming, {ICALP} 2024, July 8-12, 2024, Tallinn, Estonia}, volume 297 of {\em LIPIcs}, pages 121:1--121:20. Schloss Dagstuhl - Leibniz-Zentrum f{\"{u}}r Informatik, 2024.

\bibitem[Sin19]{sinnamon2019fast}
Corwin Sinnamon.
\newblock Fast and simple edge-coloring algorithms.
\newblock {\em arXiv preprint arXiv:1907.03201}, 2019.

\bibitem[SW21]{SaberiW21}
Amin Saberi and David Wajc.
\newblock The greedy algorithm is not optimal for on-line edge coloring.
\newblock In {\em 48th International Colloquium on Automata, Languages, and Programming (ICALP)}, volume 198 of {\em LIPIcs}, pages 109:1--109:18, 2021.

\bibitem[Viz64]{Vizing}
V.~G. Vizing.
\newblock On an estimate of the chromatic class of a p-graph.
\newblock {\em Discret Analiz}, 3:25--30, 1964.

\bibitem[WZZ24]{wang2024sampling}
Yulin Wang, Chihao Zhang, and Zihan Zhang.
\newblock {Sampling Proper Colorings on Line Graphs Using $(1+o(1))\Delta$ Colors}.
\newblock In {\em Proceedings of the 56th Annual ACM Symposium on Theory of Computing}, pages 1688--1699, 2024.

\end{thebibliography}
\bibliographystyle{alpha}

\end{document}